\documentclass[11pt,letter]{article}
\usepackage{latexsym,bm}
\usepackage{amsmath}
\usepackage{amsthm}
\usepackage{booktabs}
\usepackage{epsfig}
\usepackage{graphicx}
\usepackage{amsfonts}
\usepackage{color}
\usepackage{amssymb}
\usepackage{multirow}
\usepackage{algorithm}
\usepackage{comment}
\usepackage{algpseudocode}
\usepackage{indentfirst}
\usepackage[top=1in, bottom=1in, left=1.1in, right=1.1in]{geometry}

\usepackage{caption}
\usepackage{subcaption}

\usepackage{natbib}
\newcommand{\indep}{\rotatebox[origin=c]{90}{$\models$}}

\newtheorem{theorem}{Theorem}
\newtheorem{example}{Example}
\newtheorem{remark}{Remark}
\newtheorem{lemma}{Lemma}

\newtheorem{proposition}{Proposition}
\newtheorem{assumption}{Assumption}
\newtheorem{algo}{Algorithm}

\def\T{{ \mathrm{\scriptscriptstyle T} }}

\begin{document}
	\title{Covariate Adaptive Family-wise Error Rate Control for Genome-Wide Association Studies}

	\author{Huijuan Zhou\thanks{Institute of Statistics and Big Data, Renmin University of China, Beijing 100872, China.} \footnotemark[2],  Xianyang Zhang\thanks{Department of Statistics, Texas A\&M University, College Station, Texas 77843, U.S.A.} and Jun Chen \thanks{Division of Biomedical Statistics and Informatics, Mayo Clinic, Rochester, Minnesota 55905, U.S.A.}\\
	(huijuan@stat.tamu.edu, zhangxiany@stat.tamu.edu, Chen.Jun2@mayo.edu)}
\date{}

\maketitle

	\begin{abstract}
	The family-wise error rate (FWER) has been widely used in genome-wide association studies. With the increasing availability of functional genomics data, it is possible to increase the detection power by leveraging  these genomic functional annotations. Previous efforts to accommodate covariates in multiple testing focus on the false discovery rate control while covariate-adaptive FWER-controlling procedures remain under-developed. Here we propose a novel covariate-adaptive FWER-controlling procedure that incorporates external covariates which are potentially informative of either the statistical power or the prior null probability. An efficient algorithm is developed to implement the proposed method. We prove its asymptotic validity and obtain the rate of convergence through a perturbation-type argument. Our numerical studies show that the new procedure is more powerful than competing methods and maintains robustness across different settings. We apply the proposed approach to the UK Biobank data and analyze 27 traits with 9 million single-nucleotide polymorphisms tested for associations. Seventy-five genomic annotations are used as covariates. Our approach detects more genome-wide significant loci than other methods in 21 out of the 27 traits. 
\end{abstract}

\section{Introduction}

Multiple testing arises when we face a large number of hypotheses and aim to discover signals while controlling specific error measures. The family-wise error rate (FWER) and the false discovery rate (FDR) are two commonly used error measures employed in a wide range of scientific studies. The FWER is the probability of making one false discovery, while the FDR is the expected proportion of false positives. The FWER provides stringent control of Type I errors and is preferable if (i) the overall conclusion from various individual inferences is likely to be erroneous when at least one of them is, or (ii) the existence of a single false claim would cause significant loss. In contrast, the FDR control procedure provides less stringent control of Type I errors, and it generally delivers higher power at the cost of an increased number of Type I errors. 

\begin{table}\centering
	\caption{Possible outcomes when testing multiple hypotheses}
	\begin{tabular}{|c|ccc|}
		\hline
		& Not rejected &Rejected  &Total  \\
		\hline
		True nulls&$U$  &$V$  &$m_0$  \\
		True alternatives&$T$  &$S$  &$m_1$  \\
		Total&  $m-R$&$R$  &$m$  \\
		\hline
	\end{tabular}
	\label{table-mt-out}
\end{table}

Consider the problem of simultaneously testing $m$ hypotheses. We reject the hypotheses whose p-values are less than a cutoff $t^*$. For many FWER and FDR controlling procedures, the $t^*$ that controls either one of them at level $\alpha$ is obtained by solving the following constraint optimization problem
\begin{align}\label{eq-main}
\text{maximize}_{t\in [0,1]}R(t)\quad \text{ s.t. } \quad M(t)\le\alpha,
\end{align}
where $R(t)$ denotes the total number of rejections given the threshold $t$ and $M(t)$ is a (conservative) estimate of the FWER or FDR. The most fundamental procedure for controlling the FWER is the Bonferroni method. It corresponds to the choice of $M(t)=mt$, which is the union bound on the FWER under the assumption that the null p-values are uniformly distributed (or super-uniform) on $[0,1]$. The classical Benjamini--Hochberg procedure for controlling the FDR can also be formulated using (\ref{eq-main}) with $M(t)=mt/R(t)$ being a conservative estimate of the FDR \citep{BH:1995}.

Formulation (\ref{eq-main}) assumes that the hypotheses for different features are exchangeable. However, in many scientific applications, there are informative covariates for each hypothesis that could reflect the group structure among the hypotheses or provide information on prior null probabilities. For example, in genome-wide association studies (GWAS), single-nucleotide polymorphisms (SNP) in active chromatin state are more likely to be significantly associated with the phenotype \citep{Consortium:2017}. In a meta-analysis where samples are pooled across studies, the loci-specific sample sizes and population-level frequency can be informative for association analyses \citep{Boca:2018}. For a fixed sample size, the power to detect significant associations is determined by the effect size, minor allele frequency, and levels of linkage disequilibrium at causal and non-causal variants \citep{Kichaev:2019}. It is thus promising to incorporate these covariates to improve the detection power in GWAS.

Multiple testing procedures that leverage different types of covariates (prior) information have received considerable attention in the literature especially for the false discovery rate control.
\citet{Genovese:2006} pioneered multiple testing procedures with prior information using weighted p-values and demonstrated that their weighted procedure controls the FWER and FDR while improving power. \citet{Roeder:2009} further explored their p-value weighting procedure by introducing an optimal weighting scheme for the FWER control.
Inspired by the above works, \citet{Hu:2010} developed a group Benjamini--Hochberg procedure by estimating the proportions of null hypotheses for each group separately. \citet{Bourgon:2010} developed a particular weighting method called independent filtering, which first filters hypotheses by a criterion independent of the p-values and only tests hypotheses passing the filter. \citet{Ignatiadis:2016} proposed the independent hypothesis weighting for multiple testing with covariate information. The idea is to bin the covariate into several groups and then apply the weighted Benjamini--Hochberg procedure with piecewise constant weights. A similar idea has been used in the structure-adaptive Benjamini--Hochberg algorithm introduced in \citet{Li:2019}, where the weight assigned for each p-value is  the reciprocal of the estimated null probability of the corresponding hypothesis. The null probabilities were estimated by utilizing censored p-values and structural information believed to be present among the hypotheses. \citet{Boca:2018} employed a similar approach by using the censored p-values and a regression approach to estimate null probabilities based on informative covariates. The above procedures
can all be viewed to some extent as different variants of the weighted Benjamini--Hochberg or Bonferroni procedure. On the other hand, there are FDR-controlling procedures designed to find an optimal decision threshold by taking into account the p-value distribution under the alternatives, mostly based on the local FDR framework.  For example, \citet{Sun:2015} developed an local-FDR-based procedure to incorporate spatial information. \citet{Scott:2015} and \citet{Tansey:2018} proposed EM-type algorithms to estimate the local FDR by taking into account covariate and spatial information, respectively. \citet{Lei:2018} proposed the AdaPT procedure, which iteratively estimates the p-value thresholds based on a two-group mixture model using the partially masked p-values together with the covariates. \citet{ZhangX:2020} proposed a more computationally efficient procedure to assign each p-value a covariate-adaptive threshold. Another related method  AdaFDR addressed the local ``bump" and global ``slope" structures delivered through the covariates in what they called ``enrichment" pattern by modeling a mixture of the generalized linear model and Gaussian mixture for a threshold function \citep{ZhangM:2019}.  Other relevant works include \citet{Ferkingstad:2008}, \citet{Zablocki:2014}, \citet{Dobriban:2015},  \citet{Wen:2016},  \citet{Stephens:2017}, \citet{Xiao:2017}, \citet{Lei:2017} and \citet{Li:2017}.

Recent developments on covariate-adaptive multiple testing focus on the FDR control, while methods for the FWER control lag behind. Existing FWER-controlling methods can all be thought to be variants of the weighted Bonferroni method, with the weights reflecting only the prior null probabilities. It has been demonstrated clearly in the FDR literature that incorporating the alternative p-value distribution leads to the optimal rejection region in theory and more power in practice, see, e.g., \citet{Efron:2010}. Given the popularity of the FWER control in GWAS, we introduce a new covariate-adaptive FWER-controlling procedure, which  takes into account the prior null probabilities as well as the alternative p-value distribution, making it distinct from the existing FWER-controlling procedures. To illustrate the idea, suppose we are given a set of p-values $p_i$ together with the external covariates $x_i$. Our method is motivated by the two-group mixture model
\begin{align*}
p_i\mid x_i\sim \pi(x_i)f_0(\cdot)+\{1-\pi(x_i)\}f_1(\cdot)
\end{align*}
with $\pi(x_i)$ and $f_1(\cdot)$ reflecting the heterogeneity of the probabilities of being null and the distributional characteristics of signals. We construct an objective function to control a conservative estimate of FWER while maximizing the expected number of true rejections. Specifically, we formulate the following constrained optimization problem
\begin{align*}
\begin{split}
&\max_{t_i}\sum_{i=1}^{m}\{1-\pi(x_i)\}F_1(t_i)~~\text{s.t.}~~\sum_{i=1}^m \pi(x_i) F_0(t_i)\leq \alpha,
\end{split}
\end{align*}
where $F_0$ and $F_1$ are the cumulative distribution functions of $f_0$ and $f_1$ respectively.
To establish the asymptotic FWER control, and the rate of convergence,  new theoretical developments are needed. Existing theoretical analysis techniques developed for the FDR-controlling procedures are not applicable to the FWER-controlling procedure, since we aim to control a sum instead of a proportion encountered in the FDR control. The arguments based on the Rademacher complexity in \citet{Li:2019} do not provide a meaningful bound on the FWER. Employing a perturbation-type argument, we develop a more delicate analysis for each of the summands, which leads to a useful bound on the sum and thus the FWER. The main contributions of the paper are two-fold:
\begin{itemize}
	\item We propose a powerful covariate-adaptive FWER-controlling procedure that can incorporate multi-dimensional covariates and exploit the  information from both the null probability and the alternative distribution. We prove the asymptotic FWER control of the proposed procedure when the pairs of covariate and p-value across different hypotheses are independent and derive the exact rate of convergence based on a novel perturbation technique. We emphasize that our proofs do not rely on the correct specification of the two-group mixture model.
	
	\item We develop an efficient algorithm to implement the proposed method and demonstrate its usefulness in handling big datasets arising from GWAS. In the application to the GWAS of about 9 million SNPs and 75 covariates, we could complete the analysis in hours. The proposed method is implemented in the R package \texttt{CAMT} and available at \texttt{https://github.com/jchen1981/CAMT}.
	
\end{itemize}

Numerical studies show that our procedure controls the FWER in the strong sense and is more powerful than the competing methods. It maintains the robustness across different settings, including scenarios of model misspecification and correlated hypotheses. Even when the covariates are not informative, our procedure is as powerful as the traditional methods.

\section{Methodology}

\subsection{The setup}
Denote by $\|v\|$ the Euclidean norm of a vector $v$. With some abuse of notation, let $\|A\|$ be the spectral norm of a matrix $A$. For two symmetric matrices $A$ and $B$, $A\preceq B$ means that $B-A$ is positive semidefinite. For $a,b\in\mathbb{R}$, write $a \vee b=\max(a,b)$ and $a\wedge b=\min(a,b)$. Throughout the paper, we use $c$ to denote a positive constant which can be different from line to line.

We consider the problem of covariate-adaptive multiple testing to control the FWER. Suppose we are given $m$ hypotheses, among which $m_0$ are true nulls.
For each hypothesis, we observe a p-value $p_i$ as well as a covariate $x_i$ lying in some space $\mathcal{X}\subseteq \mathbb{R}^d$ which encodes potentially useful external information concerning the presence of a signal. Let $H_i=0$ if the $i$th null hypothesis is true and $H_i=1$ otherwise. Denote by $\mathcal{M}_0$ the set of all true null hypotheses. We transform the $i$th p-value based on a map $T_i:[0,1]\rightarrow \mathbb{R}_+$ that will be estimated from the covariates and p-values. The larger $T_i(p_i)$ is, the more likely the $i$th hypothesis is from the alternative. The motivation for such a transformation will be discussed in the next section. In a nutshell, the optimal $T_i$ is the likelihood ratio between the $i$th p-value distributions under the alternative and the null.

\subsection{Optimal rejection rule}\label{sec-opti_rej}
Let $f_0(\cdot)$ be the null p-value distribution and $f_1(\cdot)$ denote the alternative p-value distribution. Denote by $F_0(\cdot)$ and $F_1(\cdot)$
the corresponding cumulative distribution functions. Suppose we reject the $i$th hypothesis if $p_i\leq t_i$ for some cutoff $t_i$. Before presenting the procedure that inspires the choice of $T_i$, it is worth clarifying the definition of the FWER from both the frequentist and Bayesian perspectives. The key difference between these two  viewpoints lies on whether we treat the indicators $\{H_i\}$ as fixed or random quantities. From the frequentist perspective, the indicators $\{H_i\}$ are deterministic and we have by the union bound,
\begin{align*}
\text{FWER}_{\text{Freq}}=&\mathbb{P}(p_i\leq t_i\text{ for some }i\in\mathcal{M}_0)\leq \sum_{i=1}^m\mathbb{I}(H_i=0)F_0(t_i).
\end{align*}
From a Bayesian's point of view, it is natural to posit the two-group mixture model
$$p_i\mid x_i \sim \pi(x_i)f_0(\cdot)+\{1-\pi(x_i)\}f_1(\cdot).$$
In this case, conditional on $x_i$, $H_i$ is assumed to be a Bernoulli random variable with the success probability $1-\pi(x_i)$. The Bayesian FWER can be bounded as follows
\begin{align}\label{eq-fwer-bay}
\text{FWER}_{\text{Bay}}=&\mathbb{P}(p_i\leq t_i\text{ for some }i\in\mathcal{M}_0)\leq  \sum^{m}_{i=1}\mathbb{P}(p_i\leq t_i,H_i=0)=\sum^{m}_{i=1}\mathbb{E}\left\{\pi(x_i)\right\}F_0(t_i).
\end{align}
To motivate our procedure, it is more convenient to adopt the Bayesian viewpoint. But we emphasize that the proposed procedure indeed provides asymptotic FWER control in the usual frequentist sense, as shown in Section \ref{sec-theory}.

We aim to find $\{t_i\}$ to maximize the expected number of true rejections given by
\begin{align*}
\mathbb{E}\left\{\sum^{m}_{i=1}\mathbb{I}(H_i=1,p_i\leq t_i)\right\}
=\sum^{m}_{i=1}\mathbb{E}[\{1-\pi(x_i)\}]F_1(t_i)
\end{align*}
while controlling the FWER at a desired level $\alpha.$ To achieve both goals, we formulate the following constraint optimization problem
\begin{equation}\label{pro}
\begin{split}
&\max_{t_i}\sum_{i=1}^{m}\{1-\pi(x_i)\}F_1(t_i)~~\text{s.t.}~\sum_{i=1}^m \pi(x_i) F_0(t_i)\leq \alpha,
\end{split}
\end{equation}
where $\sum^{m}_{i=1}\pi(x_i)F_0(t_i)$ serves as a conservative estimate of the Bayesian FWER
based on the derivations in (\ref{eq-fwer-bay}). The Lagrangian for problem (\ref{pro}) is 
\begin{align*}
L(t_1,\dots,t_m;\lambda)=\sum_{i=1}^{m}\{1-\pi(x_i)\}F_1(t_i)-\lambda\left\{\sum_{i=1}^m \pi(x_i) F_0(t_i)-\alpha\right\}
\end{align*}
with $\lambda>0$. Differentiating the Lagrangian with respect to $t_i$ and setting the derivative to be zero (at the optimal value $t_i^*$), we obtain
$$\frac{\{1-\pi(x_i)\}f_1(t_i^*)}{\pi(x_i)f_0(t_i^*)}=\lambda.$$
Motivated by the above observation, we set
$$T_i(p)=\frac{\{1-\pi(x_i)\}f_1(p)}{\pi(x_i)f_0(p)}.$$
We note that $T_i(p)$ is related to the local FDR as follows
\begin{align*}
\frac{1}{T_i(p)+1}=\frac{\pi(x_i)f_0(p)}{\pi(x_i)f_0(p)+\{1-\pi(x_i)\}f_1(p)}=\mathbb{P}(H_i=0\mid p,x_i).
\end{align*}

In the following discussions, we suppose $f_0$ is the uniform distribution on $[0,1]$ and $f_1$ is strictly decreasing, which is a common assumption in the literature, e.g., \citet{Sun:2007} and \citet{Cao:2013}. As $T_i$ is strictly decreasing in this case, we may reduce our attention to the rejection rule $p_i\leq t_i^*$ as $T_i(p_i)\geq T_i(t_i^*):=\tau^*$. The cutoff can then be expressed as
$$t_i^*=f_1^{-1}\left\{\frac{\pi(x_i)\tau^*}{1-\pi(x_i)}\right\},$$
where $f_1^{-1}$ denotes the inversion of $f_1$. Notice that the expected number of true rejections and the conservative estimate of the Bayesian FWER in (\ref{eq-fwer-bay}) are both monotonically decreasing in $\tau$. Therefore, the solution to (\ref{pro}) satisfies that
\begin{align}\label{eq-or}
\tau^*=\min\left\{\tau>0: \sum_{i=1}^m \pi(x_i)f_1^{-1}\left\{\frac{\pi(x_i)\tau}{1-\pi(x_i)}\right\}\leq \alpha\right\}.
\end{align}
In practice, both $\pi$ and $f_1$ are unknown and need to be replaced by estimates from the data. We provide detailed discussions about estimating the unknowns in the next subsection.

\subsection{A feasible procedure}\label{feasible}
We describe a feasible procedure based on suitable estimates of $\pi$ and $f_1$. To avoid overfitting and facilitate the theoretical analysis, we adopt the idea of censoring p-values as in \citet{Storey:2002}, \citet{Li:2019} and \citet{Boca:2018}. Under the two-group mixture model, for a prespecified $0<\gamma<1$, we have
\begin{align*}
\mathbb{I}(p_i>\gamma)\mid x_i\sim \pi(x_i)\text{Bern}(1-\gamma)+\{1-\pi(x_i)\}\text{Bern}\{1-F_1(\gamma)\},
\end{align*}
where $\text{Bern}(1-\gamma)$ denotes the Bernoulli distribution with success probability $1-\gamma$. We model $f_1$ using the beta distribution $f_1(p)=kp^{k-1}$ for $0<k<1$ as it provides reasonably well approximation to a wide range of alternative distributions as demonstrated in \citet{ZhangX:2020}. Here we treat $k$ as fixed and will discuss the choice of data-driven $k$ in Section \ref{sec-EM}.

Before presenting our method, it is worth clarifying the rationale behind our procedure. Notice that $\pi(x_i)$ appears both inside and outside the function $f_1^{-1}$ in (\ref{eq-or}). To achieve asymptotic FWER control, we need a conservative estimate for $\pi(x_i)$ outside the function $f_1^{-1}$, while for the one inside $f_1^{-1}$, we require it to depend on the covariates to reflect the heterogeneity among signals while retaining certain form of stability (see more details in Section \ref{sec-theory}). The reason will become clear by inspecting the proof of Proposition \ref{pro_jm1_jm2}. We first observe that
\begin{align*}
\mathbb{E}\left\{\frac{\mathbb{I}(p_i>\gamma)}{1-\gamma}\mid x_i\right\}=\pi(x_i)+\{1-\pi(x_i)\}\frac{1-F_1(\gamma)}{1-\gamma}\geq \pi(x_i).
\end{align*}
Therefore, we suggest replacing the $\pi(x_i)$ outside $f_1^{-1}$ by $\mathbb{I}(p_i>\gamma)/(1-\gamma)$. To estimate $\pi(x_i)$ inside $f_1^{-1}$, we consider the logistic model
\begin{align*}
\log\left\{\frac{\pi(x_i)}{1-\pi(x_i)}\right\}=x_i^\T \beta.
\end{align*}
The quasi log-likelihood function is then 
\begin{align*}
L_m(\beta)=\sum^{m}_{i=1}\log\left[\pi(x_i)(1-\gamma)^{y_i}\gamma^{1-y_i}+\{1-\pi(x_i)\}(1-\gamma^{k})^{y_i}\gamma^{k(1-y_i)}\right],
\end{align*}
where $\pi(x_i)=(1+e^{-x_i^\T \beta})^{-1}$ and $y_i=\mathbb{I}(p_i>\gamma)$. Define the corresponding quasi-maximum likelihood estimator (MLE) as
\begin{align}\label{quasi-mle-1}
\hat{\beta}=\arg\max_{\beta\in\mathcal{B}}L_m(\beta),
\end{align}
where $\mathcal{B}$ is some compact subset of $\mathbb{R}^d$. Let
\begin{align*}
\hat{\pi}(x_i)=\left\{\tilde{\pi}(x_i) \vee\varepsilon_1\right\}\wedge \varepsilon_2,
\end{align*}
where $\tilde{\pi}(x_i)=(1+e^{-x_i^\T \hat{\beta}})^{-1}$ and $0 < \varepsilon_1<\varepsilon_2<1$. We have used winsorization to prevent $\hat{\pi}(x_i)$ from being too close to zero and one. Further denote
\begin{align*}
\hat{\tau}=\min\left\{\tau\geq \varepsilon: \sum_{i=1}^m \frac{\mathbb{I}(p_i>\gamma)}{1-\gamma}f_1^{-1}\left\{\frac{\hat{\pi}(x_i)\tau}{1-\hat{\pi}(x_i)}\right\}\leq \alpha\right\}
\end{align*}
for some $\varepsilon>0$. It is straightforward to show that $\hat\tau=\tilde\tau\vee\varepsilon$ with
\begin{align*}
\tilde{\tau}=
k\left[\sum_{i=1}^{m}\frac{\mathbb{I}(p_i>\gamma)}{\alpha(1-\gamma)}\left\{\frac{1-\hat{\pi}(x_i)}{\hat{\pi}(x_i)}\right\}^{1/(1-k)}\right]^{1-k}.
\end{align*}
Finally, we set
$$\hat{t}_i=\left[\frac{\{1-\hat{\pi}(x_i)\}k}{\hat{\pi}(x_i)\hat{\tau}}\right]^{1/(1-k)},$$
and reject the $i$th hypothesis if
$$p_i \leq \hat{t}_i\wedge \gamma.$$

\begin{remark}[Connection to the weighted Bonferroni procedure]
	{\rm Suppose $\varepsilon_1=\varepsilon=0$ and $\varepsilon_2=1.$ Then we have
		\begin{align*}
		\hat{t}_i=\left[\frac{\{1-\tilde{\pi}(x_i)\}k}{\tilde{\pi}(x_i)\tilde{\tau}}\right]^{1/(1-k)}=\alpha w_i,
		\end{align*}
		where
		\begin{align*}
		w_i=e^{-\frac{x_i^\T\hat\beta}{1-k}}\left\{\sum_{i=1}^{m}\frac{\mathbb{I}(p_i>\gamma)}{1-\gamma}e^{-\frac{x_i^\T\hat\beta}{1-k}}\right\}^{-1}.
		\end{align*}
		We reject the $i$th hypothesis if
		$$p_i\le \alpha w_i\wedge \gamma.$$
		In this sense, our procedure can be viewed as a particular type of weighted Bonferroni procedure. However, different from existing methods, our weight incorporates the information regarding the alternative p-value distribution, which often leads to more rejections and thus higher power, as observed in our numerical studies.
	}
\end{remark}

	\subsection{EM algorithm}\label{sec-EM}
Algorithm 1 below provides the details of our iterative algorithm to solve problem (\ref{quasi-mle-1}).
\begin{algo}
	EM algorithm for problem (\ref{quasi-mle-1}).
	\begin{tabbing}
		\qquad Input: $\{x_i,y_i\}_{i=1}^m, \gamma,k$; initializer: $\beta^{(0)}$.\\
		\qquad Output: $\hat\beta$.\\
		\qquad Notation: $b_{0i}=(1-\gamma)^{y_i}\gamma^{1-y_i}$; $ b_{1i}=(1-\gamma^{k})^{y_i}\gamma^{k(1-y_i)}$; tol: tolerance level.\\
		\qquad Iteration:\\
		\qquad\qquad E step:\\
		\qquad\qquad\qquad$Q^{(t)}_i=\mathbb{E}\{\mathbb{I}(H_i=0)\mid y_i,x_i,\beta^{(t)}\}=\pi^{(t)}_ib_{0i}/\{\pi^{(t)}_ib_{0i}+(1-\pi^{(t)}_i)b_{1i}\}$,\\
		\qquad\qquad\qquad where $\pi^{(t)}_i=(1+e^{-x_i^\T \beta^{(t)}})^{-1}$. \\
		\qquad\qquad M step:\\
		\qquad\qquad\qquad$
		\beta^{(t+1)}=\arg\max_{\beta\in \mathcal{B}}\sum_{i=1}^m\{Q_i^{(t)}\log(\pi_i)+(1-Q_i^{(t)})\log(1-\pi_i)\}$,\\
		\qquad\qquad\qquad where $\pi_i=(1+e^{-x_i^\T\beta})^{-1}$.\\
		\qquad Until: $|L_m(\beta^{(t+1)})-L_m(\beta^{(t)})|/|L_m(\beta^{(t)})|<\text{tol}$.\\
		\qquad Return: $\beta^{(t+1)}$ after a sufficient number of iterations.
	\end{tabbing}
\end{algo}

The theory in Section \ref{sec-theory} below shows that our procedure controls FWER asymptotically for any fixed $k$. However, a suitable choice of $k$, which produces a beta-distribution closer to the true $f_1$ (especially on the small p-value region), will improve the statistical power. In practice, an EM algorithm can be used to estimate the $k$ and $\beta$ jointly. To be precise, we define the quasi log-likelihood function,
\begin{align*}
L_m(\beta,k)=\sum^{m}_{i=1}\log\left[\pi(x_i)(1-\gamma)^{y_i}\gamma^{1-y_i}+\{1-\pi(x_i)\}(1-\gamma^{k})^{y_i}\gamma^{k(1-y_i)}\right].
\end{align*}
Then we estimate $(\beta,k)$ jointly by the quasi-MLE defined as
\begin{align}\label{quasi-mle-2}
(\hat\beta,\hat k)=\arg\max_{\beta\in\mathcal{B},k\in(0,1)}L_{m}(\beta,k).
\end{align}
We summarize the algorithm for solving problem (\ref{quasi-mle-2}) in Algorithm 2.
\begin{algo}
	EM algorithm for problem (\ref{quasi-mle-2}).
	\begin{tabbing}
		\qquad Input: $\{x_i,y_i\}_{i=1}^m, \gamma$; initializer: $\beta^{(0)},k^{(0)}$.\\
		\qquad Output: $\hat\beta,\hat k$.\\
		\qquad Notation: $b_{0i}=(1-\gamma)^{y_i}\gamma^{1-y_i}$; tol: tolerance level.\\
		\qquad Iteration:\\
		\qquad\qquad E step:\\
		\qquad\qquad\qquad$Q^{(t)}_i=\mathbb{E}\{\mathbb{I}(H_i=0)\mid y_i,x_i,\beta^{(t)},k^{(t)}\}=\pi^{(t)}_ib_{0i}/\{\pi^{(t)}_ib_{0i}+(1-\pi^{(t)}_i)b_{1i}^{(t)}\},$\\
		\qquad\qquad\qquad where $\pi^{(t)}_i=(1+e^{-x_i^\T \beta^{(t)}})^{-1}$, $b_{1i}^{(t)}=(1-\gamma^{k^{(t)}})^{y_i}\gamma^{k^{(t)}(1-y_i)}$.\\
		\qquad\qquad M step:\\
		\qquad\qquad\qquad$
		\beta^{(t+1)}=\arg\max_{\beta\in \mathcal{B}}\sum_{i=1}^m \{Q_i^{(t)}\log(\pi_i)+(1-Q_i^{(t)})\log(1-\pi_i)\},
		$\\
		\qquad\qquad\qquad where $\pi_i=(1+e^{-x_i^\T\beta})^{-1}$; \\
		\qquad\qquad\qquad$
		k^{(t+1)}=\arg\max_{k\in (0,1)}\sum_{i=1}^m(1-Q_i^{(t)})\{y_i\log(1-\gamma^{k})+k(1-y_i)\log(\gamma)\}.
		$\\
		\qquad Until: $|L_m(\beta^{(t+1)},k^{(t+1)})-L_m(\beta^{(t)},k^{(t)})|/|L_m(\beta^{(t)},k^{(t)})|<\text{tol}$.\\
		\qquad Return: $\beta^{(t+1)},k^{(t+1)}$ after a sufficient number of iterations.
	\end{tabbing}
\end{algo}

	\section{Asymptotic FWER control}\label{sec-theory}
In this section, we prove the asymptotic FWER control for the procedure proposed in Section \ref{feasible}. Throughout this section, we shall adopt the frequentist viewpoint, i.e., we view the indicators $\{H_i\}$ as a deterministic sequence. 

Let $p_{j\rightarrow a}=(p_1,\dots,p_{j-1},a,p_{j+1},\dots,p_m)^\T\in\mathbb{R}^m$ for $a=0,1$. We define $\hat\beta(p_{j\to a})$ and $\hat t_i(p_{j\to a})$ by setting the $j$th p-value to be equal to $a$ when estimating the corresponding quantities. We make the following assumption.
\begin{assumption}\label{ass_unif_p}
	Denote by $F_{0i}$ the cumulative distribution function for $p_i$ with $H_i=0.$ Suppose that $\{p_i\}_{i\in\mathcal{M}_0}$ are super-uniform, i.e., $F_{0i}(t)\le t$ for all $t\in [0,1]$ and $i\in\mathcal{M}_0$.
\end{assumption}
Assumption \ref{ass_unif_p} is standard in the literature, see e.g., \citet{BY:2001}.
\begin{proposition}\label{pro_jm1_jm2}
	If $\{p_i\}\in\mathcal{M}_0$ are mutually independent and are independent with the non-null p-values, then under Assumption \ref{ass_unif_p}, we have
	\begin{align*}
	\text{FWER}\le J_m + \alpha\le c(J_{m,1}+J_{m,2})+\alpha,
	\end{align*}
	where
	\begin{align*}
	J_m=&\sum_{j=1}^m\mathbb{E}\left\{\left|\hat{t}_j(p_{j\rightarrow 0})
	-\hat{t}_j(p_{j\rightarrow 1})\right|\right\} \nonumber,\\
	J_{m,1}=&\sum_{j=1}^m\mathbb{E}\left[\frac{|x_j^\T\{\hat\beta(p_{j\rightarrow 0})-\hat\beta(p_{j\rightarrow 1})\}|}{\left\{c\alpha^{-1}\sum_{i\neq j}\mathbb{I}(p_i>\gamma)\right\}\vee\varepsilon^{1/(1-k)}}\right],\\
	J_{m,2}=&\sum_{j=1}^m\mathbb{E}\left(\frac{\alpha^{-1}\sum_{i\neq j}\mathbb{I}(p_i>\gamma)|x_i^\T\{\hat\beta(p_{j\rightarrow 0})-\hat\beta(p_{j\rightarrow 1})\}|+\alpha^{-1}}{\left[\left\{c\alpha^{-1}\sum_{i\neq j}\mathbb{I}(p_i>\gamma)\right\}\vee\varepsilon^{1/(1-k)}\right]^2}\right),
	\end{align*}
	and $\varepsilon$ has been defined in Section \ref{feasible}.
\end{proposition}
The above proposition shows that the validity of the asymptotic FWER control relies on the stability of $\hat t_j$, i.e., the smallness of $|\hat t_j(p_{j\to 0})-\hat t_j(p_{j\to 1})|$ which in turn depends on $\|\hat \beta(p_{j\to 0})-\hat\beta(p_{j\to 1})\|$. Set $z_i=(x_i,y_i)$, where $y_i=\mathbb{I}\{p_i>\gamma\}$. Define
\begin{align*}
l(\beta;z_i)=\log\left\{\frac{1}{1+e^{-x_i^\T\beta}}(1-\gamma)^{y_i}\gamma^{1-y_i}+\frac{e^{-x_i^\T\beta}}{1+e^{-x_i^\T\beta}}(1-\gamma^k)^{y_i}\gamma^{k(1-y_i)}\right\},
\end{align*}
and $\mathbb{P}_ml(\beta)=m^{-1}\sum_{i=1}^ml(\beta;z_i)$. To ensure $\|\hat \beta(p_{j\to 0})-\hat\beta(p_{j\to 1})\|$ to be small, we impose the following assumptions.
\begin{assumption}\label{ass_indp_z}
	Suppose $z_i\in\mathbb{R}^{d+1}$ are independent and possibly non-identically distributed.
\end{assumption}
Assumption \ref{ass_indp_z} is not uncommon in the multiple testing literature, see e.g., \citet{Ignatiadis:2016}. We suspect that the results still hold when $z_i$ is a sequence of weakly dependent variables although a rigorous proof is left for future investigation.
\begin{assumption}\label{ass_l_beta}
	There exists a continuous function of $\beta$, denoted by $\mathcal{L}(\beta)$, such that
	\begin{align*}
	\lim_{m\rightarrow+\infty}\sup_{\beta\in\mathcal{B}}\left|\mathbb{E}\left\{\mathbb{P}_ml(\beta)\right\}-\mathcal{L}(\beta)\right|=0.
	\end{align*}
\end{assumption}
\begin{assumption}\label{ass_unique_beta}
	Suppose $\mathcal{L}(\beta)$ has a unique global maximizer $\beta^*$ over the compact space $\mathcal{B}$.
\end{assumption}
Assumption 4 is needed in our perturbation argument. If the maximizer is not unique, there seems no guarantee that the difference between $\hat\beta(p_{j\to0})$ and $\hat\beta(p_{j\to1})$ will be small.
\begin{proposition}\label{pro_beta01} Suppose Assumptions \ref{ass_indp_z}--\ref{ass_unique_beta} are satisfied and further assume $\sup_{1\le i\le m}\mathbb{E}\left(\|x_i\|^{8}\right)<\infty$. Then we have
	\begin{align*}
	\hat\beta(p_{j\rightarrow 0})-\hat\beta(p_{j\rightarrow 1})=(S_j^*+\Delta_j)^{-1}(U_j^*+\Pi_j),
	\end{align*}
	where $S_j^*$ and $U_j^*$ are the leading terms such that $S_j^{*}=-\sum_{i\neq j}\nabla^2l(\beta^*;z_i)$ and $\sup_{1\leq j\leq m}\|U_j^{*}\|=O_\mathbb{P}(1)$, and $\Delta_j$ and $\Pi_j$ are the remainder terms satisfying that $$\sup_{1\leq j\leq m}\|\Delta_j\|=o_{\mathbb{P}}(m)~\text{and}~\sup_{1\leq j\leq m}\|\Pi_j\|=o_{\mathbb{P}}(1).$$
\end{proposition}
Given Propositions \ref{pro_jm1_jm2} and \ref{pro_beta01}, we have the following theorem of asymptotic FWER control.

\begin{theorem}\label{the_asym_fwer}
	Suppose the following conditions are satisfied:\\
	(i) Assumptions \ref{ass_unif_p}--\ref{ass_unique_beta} hold;\\
	(ii) for some $q\ge 2$ and $\epsilon>0$, we have $\sup_{1\le i\le m}\mathbb{E}\left(\|x_i\|^{4q+\epsilon}\right)< \infty$;\\
	(iii) we have $\sup_{\beta\in\mathcal{B}}\left|\mathbb{E}\left\{\mathbb{P}_ml(\beta)\right\}-\mathcal{L}(\beta)\right|=O(m^{-1/2})$;\\
	(iv) the function $\mathcal{L}(\beta)$ is twice continuously differentiable;  \\
	(v) the global maximizer $\beta^*$ is not on the boundary of $\mathcal{B}$;\\
	(vi) for some $c>0$, we have $\nabla^2\mathcal{L}(\beta^*)\preceq -cI$, where $I$ denotes the identity matrix;\\
	(vii) for large enough $m$ and some $c>0$, we have $\mathbb{E}\left\{\nabla^2\mathbb{P}_ml(\beta^{*})\right\}\preceq -cI$; and\\
	(viii) the number of true null hypotheses $m_0$ satisfies that $ \liminf m_0/m>0$.
	\\Then
	\begin{align*}
	\text{FWER}\le J_m+\alpha=\begin{cases}
	o(\alpha m^{\frac{1-q}{4}})+\alpha, &\text{if }2\le q\le 2+\sqrt{5} ,\\
	O(\alpha m^{\frac{-q}{1+q}})+\alpha, &\text{if }q>2+\sqrt{5} .
	\end{cases}
	\end{align*}
\end{theorem}
Theorem \ref{the_asym_fwer} derives the bound and its exact order on the FWER. Interestingly, the order of the bound depends crucially on the tail behavior of the covariates. And it shows an interesting phase transition depending on the value of $q$. We briefly explain this result as follows. From Proposition \ref{pro_jm1_jm2}, we can see that the FWER is upper bounded by an expression of the form $\alpha+\sum_{i=1}^{m}r_i$ with $r_i\ge 0$. Our argument optimizes the summation $\sum_{i=1}^{m}r_i$ in the upper bound. Depending on the value of $q$, the dominant term in this summation will change, which eventually leads to different convergence rates. When the covariates have exponential tails, the rate of convergence can be as close to $m^{-1}$ as possible. The details of the proof are provided in the supplementary material. As we discussed earlier, Assumptions \ref{ass_unif_p}--\ref{ass_unique_beta} enable us to show that the upper bound on the FWER relies on the smallness of $\|\hat \beta(p_{j\to 0})-\hat\beta(p_{j\to 1})\|$ and to get the expression for $\hat \beta(p_{j\to 0})-\hat\beta(p_{j\to 1})$. As our goal is to quantify the exact rate of  convergence of the FWER upper bound to the nominal level $\alpha$, we further need to quantify the exact difference between $\hat \beta(p_{j\to 0})$ and $\hat \beta(p_{j\to 1})$. Through Conditions (iii)--(v) and the strong-concavity Condition (vi), we obtain the concentration inequality for $\|\hat \beta(p_{j\to a})-\beta^*\|$. Conditions (ii) and (vii) are used for controlling the inverse $(S_j^*+\Delta_j)^{-1}$ in the expression of $\hat \beta(p_{j\to 0})-\hat\beta(p_{j\to 1})$. Condition (viii) requires the number of true null hypotheses to be at least some positive proportion of all hypotheses, which is fairly mild. We give one toy example where all the conditions are satisfied.
\begin{example}\label{example}
	Suppose all hypotheses are true nulls and $(x_i, p_i)$ are i.i.d.with $x_i$ being 1-dimensional, $x_i\indep p_i$ and $p_i\sim \text{Unif}([0,1])$, i.e., the uniform distribution on $[0,1]$. Then
	$\mathcal{L}(\beta)=\mathbb{E}\{\mathbb{P}_ml(\beta)\}=\mathbb{E}\left\{l(\beta;z_1)\right\}$.
	If $x_i$ follows a distribution symmetric about zero and $\mathbb{P}(x_i\neq 0)>0$, it can be shown that $\mathcal{L}'(\beta)=0$ if $\beta=0$, $\mathcal{L}'(\beta)>0$ if $\beta<0$, $\mathcal{L}'(\beta)<0$ if $\beta>0$, and $\mathcal{L}'(\beta)=-\mathcal{L}'(-\beta)$. Thus $\beta^*=0$ is the unique maximizer.
	We can further prove that $\mathcal{L}''(0)\le -c$ as long as $\mathbb{E}(x_i^2)>c'$ for some $c'>0$. Other conditions are naturally satisfied. When $x_i$ follows a non-symmetric distribution, we also illustrate its obedience to these conditions. One mandatory requirement for the distribution of $x_i$ is that $\mathbb{P}(x_i>0)>0$ and $\mathbb{P}(x_i<0)>0$. In practice, we could always achieve this by shifting the covariate via subtracting the median (or by standardizing the covariate). See more details in the supplementary material.
\end{example}

	\section{Numerical studies}\label{sec-sim}
\subsection{Simulation setups}
We conduct comprehensive simulations to evaluate the finite-sample performance of the proposed method and compare it to competing methods. For genome-scale multiple testing, the numbers of hypotheses could range from thousands to millions. For demonstration purpose, we start with $m=10, 000$ hypotheses. To study the impact of signal density and strength, we simulate three levels of signal density (sparse, medium and dense signals) and six levels of signal strength (from very weak to very strong). To demonstrate the power improvement by using external covariates, we simulate covariates of varying informativeness (non-informative, moderately informative and strongly informative). For simplicity, we simulate one covariate $x_i \sim N (0, 1) $ for $i = 1,\cdots, m$. Given $x_i$, we denote $\pi(x_i)$ by $\pi_{i}$ and let
\begin{align*}
\pi_{i}=\frac{\exp(\eta_i)}{1+\exp(\eta_i)},\quad\eta_i=\eta_0+k_dx_i,
\end{align*}
where $\eta_0$ and $k_d$ determine the baseline signal density and the informativeness of the covariate, respectively. We set $\eta_0 = 3.5, 2.5 \text{ and } 1.5$, which achieves a signal density of around 3\%, 8\%, and 18\% respectively at the baseline (i.e., no covariate effect), representing sparse, medium and dense signals. Here $k_d$ is set to be 0, 1 and 1.5, representing a non-informative, moderately informative and strongly informative covariate. Based on $\pi_{i}$, the underlying truth $H_i$ is simulated from
\begin{align*}
H_i\sim\text{Bern}(1-\pi_{i}).
\end{align*}
Finally, we simulate independent z-scores using
\begin{align*}
z_i\sim N(k_sH_i,1),
\end{align*}
where $k_s$ controls the signal strength (effect size), and we use six values equally spaced on [2, 2.8] and we label them as $\{1,2,...,6\}$. Z- scores are converted into p-values using the one-sided formula $1-\Phi(z_i)$. P-values together with $x_i$ are used as the input for the proposed method.

In addition to the basic setting (denoted as S0), we investigate other settings to study the robustness of the proposed method. Specifically, we study
\begin{itemize}
	\item Setup S1. \textit{Additional $f_1$ distribution}. Instead of simulating normal z-scores under $f_1$, we simulate z-scores from a non-central gamma distribution with the shape parameter $2$. The scale/non-centrality parameters of the non-central gamma distribution are chosen to match the variance and mean of the normal distribution under S0.
	
	\item Setup S2. \textit{Correlated hypotheses}. We further investigate the effect of dependency among hypotheses by simulating correlated multivariate normal z-scores. Four correlation structures, including two block correlation structures and two AR(1) correlation structures, are investigated. For the block correlation structure, we divide the 10,000 hypotheses into 500 equal-sized blocks. Within each block, we simulate equal positive correlations $(\rho=0.5)$ (S2.1). On top of S2.1, we divide the block into 2 by 2 sub-blocks, and simulate negative correlations $(\rho=-0.5)$ between the two sub-blocks (S2.2). For AR(1) structure, we investigate both $\rho=0.75^{|i-j|}$ (S2.3) and $\rho=(-0.75)^{|i-j|}$ (S2.4).
\end{itemize}
We present the simulation result for the Setup S0 in the main text and the results for the Setups S1 and S2 in the supplementary material.

\subsection{Competing methods}
We compared the proposed covariate-adaptive FWER-controlling procedure (denoted by CAMT.fwer) to IHW-Bonferroni, weighted Bonferroni and Holm's step-down methods \citep{Holm:1979}. The covariate-adaptive FWER-controlling procedure, implemented using the CAMT.fwer function in the R package \texttt{CAMT}, used the model $\log[{\pi(x_i)}/\{1-\pi(x_i)\}]=x_i^\T\beta$, set $f_1(p)=kp^{k-1}$ and estimated $\beta$ and $k$ jointly using Algorithm 2. The weighted Bonferroni method rejected the $i$th hypothesis if $p_i < {\alpha}/{(m\pi_i)}$, where $\pi_i$'s were estimated from CAMT.fwer. The IHW-Bonferroni method was implemented using the R package \texttt{IHW}, and Holm’s step-down method using the holm function from the R package \texttt{mutoss}. We also implemented an oracle procedure based on the proposed optimal rejection rule, where $\pi_i$'s and $f_1$ were the true null probabilities and alternative density that generated the data.

\citet{Storey:2004} proposed the bootstrap method to estimate the overall null probability $\pi$, which is implemented in the R package \texttt{qvalue}. The method uses censored p-values $\mathbb{I}\{p_i>\lambda\}$ with $\lambda=0.05,0.1,...,0.95$ to obtain the corresponding estimates of the null probability, $\pi_{\lambda}$, and returns the best  $\pi_{\hat\lambda}$. We set $\gamma=\hat\lambda$. We evaluated the performance based on the FWER control (probability of making at least one false positive) and power (true positive rate) with a target FWER level of 5\%. Results were averaged over 1000 simulation runs. In addition, we investigated the FWER control across different target levels, $\alpha=0.01,0.05,0.1,0.15,0.2$, for cases where there are no signals and under the Setup S0 with moderate signal density ($\eta_0=2.5$), signal strength ($k_s=2.4$) and covariate informativeness ($k_d=1$). 

\subsection{Simulation results}
We showcase the simulation results of Setup S0 in Figure \ref{fig_s0} and Setups S1--S2 in the supplementary material as well as the FWER control across different target levels (see Figures \ref{fig_no_signal}--\ref{fig_s2_4} in the supplementary material). All methods control the FWER around the  5\% target level  (Figure \ref{fig_s0}A). We additionally draw the 95\% confidence intervals (CIs) of the proposed method CAMT.fwer and observe that almost all the intervals cover the 5\% target level (dashed line) (Figure \ref{fig_s0}A), which suggests adequate FWER control of CAMT.fwer  under finite samples. In terms of power  (Figure \ref{fig_s0}B), generally, the five competing methods from the best to the worst are oracle, CAMT.fwer, IHW-Bonferroni  and weighted Bonferroni (the performance of these two methods depends on the cases), Holm's step-down methods. The oracle procedure represents the performance upper bound and dominates other methods.

\begin{figure}
	\begin{subfigure}[b]{1\textwidth}
		\centering
		\includegraphics[scale=0.45]{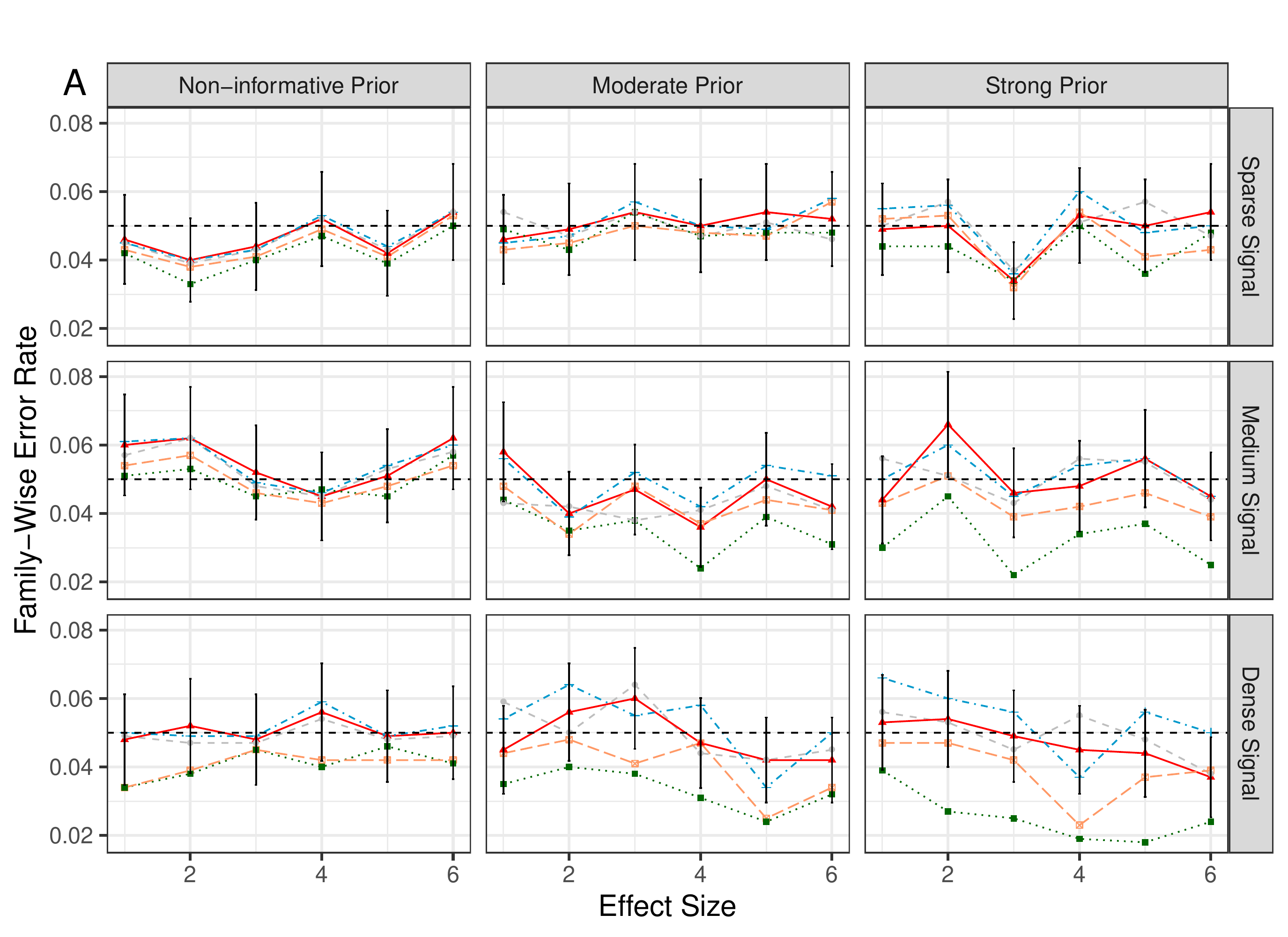}
	\end{subfigure}
	\begin{subfigure}[b]{1\textwidth}
		\centering
		\includegraphics[scale=0.45]{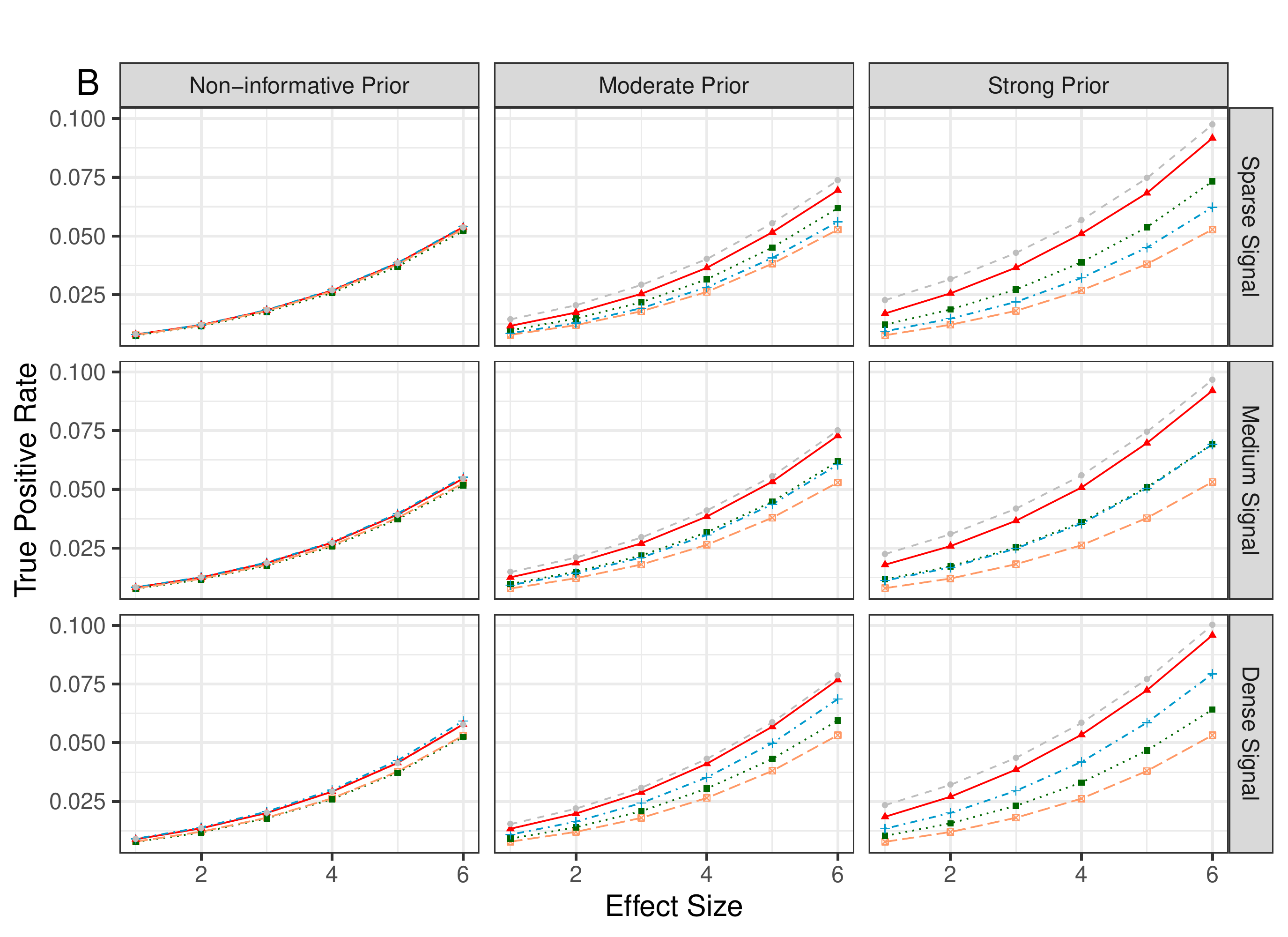}
	\end{subfigure}
	\caption{Performance comparison under the basic setting (S0). Family-wise error rates (A) and true positive rates (B) were averaged over 1000 simulation runs. The dashed gray, solid red, dotted green, dot-dashed blue and long-dashed orange lines represent the oracle, CAMT.fwer, IHW-Bonferroni, weighted Bonferroni and Holm's step-down methods respectively. The error bars (A) represent the 95\% CIs of the method CAMT.fwer and the dashed horizontal line indicates the target FWER level of 0.05.}
	\label{fig_s0}
\end{figure}

We now study the impact of the external (prior) information , signal density and strength  (Figure \ref{fig_s0}B). First, we note that the power increases with the signal strength (effect size) for all methods as expected. Second, as the prior informativeness increases, the performance difference  between methods widens.  CAMT.fwer  is close to the oracle procedure: it is as powerful as other methods when the prior is not informative and is substantially more powerful when the prior is highly informative. Both IHW-Bonferroni and weighted Bonferroni methods improve over Holm's step-down method when the prior is informative. Third, the proposed method maintains high power across different signal densities.  In contrast, IHW-Bonferroni method performs better than weighted Bonferroni method when the signal is sparse and performs worse when the signal is dense.

Figures \ref{fig_no_signal}--\ref{fig_moderate_fwer} show the weak and strong FWER control of the competing methods across different target levels. All the methods including CAMT.fwer control the FWER at the target level. Figure \ref{fig_moderate_power} compares the power across different target levels at moderate signal density, signal strength and prior informativeness. CAMT.fwer remains more powerful than other methods. In fact, as the target level increases, the power difference becomes larger.

We next study the robustness of the proposed method under the Setup S1 (additional $f_1$ distribution, Figure \ref{fig_s1}) and Setup S2 (correlated hypotheses, Figures \ref{fig_s2_1}--\ref{fig_s2_4}).  The general trend remains similar to Setup S0, indicating that CAMT.fwer is robust to different $f_1$ distributions and various correlation structures.  Interestingly,  as we generate z-scores from non-central gamma distribution for the alternative in Setup S1, the power of CAMT.fwer is even closer to that of the oracle procedure (Figure \ref{fig_s1}), indicating that the beta distribution can model the alternative p-value distribution very accurately in this case.

\section{Application to GWAS of UK Biobank data}\label{sec-real}
To demonstrate the use of the proposed procedure in real world applications, we applied CAMT.fwer to UK Biobank data \citep{Kichaev:2019}. We downloaded the data (p-values and functional annotations) from \texttt{https://data.broadinstitute.org/alkesgroup/UKBB/} and \texttt{https://data.broadinstitute.org/alkesgroup/FINDOR/}. The genome-wide association p-values for 9 million SNPs and 27 traits were calculated using BOLT-LMM \citep{Loh:2018} based on 459K samples.  The annotation data consists of 75 coding, conserved, regulatory, and linkage-disequilibrium-related annotations that have previously been shown to be enriched for the disease heritability \citep{Kichaev:2019}. We compared our method with IHW-Bonferroni, weighted Bonferroni and Holm's step-down methods. For the IHW-Bonferroni method, as it can only deal with one-dimensional covariate, we chose the covariate that had the maximum Spearman correlation with the p-values out of the 75 covariates for the 27 traits separately. For the weighted Bonferroni method, we rejected the $i$th hypothesis if $p_i < {\alpha}/{(m\pi_i)}$, where $\pi_i$'s were estimated from CAMT.fwer.  The details of the use of CAMT.fwer are given below.

Appropriate initial values of $(\beta, k)$ are important for the algorithm to reach convergence in less iterations and reduce the computation time significantly.  To achieve this end, we estimate those initial values based on small p-values (so the initial beta distribution fits the small p-value region more accurately). Let $\pi^s$ be the estimate of the proportion of the true null hypotheses based on Storey's procedure (R package \texttt{qvalue}). We define the ``small p-values" as the first $m(1-\pi^s)$ smallest p-values and let $u$ be the maximum value of those small p-values. Note that
$$f(p\mid p<u)=\frac{\pi+(1-\pi)kp^{k-1}}{\pi u +(1-\pi)u^k},$$
is the conditional density of the mixture model $f(p)=\pi+(1-\pi)k p^{k-1}$ given that the value is less than $u$. We estimate $\pi$ and $k$ by maximizing the (conditional) log-likelihood function,
\begin{align*}
(\tilde\pi,\tilde k)=\arg\max_{\pi\in (0,1), k\in(0,1)}\sum_{i:p_i<u}\log\left\{\pi+(1-\pi)kp_i^{k-1}\right\}-n\log\{\pi u + (1 - \pi) u ^ k\},
\end{align*}
where $n$ is the number of p-values that are smaller than $u$. Let $\tilde\beta=(\log\{\tilde\pi/(1-\tilde\pi)\},0)^\T$. Then we set $(\tilde\beta, \tilde k)$ as the initializer in Algorithm 2.

Due to the linkage disequilibrium between SNPs, after getting the rejected SNPs, we used PLINK’s linkage-disequilibrium-based clumping algorithm with a 5 Mb window and an $r^2$ threshold of 0.01 to form clumps of SNPs. The British population in the 1000 genomes data \citep{1000:2015} was used to calculate the linkage disequilibrium. The rejected SNPs belonging to the same clump count for only one significant locus. The numbers of significant loci at the 5\% FWER level detected by the four competing methods are presented in Table \ref{table-rej-1}.  We present the numbers of rejections before clumping in the supplementary material. CAMT.fwer detected more loci than other methods in 21 out of the 27 traits. Averaged across the traits, our approach attained 4.20\% increase in significant loci detected compared with the Holm's step-down method.

\begin{table}\centering
	\caption{Significant loci detected at the FWER level of 0.05. Improve=$(\text{CAMT.fwer}-\text{Holm})/\text{Holm}\times 100\%$. The numbers with subscript $*$ are the maximum numbers of rejections among the four competing methods for the corresponding traits}
	\footnotesize
		\begin{tabular}{lccccc} 
			\\[-1.8ex]\hline 
			\hline \\[-1.8ex] 
			& Holm & IHW & weighted Bonferroni & CAMT.fwer & Improve \\ 
			\hline \\[-1.8ex] 
			Balding Type I & $836_*$ & $836_*$ & $836_*$ & 833 & -0.4\% \\ 
			BMI & 1287 & 1287 & 1347 & $1364_*$ & 6\% \\ 
			Heel T Score & 2104 & 2104 & 2144 & $2146_*$ & 2\% \\ 
			Height & 3463 & 3460 & $3555_*$ & 3550 & 2.5\% \\ 
			Waist-hip Ratio & 909 & 909 & 937 & $952_*$ & 4.7\% \\ 
			Eosinophil Count & 1750 & 1750 & $1817_*$ & 1797 & 2.7\% \\ 
			Mean Corpular Hemoglobin & 1913 & 1913 & $1953_*$ & 1925 & 0.6\% \\ 
			Red Blood Cell Count & 1570 & 1570 & 1609 & $1633_*$ & 4\% \\ 
			Red Blood Cell Distribution Width & 1470 & 1470 & $1493_*$ & 1470 & 0\% \\ 
			White Blood Cell Count & 1393 & 1393 & 1430 & $1462_*$ & 5\% \\ 
			Auto Immune Traits & 179 & 179 & $180_*$ & 138 & -22.9\% \\ 
			Cardiovascular Diseases & 512 & 512 & 529 & $540_*$ & 5.5\% \\ 
			Eczema & 423 & 423 & 426 & $431_*$ & 1.9\% \\ 
			Hypothyroidism & 373 & 373 & 377 & $424_*$ & 13.7\% \\ 
			Respiratory and Ear-nose-throat Diseases & 228 & 228 & 231 & $236_*$ & 3.5\% \\ 
			Type 2 Diabetes & 156 & 156 & 158 & $160_*$ & 2.6\% \\ 
			Age at Menarche & 634 & 634 & 648 & $652_*$ & 2.8\% \\ 
			Age at Menopause & 200 & 200 & 201 & $203_*$ & 1.5\% \\ 
			FEV1-FVC Ratio & 1537 & 1537 & 1575 & $1599_*$ & 4\% \\ 
			Forced Vital Capacity (FVC) & 867 & 867 & 924 & $947_*$ & 9.2\% \\ 
			Hair Color & 1606 & 1606 & 1616 & $1629_*$ & 1.4\% \\ 
			Morning Person & 204 & 204 & 217 & $229_*$ & 12.3\% \\ 
			Neuroticism & 176 & 115 & 189 & $198_*$ & 12.5\% \\ 
			Smoking Status & 221 & 159 & 232 & $254_*$ & 14.9\% \\ 
			Sunburn Occasion & 232 & 232 & 232 & $237_*$ & 2.2\% \\ 
			Systolic Blood Pressure & 1108 & 1108 & 1148 & $1157_*$ & 4.4\% \\ 
			Years of Education & 383 & 383 & 416 & $447_*$ & 16.7\% \\ 
			\hline \\[-1.8ex] 
	\end{tabular}
	\label{table-rej-1}
\end{table}

\section{Discussions}\label{sec-disc}
To conclude, we point out a few future research directions. First, in the two-group mixture model, we assume that the success probabilities $\pi(x_i)$ vary with $x_i$ while $f_1$ is independent of $x_i$. This assumption is reasonable in some applications but it can be restrictive when the covariates also affect the effect sizes. It is thus of interest to develop
a procedure by allowing $f_1$ to be dependent on $x_i$ in such scenarios. Second, modeling $f_1$ and $\pi$ using nonparametric procedures would give us the flexibility to capture more complicated signal patterns. Finally, extending the method to accommodate more general structural information such as the phylogenetic tree structure \citep{Xiao:2017} is an interesting direction.

\section*{Acknowledgement}
Both Zhou and Zhang acknowledge partial support from NSF DMS-1830392 and NSF DMS-1811747. Zhou acknowledges partial support from China Scholarship Council. Chen acknowledges the support from Mayo Clinic Center for Individualized Medicine. We thank Dr. Kejun He for the help for the proof of Lemma \ref{lem_leave_one} in the supplementary material. We are also grateful to the associate editor and two reviewers for the insightful comments, which substantially improve the paper.

\newpage
\setcounter{section}{0}
\renewcommand{\thesection}{S\arabic{section}}
\setcounter{subsection}{0}
\renewcommand{\thesubsection}{S\arabic{subsection}}
\setcounter{equation}{0}
\renewcommand{\theequation}{S\arabic{equation}}
\setcounter{figure}{0}
\renewcommand{\thefigure}{S\arabic{figure}}
\setcounter{table}{0}
\renewcommand{\thetable}{S\arabic{table}}
\setcounter{lemma}{0}
\renewcommand{\thelemma}{S\arabic{lemma}}
\setcounter{lemma}{0}
\renewcommand{\thelemma}{S\arabic{lemma}}
\setcounter{proposition}{0}
\renewcommand{\theproposition}{S\arabic{proposition}}
\setcounter{remark}{0}
\renewcommand{\theremark}{S\arabic{remark}}

\begin{center}
	\Large\bf{Supplementary Material}
\end{center}

\begin{abstract}
	The supplementary material is organized as follows. In Section \ref{sec-example1}, we provide more technical details for Example \ref{example}. Section \ref{sec-inter-prop2} presents Lemmas \ref{lem_ulln}--\ref{lem_consistent} and their proofs that are useful in the proof of Proposition \ref{pro_beta01}. Section \ref{sec-props12} presents the proofs of Propositions \ref{pro_jm1_jm2}--\ref{pro_beta01}. In Section \ref{sec-interm}, we provide other intermediate results for the proof of Theorem \ref{the_asym_fwer} together with their proofs. In Section \ref{sec-theorem1}, we prove Theorem \ref{the_asym_fwer}. Sections \ref{sec-simu} and \ref{sec-gwas} present the additional simulation results and the numbers of rejections before clumping mentioned in Section \ref{sec-real} of the main paper, respectively.
\end{abstract}

\section{More about Example 1}\label{sec-example1}

Suppose all hypotheses are true nulls and $(x_i, p_i)$ are i.i.d.with $x_i$ being 1-dimensional, $x_i\indep p_i$ and $p_i\sim \text{Unif}([0,1])$. Then
\begin{align*}
&\mathcal{L}(\beta)=\mathbb{E}\{\mathbb{P}_ml(\beta)\}=\mathbb{E}\left\{l(\beta;z_i)\right\}\\
=&\mathbb{E}\left[(1-\gamma)\log\left\{1-\gamma+(\gamma-\gamma^k)\frac{e^{-x_i\beta}}{1+e^{-x_i\beta}}\right\}+\gamma\log\left\{\gamma-(\gamma-\gamma^k)\frac{e^{-x_i\beta}}{1+e^{-x_i\beta}}\right\}\right]
\end{align*}
and $\mathcal{L}'(\beta)=(\gamma-\gamma^k)^2\mathbb{E}\left\{g(x_i;\beta)\right\}$, where 
\begin{align*}
g(x;\beta)=\frac{xe^{-x\beta}}{(1+e^{-x\beta})^2}\frac{\frac{e^{-x\beta}}{1+e^{-x\beta}}}{\left\{1-\gamma+(\gamma-\gamma^k)\frac{e^{-x\beta}}{1+e^{-x\beta}}\right\}\left\{\gamma-(\gamma-\gamma^k)\frac{e^{-x\beta}}{1+e^{-x\beta}}\right\}}.
\end{align*}
We observe that $	\gamma-\gamma^k < 0$,
\begin{align}
\begin{split}
1-\gamma^k<&1-\gamma+(\gamma-\gamma^k)\frac{e^{-x\beta}}{1+e^{-x\beta}}<1-\gamma,\\
\gamma<&\gamma-(\gamma-\gamma^k)\frac{e^{-x\beta}}{1+e^{-x\beta}}<\gamma^k,
\end{split}
\label{eq-l-beta-z}
\end{align}
and
\begin{align*}
g(x;\beta)+g(-x;\beta)=&\frac{xe^{-x\beta}}{(1+e^{-x\beta})^2}\frac{\frac{e^{-x\beta}}{1+e^{-x\beta}}}{\left\{1-\gamma+(\gamma-\gamma^k)\frac{e^{-x\beta}}{1+e^{-x\beta}}\right\}\left\{\gamma-(\gamma-\gamma^k)\frac{e^{-x\beta}}{1+e^{-x\beta}}\right\}}\\
&-\frac{xe^{x\beta}}{(1+e^{x\beta})^2}\frac{\frac{e^{x\beta}}{1+e^{x\beta}}}{\left\{1-\gamma+(\gamma-\gamma^k)\frac{e^{x\beta}}{1+e^{x\beta}}\right\}\left\{\gamma-(\gamma-\gamma^k)\frac{e^{x\beta}}{1+e^{x\beta}}\right\}}\\
=&\frac{xe^{x\beta}}{(1+e^{x\beta})^2}\Bigg[\frac{\frac{1}{1+e^{x\beta}}}{\left\{1-\gamma+(\gamma-\gamma^k)\frac{1}{1+e^{x\beta}}\right\}\left\{\gamma-(\gamma-\gamma^k)\frac{1}{1+e^{x\beta}}\right\}}\\
&\qquad\qquad\qquad-\frac{\frac{e^{x\beta}}{1+e^{x\beta}}}{\left\{1-\gamma+(\gamma-\gamma^k)\frac{e^{x\beta}}{1+e^{x\beta}}\right\}\left\{\gamma-(\gamma-\gamma^k)\frac{e^{x\beta}}{1+e^{x\beta}}\right\}}\Bigg]\\
=&\frac{xe^{x\beta}}{(1+e^{x\beta})^2}\Bigg[\frac{\frac{e^{x\beta}}{1+e^{x\beta}}}{\left\{1-\gamma+(\gamma-\gamma^k)\frac{1}{1+e^{x\beta}}\right\}\left\{\gamma e^{x\beta}-(\gamma-\gamma^k)\frac{e^{x\beta}}{1+e^{x\beta}}\right\}}\\
&\qquad\qquad\qquad-\frac{\frac{e^{x\beta}}{1+e^{x\beta}}}{\left\{1-\gamma+(\gamma-\gamma^k)\frac{e^{x\beta}}{1+e^{x\beta}}\right\}\left\{\gamma-(\gamma-\gamma^k)\frac{e^{x\beta}}{1+e^{x\beta}}\right\}}\Bigg].
\end{align*}
Thus if $\beta = 0$ then $g(x;\beta)+g(-x;\beta)=0$, if $\beta > 0$ then $g(x;\beta)+g(-x;\beta)<0$ (as long as $x\neq 0$), and if $\beta < 0$ then $g(x;\beta)+g(-x;\beta)>0$ (as long as $x\neq 0$). Note also that $g(0;\beta)=0$ and $g(x;\beta)+g(-x;\beta)=-\{g(x;-\beta)+g(-x;-\beta)\}$. Therefore, if $x_i$ follows a distribution that is symmetric about zero and $\mathbb{P}(x_i\neq 0) >0$, then we have
\begin{align*}
\mathcal{L}'(\beta)=&(\gamma-\gamma^k)^2\mathbb{E}\{g(x_i;\beta)\mathbb{I}(x_i>0) + g(x_i;\beta)\mathbb{I}(x_i\le 0)\}\\
=&(\gamma-\gamma^k)^2\mathbb{E}\{g(x_i;\beta)\mathbb{I}(x_i>0) + g(-x_i;\beta)\mathbb{I}(x_i> 0)\}\\
=&(\gamma-\gamma^k)^2\mathbb{E}[\{g(x_i;\beta)+g(-x_i;\beta)\}\mathbb{I}(x_i>0)]
\end{align*}
and hence
$\mathcal{L}'(\beta)=-\mathcal{L}'(-\beta)$, and $\mathcal{L}'(\beta)=0$ if $\beta=0$, $\mathcal{L}'(\beta)>0$ if $\beta<0$ and $\mathcal{L}'(\beta)<0$ if $\beta>0$.
Thus $\beta^*=0$. Let $h(x;\beta)=e^{-x\beta}/(1+e^{-x\beta})$ and $h'(x;\beta)=-xe^{-x\beta}/(1+e^{-x\beta})^2$. Then 
\begin{align*}
g(x;\beta)=&\frac{x\left\{h(x;\beta)\right\}^3e^{x\beta}}{\left\{1-\gamma+(\gamma-\gamma^k)h(x;\beta)\right\}\left\{\gamma-(\gamma-\gamma^k)h(x;\beta)\right\}},\\
g'(x;\beta)=&\frac{\left[3x\left\{h(x;\beta)\right\}^2h'(x;\beta)e^{x\beta}+x^2\left\{h(x;\beta)\right\}^3e^{x\beta}\right]\left\{1-\gamma+(\gamma-\gamma^k)h(x;\beta)\right\}\left\{\gamma-(\gamma-\gamma^k)h(x;\beta)\right\}}{\left\{1-\gamma+(\gamma-\gamma^k)h(x;\beta)\right\}^2\left\{\gamma-(\gamma-\gamma^k)h(x;\beta)\right\}^2}\\
&-\frac{x\left\{h(x;\beta)\right\}^3e^{x\beta}\left[(\gamma-\gamma^k)h'(x;\beta)\left\{2\gamma-1-2(\gamma-\gamma^k)h(x;\beta)\right\}\right]}{\left\{1-\gamma+(\gamma-\gamma^k)h(x;\beta)\right\}^2\left\{\gamma-(\gamma-\gamma^k)h(x;\beta)\right\}^2},
\end{align*}
and
\begin{align*}
g'(x;0)=&\frac{-\frac{x^2}{16}\left\{1-\gamma+(\gamma-\gamma^k)/2\right\}\left\{\gamma-(\gamma-\gamma^k)/2\right\}+\frac{x^2}{32}(\gamma-\gamma^k)(\gamma+\gamma^k-1)}{\left\{1-\gamma+(\gamma-\gamma^k)/2\right\}^2\left\{\gamma-(\gamma-\gamma^k)/2\right\}^2}\\
=&-\frac{x^2}{32}u(\gamma,k),
\end{align*}
where
\begin{align*}
u(\gamma,k)=\frac{2\left\{1-\gamma+(\gamma-\gamma^k)/2\right\}\left\{\gamma-(\gamma-\gamma^k)/2\right\}-(\gamma-\gamma^k)(\gamma+\gamma^k-1)}{\left\{1-\gamma+(\gamma-\gamma^k)/2\right\}^2\left\{\gamma-(\gamma-\gamma^k)/2\right\}^2}.
\end{align*}
We study $u(\gamma,k)$ numerically. Let $\gamma$ range from $0.05$ to $0.95$. Then the values of $\min\{u(\gamma,k);k\in[0, 1]\}$ are shown in Figure \ref{u_gamma_k}. We point out that the left most point in Figure \ref{u_gamma_k} is $\min\{u(0.05,k):k\in[0, 1]\}=u(0.05,0.23)=4.94$.
\begin{figure}
	\centering
	\includegraphics[scale=0.45]{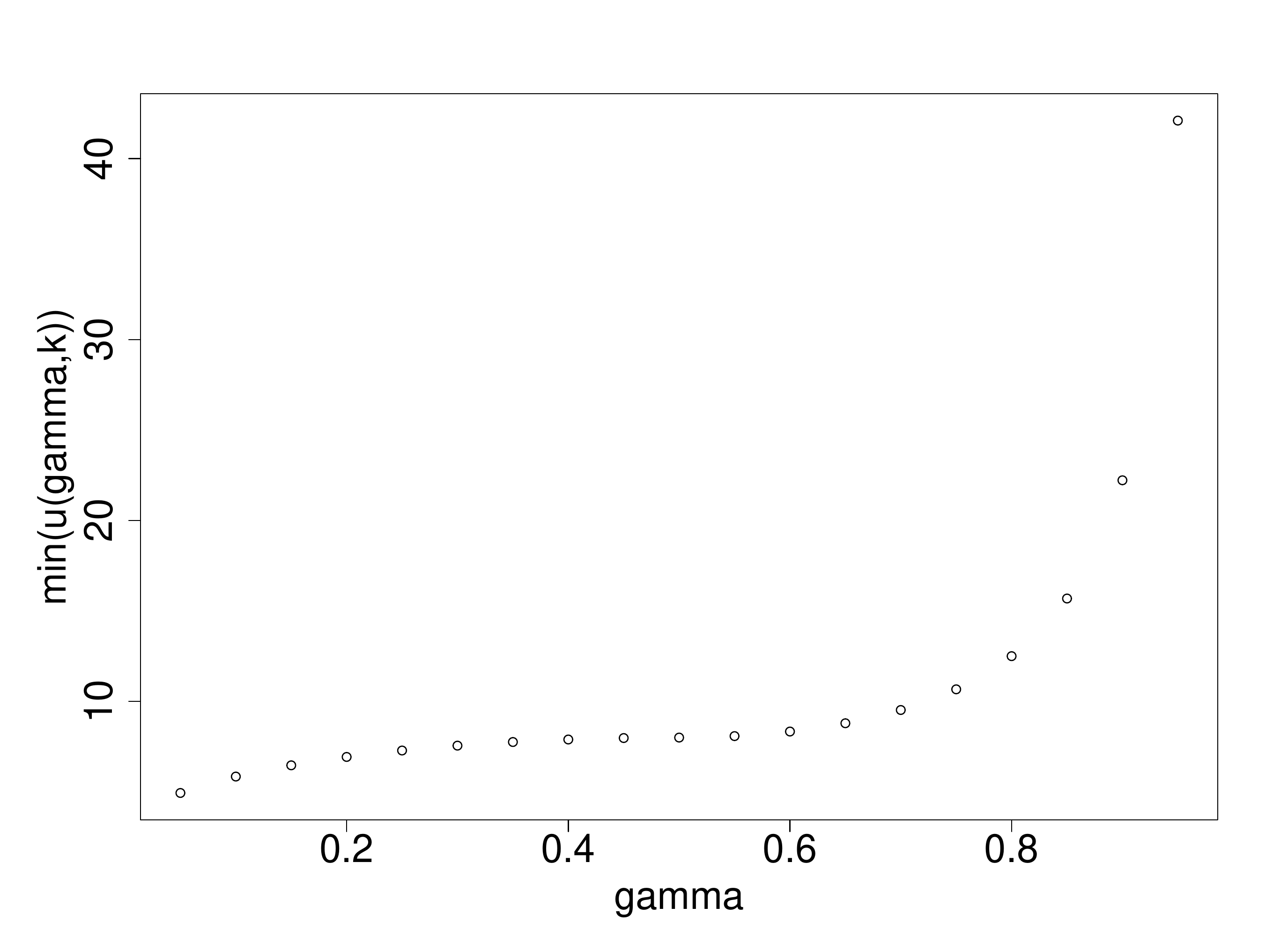}
	\caption{Values of $\min\{u(\gamma,k);k\in[0, 1]\}$ with $\gamma$ ranging from 0.05 to 0.95.}
	\label{u_gamma_k}
\end{figure}
Therefore 
$$\mathcal{L}''(0)=\frac{\text{d}}{\text{d}\beta}(\gamma-\gamma^k)^2\mathbb{E}\{g(x_i,\beta)\}\Big|_{\beta=0}=-\frac{(\gamma-\gamma^k)^2u(\gamma,k)}{32}\mathbb{E}(x_i^2)\le -c$$
as long as $\mathbb{E}(x_i^2)> c'$ for some $c'>0$. Next, we illustrate the behavior of $\mathcal{L}'(\beta)$ by considering four cases: (i) $x_i\sim N(0, 1)$; (ii) $x_i\sim \text{Gamma}(1, 0.5)$, i.e., Gamma distribution with shape 1 and rate 0.5; (iii) $x_i=w_i-\text{median}(w_1,...,w_m)$, where $w_i\sim\text{Gamma}(1, 0.5)$; (iv) $x_i=w_i-\text{median}(w_1,...,w_m)$, where $w_i\sim\text{Pois}(2)$, i.e., Poisson distribution with parameter 2. Let $\gamma=0.5$, $k=0.25$. When $x_i$ follows the standard normal distribution which is symmetric about zero, the behavior of $\mathcal{L}'(\beta)$ as observed in (i) of Figure \ref{l_beta_deriv} is consistent with our theory. As illustrated by (ii) of Figure \ref{l_beta_deriv}, Condition (v) is violated if $\mathbb{P}(X\geq 0)=1$ (or $\mathbb{P}(X\leq 0)=1$). Nevertheless, we can resolve this issue by simply subtracting the observations by the sample median. As we can see from (iii)--(iv) of Figure \ref{l_beta_deriv}, the pattern of the curve of $\mathcal{L}'(\beta)$ is similar to that of the standard normal distribution. In addition, by standardizing the observations, we can obtain the similar curves as (iii)--(iv) of Figure \ref{l_beta_deriv}.
\begin{figure}
	\centering
	\includegraphics[scale=0.45]{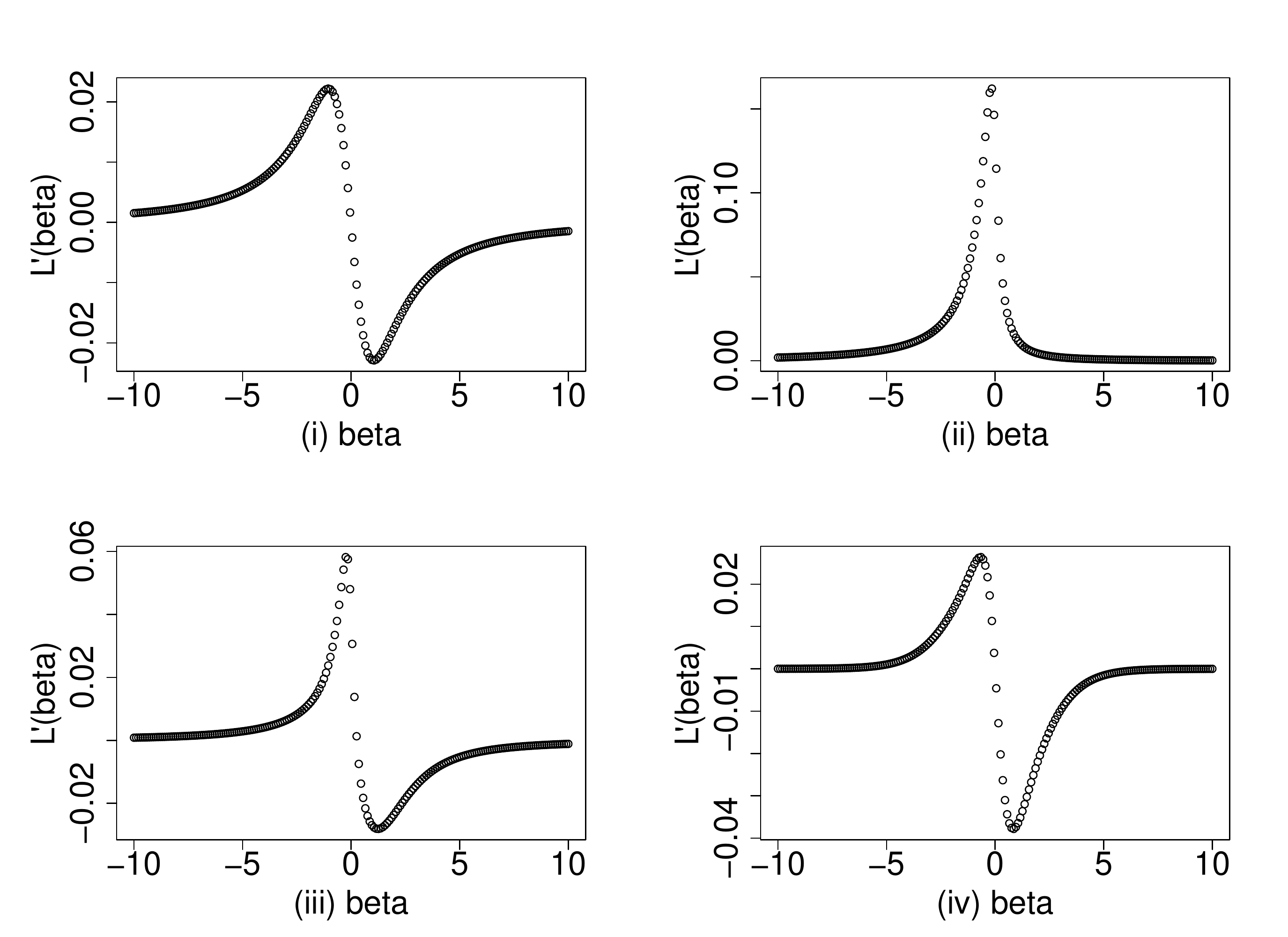}
	\caption{Curves of $\mathcal{L}'(\beta)$ with $\gamma=0.5$ and $k=0.25$ under four different cases. (i) $x_i\sim N(0, 1)$; (ii) $x_i\sim \text{Gamma}(\text{shape}=1, \text{rate}=0.5)$; (iii) $x_i=w_i-\text{median}(w_1,...,w_m)$, where $w_i\sim\text{Gamma}(\text{shape}=1, \text{rate}=0.5)$; (iv) $x_i=w_i-\text{median}(w_1,...,w_m)$, where $w_i\sim\text{Pois}(2)$.}
	\label{l_beta_deriv}
\end{figure}

\section{Intermediate results for Proposition 2}\label{sec-inter-prop2}

\begin{lemma}\label{lem_ulln}
	Assume that $\sup_{1\le i\le m}\mathbb{E}(\|x_i\|^2)<\infty$. Then under Assumptions 2 and 3, we have
	\begin{align*}
	\sup_{\beta\in \mathcal{B}}|\mathbb{P}_ml(\beta)-\mathcal{L}(\beta)|\to0,\quad\text{in probability.}
	\end{align*}
\end{lemma}

\begin{proof}
	Note that 
	\begin{align*}
	\begin{split}
	l\{\beta;(x,1)\}=&\log\left[(1+e^{-x^\T\beta})^{-1}(1-\gamma)+\left\{1-(1+e^{x^\T\beta})^{-1}\right\}(1-\gamma^k)\right],\\
	l\{\beta;(x,0)\}=&\log\left[(1+e^{-x^\T\beta})^{-1}\gamma+\left\{1-(1+e^{x^\T\beta})^{-1}\right\}\gamma^k\right].
	\end{split}
	\end{align*}
	Thus $l(\beta;z)$ is bounded uniformly over $\beta$ and $z$ according to (\ref{eq-l-beta-z}), that is, $L_1\le l(\beta;z)\le L_2$, where $L_1 = \log(1-\gamma^k)\wedge\log(\gamma)$ and $L_2=\log(1-\gamma)\vee\log(\gamma^k)$.
	Note also that for $0<\gamma<1$,
	\begin{align*}
	\| \nabla l\{\beta;(x,1)\}\|&=\left\|\frac{(\gamma^k-\gamma)(1+e^{-x^\T\beta})^{-2}e^{-x^\T\beta}x}{(\gamma^k-\gamma)(1+e^{-x^\T\beta})^{-1}+(1-\gamma^k)}\right\|\le\frac{\gamma^k-\gamma}{1-\gamma^k}\|x\|,\\
	\| \nabla l\{\beta;(x,0)\}\|&=\left\|\frac{(\gamma-\gamma^k)(1+e^{-x^\T\beta})^{-2}e^{-x^\T\beta}x}{(\gamma-\gamma^k)(1+e^{-x^\T\beta})^{-1}+\gamma^k}\right\|\le\frac{\gamma^k-\gamma}{\gamma}\|x\|.
	\end{align*}
	Thus for any $\beta_1,\beta_2\in\mathcal{B}$, we have
	\begin{align*}
	\left|l(\beta_1;z)-l(\beta_2;z)\right|\le c\|x\|\|\beta_1-\beta_2\|,
	\end{align*}
	which implies that $l(\beta;z)$ is $c\|x\|$-Lipschitz continuous in $\beta$. Let $\{\beta_n\}_{n=1}^N$ be an $\epsilon/(2c)$ covering of $\mathcal{B}$, that is, for any $\beta\in\mathcal{B}$, there exists $n$, such that $\|\beta-\beta_n\|\le\epsilon/(2c)$. Then for $i=1,...,m$, we have
	\begin{align*}
	\left|l(\beta;z_i)-l(\beta_n;z_i)\right|\le c\|x_i\|\|\beta-\beta_n\|\le\frac{\epsilon}{2}\|x_i\|,
	\end{align*}
	and hence
	\begin{align*}
	l(\beta_n;z_i)-\frac{\epsilon}{2}\|x_i\|\le l(\beta;z_i)\le l(\beta_n;z_i)+\frac{\epsilon}{2}\|x_i\|,
	\end{align*}
	and
	\begin{align*}
	\mathbb{P}_ml(\beta)-\mathbb{E}\left\{\mathbb{P}_ml(\beta)\right\}\le& \mathbb{P}_ml(\beta_n)+\frac{\epsilon}{2m}\sum_{i=1}^m\|x_i\|-\left[\mathbb{E}\left\{\mathbb{P}_ml(\beta_n)\right\}-\frac{\epsilon}{2m}\sum_{i=1}^m\mathbb{E}\left(\|x_i\|\right)\right]\\
	=&\mathbb{P}_ml(\beta_n)-\mathbb{E}\left\{\mathbb{P}_ml(\beta_n)\right\}+\frac{\epsilon}{2m}\sum_{i=1}^m\left\{\|x_i\|+\mathbb{E}\left(\|x_i\|\right)\right\},\\
	\mathbb{E}\left[\mathbb{P}_ml(\beta)\right]-\mathbb{P}_ml(\beta)\le&
	\mathbb{E}\left\{\mathbb{P}_ml(\beta_n)\right\}+\frac{\epsilon}{2m}\sum_{i=1}^m\mathbb{E}(\|x_i\|)-\left(\mathbb{P}_ml(\beta_n)-\frac{\epsilon}{2m}\sum_{i=1}^m\|x_i\|\right)\\
	=&\mathbb{E}\left\{\mathbb{P}_ml(\beta_n)\right\}-\mathbb{P}_ml(\beta_n)+\frac{\epsilon}{2m}\sum_{i=1}^m\left\{\|x_i\|+\mathbb{E}\left(\|x_i\|\right)\right\}.
	\end{align*}
	Then we have
	\begin{align*}
	\sup_{\beta\in\mathcal{B}}\left|\mathbb{P}_ml(\beta)-\mathbb{E}\left\{\mathbb{P}_ml(\beta)\right\}\right|\le
	&\max_{1\le n\le N}\left|\mathbb{P}_ml(\beta_n)-\mathbb{E}\left\{\mathbb{P}_ml(\beta_n)\right\}\right|+\frac{\epsilon}{2m}\sum_{i=1}^m\left\{\|x_i\|+\mathbb{E}\left(\|x_i\|\right)\right\}\\
	=&\max_{1\le n\le N}\left|\mathbb{P}_ml(\beta_n)-\mathbb{E}\left\{\mathbb{P}_ml(\beta_n)\right\}\right|+\frac{\epsilon}{2m}\sum_{i=1}^m\left\{\|x_i\|-\mathbb{E}\left(\|x_i\|\right)\right\}+\frac{\epsilon}{m}\sum_{i=1}^m\mathbb{E}\left(\|x_i\|\right)\\
	=&o_{\mathbb{P}}(1)+O(\epsilon),
	\end{align*}
	where we have used the fact that $l(\beta;z)$ is bounded, Assumption 2 (which ensures that $x_i$'s are independent), the assumption that $\sup_{1\le i\le m}\mathbb{E}(\|x_i\|^2)<\infty$ and the Chebyshev's inequality. As $\epsilon$ can be arbitrarily small, under Assumption 3, we have
	\begin{align*}
	\sup_{\beta\in \mathcal{B}}|\mathbb{P}_ml(\beta)-\mathcal{L}(\beta)|\le\sup_{\beta\in \mathcal{B}}|\mathbb{P}_ml(\beta)-\mathbb{E}\left\{\mathbb{P}_ml(\beta)\right\}|+\sup_{\beta\in \mathcal{B}}|\mathbb{E}\left\{\mathbb{P}_ml(\beta)\right\}-\mathcal{L}(\beta)|=o_{\mathbb{P}}(1).
	\end{align*}
	~
\end{proof}

\begin{lemma}\label{lem_consistent} 
	Under Assumptions 2--4 and the assumption that $\sup_{1\le i\le m}\mathbb{E}(\|x_i\|^2)<\infty$, we have $\hat\beta \to\beta^{*}$ in probability. Moreover, for $a=0,1$,
	$$\sup_{1\leq j\leq m}\left|\hat\beta(p_{j\to a})- \beta^*\right|\to0,\quad\text{in probability.}$$
\end{lemma}

\begin{proof}
	By Assumptions 3 and 4, for every $\epsilon>0$, we know that there exists a $\delta>0$ such that $\mathcal{L}(\beta)<\mathcal{L}(\beta^*)-\delta$ for $\|\beta-\beta^{*}\|>\epsilon$. Therefore
	\begin{align*}
	\mathbb{P}\left(\|\hat\beta-\beta^*\|>\epsilon\right)&\le \mathbb{P}\left\{\mathcal{L}(\beta^*)-\mathcal{L}(\hat\beta)>\delta\right\}.
	\end{align*}
	To prove $\hat\beta \to\beta^{*}$ in probability, it suffices to show that $\mathcal{L}(\beta^*)-\mathcal{L}(\hat\beta)= o_\mathbb{P}(1)$. To this end, we note that $\mathcal{L}(\beta^*)\ge\mathcal{L}(\hat\beta)$ and
	\begin{align*}
	\mathcal{L}(\beta^*)-\mathcal{L}(\hat\beta)&=\mathcal{L}(\beta^*)-\mathbb{P}_ml(\beta^*)+\mathbb{P}_ml(\beta^*)-\mathbb{P}_ml(\hat\beta)+\mathbb{P}_ml(\hat\beta)-\mathcal{L}(\hat\beta)\\
	&\le \mathcal{L}(\beta^*)-\mathbb{P}_ml(\beta^*)+\sup_{\beta\in \mathcal{B}}|\mathbb{P}_ml(\beta)-\mathcal{L}(\beta)|= o_{\mathbb{P}}(1),
	\end{align*}
	where the inequality follows from the fact that $\hat\beta$ is the maximizer of $\mathbb{P}_ml(\beta)$ and hence $\mathbb{P}_ml(\beta^*)-\mathbb{P}_ml(\hat\beta)\le 0$ and the last equality follows from Lemma \ref{lem_ulln}. To prove the uniform convergence of $\hat\beta(p_{j\to a})$, let
	\begin{align*}
	\mathbb{P}^{j\to a}_ml(\beta)=&\frac{1}{m}\sum_{i\neq j}l(\beta;z_i)+\frac{1}{m}l\{\beta;(x_j,a)\}
	=\mathbb{P}_ml(\beta)+\frac{1}{m}\left[l\{\beta;(x_j,a)\}-l(\beta;z_j)\right],
	\end{align*}
	and $L=L_2-L_1$ with $L_1$ and $L_2$ defined in the proof of Lemma \ref{lem_ulln}. We have
	\begin{align*}
	\sup_{1\le j\le m}\sup_{\beta\in\mathcal{B}}\left|\mathbb{P}^{j\to a}_ml(\beta)-\mathbb{P}_ml(\beta)\right|\le\frac{L}{m}.
	\end{align*}
	Note that
	\begin{align*}
	\mathbb{P}\left\{\sup_{1\le j\le m}\left|\hat\beta(p_{j\to a})- \beta^*\right|>\epsilon\right\}
	\le& \mathbb{P}\left[\inf_{1\le j\le m}\mathcal{L}\{\hat\beta(p_{j\to a})\}<\mathcal{L}(\beta^*)-\delta\right]
	\\=&\mathbb{P}\left(\sup_{1\le j\le m}\left[\mathcal{L}(\beta^*)-\mathcal{L}\{\hat\beta(p_{j\to a})\}\right]>\delta\right).
	\end{align*}
	Then we only need to prove that $\sup_{1\le j\le m}[\mathcal{L}(\beta^*)-\mathcal{L}\{\hat\beta(p_{j\to a})\}]= o_\mathbb{P}(1).$ The proof is completed by noting that $\sup_{1\le j\le m}[\mathcal{L}(\beta^*)-\mathcal{L}\{\hat\beta(p_{j\to a})\}]\ge 0$ and
	\begin{align*}
	&\sup_{1\le j\le m}\left[\mathcal{L}(\beta^*)-\mathcal{L}\{\hat\beta(p_{j\to a})\}\right]\\
	=&\sup_{1\le j\le m}\bigg[\mathcal{L}(\beta^*)-\mathbb{P}_ml(\beta^*)+\mathbb{P}_ml(\beta^*)-\mathbb{P}^{j\to a}_ml(\beta^*)+\mathbb{P}^{j\to a}_ml(\beta^*)-\mathbb{P}^{j\to a}_ml\{\hat\beta(p_{j\to a})\}\\
	&\qquad\qquad+\mathbb{P}^{j\to a}_ml\{\hat\beta(p_{j\to a})\}-\mathbb{P}_ml\{\hat\beta(p_{j\to a})\}+\mathbb{P}_ml\{\hat\beta(p_{j\to a})\}-\mathcal{L}\{\hat\beta(p_{j\to a})\}\bigg]\\
	\le&\mathcal{L}(\beta^*)-\mathbb{P}_ml(\beta^*)+\frac{2L}{m}+\sup_{\beta\in\mathcal{B}}|\mathbb{P}_ml(\beta)-\mathcal{L}(\beta)|=o_\mathbb{P}(1),
	\end{align*}
	where we have used the fact that $\hat\beta(p_{j\to a})$ is the maximizer of $\mathbb{P}^{j\to a}_ml(\beta)$ and hence $\sup_{1\le j\le m}[\mathbb{P}_m^{j\to a}l(\beta^*)-\mathbb{P}^{j\to a}_ml\{\hat\beta(p_{j\to a})\}]\le 0$ and the result from Lemma \ref{lem_ulln}.
\end{proof}

\section{Proofs of Propositions 1 and 2}\label{sec-props12}
\begin{proof}[Proof of Proposition 1]
	We recall some notations defined in the main text and give some new definitions. Let $p_{j\rightarrow a}=(p_1,\dots,p_{j-1},a,p_{j+1},\dots,p_m)\in\mathbb{R}^m$ for $a=0,1$, and $p_{-j}=(p_1,...,p_{j-1},p_{j+1},...,p_m)\in\mathbb{R}^{m-1}.$ We define $\hat\beta(p_{j\to a}), \hat\pi(x_i;p_{j\to a}),\tilde\tau(p_{j\to a}),\hat\tau(p_{j\to a})$ and $\hat t_i(p_{j\to a})$ by setting the $j$th p-value to be equal to $a$ when estimating the corresponding quantities. First, we prove the first inequality. Observe that
	\begin{align*}
	\text{FWER}\leq&  \sum_{j\in\mathcal{M}_0}\mathbb{P}\left(p_j\leq\hat{t}_j\wedge \gamma \right)
	\\ \le & \sum_{j\in\mathcal{M}_0}\mathbb{E}\left\{\mathbb{I}\left(p_j\leq\hat{t}_j\wedge \gamma \right)-\frac{\mathbb{I}(p_j>\gamma)}{1-\gamma}\hat{t}_j\right\}
	+\sum_{i=1}^m \mathbb{E}\left\{\frac{\mathbb{I}(p_i>\gamma)}{1-\gamma}\hat{t}_i\right\}\\
	\le&  \sum_{j\in\mathcal{M}_0}\mathbb{E}\left\{\mathbb{I}\left(p_j\leq\hat{t}_j\wedge \gamma \right)-\frac{\mathbb{I}(p_j>\gamma)}{1-\gamma}\hat{t}_j\right\}+\alpha,
	\end{align*}
	where we have used the fact that $\sum_{i=1}^m\mathbb{I}(p_i>\gamma)\hat{t}_i/(1-\gamma)\leq \alpha.$ Denote by $\mathbb{E}_0$ and $\mathbb{P}_0$ the expectation and probability under the null. If $\{p_i\}\in\mathcal{M}_0$ are mutually independent and are independent with the non-null p-values, then under Assumption 1, we have
	\begin{align*}
	&\mathbb{E}_0\left\{\mathbb{I}\left(p_j\leq\hat{t}_j\wedge \gamma \right)-\frac{\mathbb{I}(p_j>\gamma)}{1-\gamma}\hat{t}_j\mid p_{-j}\right\}\\
	=&\mathbb{E}_0\left\{\mathbb{I}\left(p_j\leq\hat{t}_j\wedge \gamma \right)-\frac{\mathbb{I}(p_j>\gamma)}{1-\gamma}\hat{t}_j\mid p_{-j},p_j>\gamma\right\}\mathbb{P}_0(p_j>\gamma)\\
	&+\mathbb{E}_0\left\{\mathbb{I}\left(p_j\leq\hat{t}_j\wedge \gamma \right)-\frac{\mathbb{I}(p_j>\gamma)}{1-\gamma}\hat{t}_j\mid p_{-j},p_j\le\gamma\right\}\mathbb{P}_0(p_j\le\gamma)\\
	=&\left\{0-\frac{1}{1-\gamma}\hat t_j(p_{j\to 1})\right\}\{1-F_{0j}(
	\gamma)\}+\left[\frac{F_{0j}\{\hat{t}_j(p_{j\to 0})\wedge \gamma\}}{F_{0j}(\gamma)}-0\right]F_{0j}(\gamma)\\
	=&F_{0j}\{\hat{t}_j(p_{j\to 0})\wedge \gamma\}-\frac{1-F_{0j}(
		\gamma)}{1-\gamma}\hat t_j(p_{j\to 1})\\
	\le&\hat{t}_j(p_{j\to 0})\wedge \gamma-\hat t_j(p_{j\to 1}) \le\hat{t}_j(p_{j\to 0})-\hat t_j(p_{j\to 1}),
	\end{align*}
	where the first inequality is due to the assumption that $F_{0j}(t)\leq t.$ Thus, we have
	\begin{align*}
	\text{FWER}
	&\le  \sum_{j\in\mathcal{M}_0}\mathbb{E}\left[\mathbb{E}_0\left\{\mathbb{I}\left(p_j\leq\hat{t}_j\wedge \gamma \right)-\frac{\mathbb{I}(p_j>\gamma)}{1-\gamma}\hat{t}_j\mid p_{-j}\right\}\right] +\alpha \le J_m+\alpha,
	\end{align*}
	where
	\begin{align*}
	J_m=\sum_{j=1}^m\mathbb{E}\left\{\left|\hat{t}_j(p_{j\rightarrow 0})
	-\hat{t}_j(p_{j\rightarrow 1})\right|\right\}.
	\end{align*}
	For the second inequality, 	note that
	\begin{align*}
	\left|\frac{\hat{t}_j(p_{j\rightarrow 0})
		-\hat{t}_j(p_{j\rightarrow 1})}{k^{1/(1-k)}}\right|
	=&\left|\left\{\frac{1-\hat{\pi}(x_j;p_{j\rightarrow 0})}{\hat{\pi}(x_j;p_{j\rightarrow 0})\hat{\tau}(p_{j\rightarrow 0})}\right\}^{1/(1-k)}-
	\left\{\frac{1-\hat{\pi}(x_j;p_{j\rightarrow 1})}{\hat{\pi}(x_j;p_{j\rightarrow 1})\hat{\tau}(p_{j\rightarrow 1})}\right\}^{1/(1-k)}\right| \nonumber
	\\\le&\left|\left\{\frac{1-\hat{\pi}(x_j;p_{j\rightarrow 0})}{\hat{\pi}(x_j;p_{j\rightarrow 0})\hat{\tau}(p_{j\rightarrow 0})}\right\}^{1/(1-k)}-
	\left\{\frac{1-\hat{\pi}(x_j;p_{j\rightarrow 1})}{\hat{\pi}(x_j;p_{j\rightarrow 1})\hat{\tau}(p_{j\rightarrow 0})}\right\}^{1/(1-k)}\right| \nonumber
	\\&+
	\left|\left\{\frac{1-\hat{\pi}(x_j;p_{j\rightarrow 1})}{\hat{\pi}(x_j;p_{j\rightarrow 1})\hat{\tau}(p_{j\rightarrow 0})}\right\}^{1/(1-k)}-
	\left\{\frac{1-\hat{\pi}(x_j;p_{j\rightarrow 1})}{\hat{\pi}(x_j;p_{j\rightarrow 1})\hat{\tau}(p_{j\rightarrow 1})}\right\}^{1/(1-k)}\right| \nonumber
	\\ \leq & \frac{c|\hat{\pi}(x_j;p_{j\rightarrow 0})-\hat{\pi}(x_j;p_{j\rightarrow 1})|}{\hat{\tau}(p_{j\rightarrow 0})^{1/(1-k)}}+c\left|\frac{1}{\hat{\tau}(p_{j\rightarrow 0})^{1/(1-k)}}-\frac{1}{\hat{\tau}(p_{j\rightarrow 1})^{1/(1-k)}}\right|, 
	\end{align*}
	where the last inequality follows by using the mean-value theorem for the function $f(x)=(1/x-1)^{1/(1-k)}$ and the fact that $\hat{\pi}$ is bounded from below by $\varepsilon_1.$ Thus we have
	\begin{align*}
	J_m\le c(I_{m,1}+I_{m,2}),
	\end{align*}
	where
	\begin{align*}
	&I_{m,1}=\mathbb{E}\left\{\sum_{j=1}^m \frac{|\hat{\pi}(x_j;p_{j\rightarrow 0})-\hat{\pi}(x_j;p_{j\rightarrow 1})|}{\hat{\tau}(p_{j\rightarrow 0})^{1/(1-k)}}\right\},
	\\&I_{m,2}=\mathbb{E}\left\{\sum_{j=1}^m\left|\frac{1}{\hat{\tau}(p_{j\rightarrow 0})^{1/(1-k)}}-\frac{1}{\hat{\tau}(p_{j\rightarrow 1})^{1/(1-k)}}\right|\right\}.
	\end{align*}
	Next we derive upper bounds for $I_{m,1}$ and $I_{m,2}$. To deal with $I_{m,1}$, we note that for any $1\leq i,j\leq m$,
	\begin{align*}
	|\hat{\pi}(x_i;p_{j\rightarrow 0})-\hat{\pi}(x_i;p_{j\rightarrow 1})|&=\left|\left\{\frac{1}{1+e^{-x_i^\T \hat\beta(p_{j\to 0})}}\vee \varepsilon_1\right\}\wedge \varepsilon_2-\left\{\frac{1}{1+e^{-x_i^\T \hat\beta(p_{j\to1})}}\vee \varepsilon_1\right\}\wedge \varepsilon_2\right|\\
	&\le\left|\frac{1}{1+e^{-x_i^\T \hat\beta(p_{j\to 0})}}-\frac{1}{1+e^{-x_i^\T \hat\beta(p_{j\to1})}}\right|\\
	&\le  |x_i^\T\{\hat\beta(p_{j\rightarrow 0})-\hat\beta(p_{j\rightarrow 1})\}|,
	\end{align*}
	where the first inequality is due to the Lipschitz continuity of the function $f(x)=(x\vee \epsilon_1)\wedge \epsilon_2$, and the last inequality follows from an application of the mean-value theorem to the function $f(x)=(1+e^{-x})^{-1} $. For the ease of notation, set $\hat{b}(x_i)=\left[\{1-\hat{\pi}(x_i)\}/\hat{\pi}(x_i)\right]^{1/(1-k)}$.
	As $\varepsilon_1\le\hat\pi\le\varepsilon_2$, we have $b_1\le\hat b\le b_2$ for some constants $b_1$ and $b_2$ with $0<b_1\leq b_2$. It is straightforward to see that
	\begin{align*}
	\begin{split}
	\frac{|\hat{\pi}(x_j;p_{j\rightarrow 0})-\hat{\pi}(x_j;p_{j\rightarrow 1})|}{\hat{\tau}(p_{j\rightarrow 0})^{1/(1-k)}}
	\le&\frac{|x_j^\T\{\hat\beta(p_{j\rightarrow 0})-\hat\beta(p_{j\rightarrow 1})\}|}{\left\{k^{1/(1-k)}(1-\gamma)^{-1}\alpha^{-1}\sum_{i\neq j}\mathbb{I}(p_i>\gamma)\hat{b}(x_i;p_{j\rightarrow 0})\right\}\vee\varepsilon^{1/(1-k)}}\\
	\le &\frac{|x_j^\T\{\hat\beta(p_{j\rightarrow 0})-\hat\beta(p_{j\rightarrow 1})\}|}{\left\{c\alpha^{-1}\sum_{i\neq j}\mathbb{I}(p_i>\gamma)\right\}\vee\varepsilon^{1/(1-k)}},
	\end{split}
	\end{align*}	
	For $I_{m,2}$, we notice that
	\begin{align*}
	\begin{split}
	&\left|\frac{1}{\hat{\tau}(p_{j\rightarrow 0})^{1/(1-k)}}-\frac{1}{\hat{\tau}(p_{j\rightarrow 1})^{1/(1-k)}}\right|
	\le\frac{\left|\tilde{\tau}(p_{j\rightarrow 1})^{1/(1-k)}-\tilde{\tau}(p_{j\rightarrow 0})^{1/(1-k)}\right|}{\left\{\tilde{\tau}(p_{j\rightarrow 0})^{1/(1-k)}\vee\varepsilon^{1/(1-k)}\right\}\left\{\tilde{\tau}(p_{j\rightarrow 1})^{1/(1-k)}\vee\varepsilon^{1/(1-k)}\right\}}\\
	=&\frac{k^{1/(1-k)}(1-\gamma)^{-1}\alpha^{-1}\left|\sum_{i\neq j}\mathbb{I}(p_i>\gamma)\{\hat{b}(x_i;p_{j\rightarrow 1})-\hat{b}(x_i;p_{j\rightarrow 0})\}+\hat{b}(x_j;p_{j\rightarrow 1})\right|}{\left\{\tilde{\tau}(p_{j\rightarrow 0})^{1/(1-k)}\vee\varepsilon^{1/(1-k)}\right\}\left\{\tilde{\tau}(p_{j\rightarrow 1})^{1/(1-k)}\vee\varepsilon^{1/(1-k)}\right\}}\\
	\le &\frac{c\alpha^{-1}\sum_{i\neq j}\mathbb{I}(p_i>\gamma)|x_i^\T\{\hat\beta(p_{j\rightarrow 0})-\hat\beta(p_{j\rightarrow 1})\}|+c	\alpha^{-1}}{\left[\left\{c\alpha^{-1}\sum_{i\neq j}\mathbb{I}(p_i>\gamma)\right\}\vee\varepsilon^{1/(1-k)}\right]^2},
	\end{split}
	\end{align*}
	where we used the fact that $|\hat{b}(x_i;p_{j\rightarrow 1})-\hat{b}(x_i;p_{j\rightarrow 0})|\leq c|x_i^\T\{\hat\beta(p_{j\rightarrow 0})-\hat\beta(p_{j\rightarrow 1})\}|$ which follows from the mean-value theorem. Summarizing the above results, we have
	\begin{align*}
	J_m\le c(I_{m,1}+I_{m,2})\le c(J_{m,1}+J_{m,2}),
	\end{align*}
	where
	\begin{align*}
	J_{m,1}=&\sum_{j=1}^m\mathbb{E}\left[\frac{|x_j^\T\{\hat\beta(p_{j\rightarrow 0})-\hat\beta(p_{j\rightarrow 1})\}|}{\left\{c\alpha^{-1}\sum_{i\neq j}\mathbb{I}(p_i>\gamma)\right\}\vee\varepsilon^{1/(1-k)}}\right],\\
	J_{m,2}=&\sum_{j=1}^m\mathbb{E}\left(\frac{\alpha^{-1}\sum_{i\neq j}\mathbb{I}(p_i>\gamma)|x_i^\T\{\hat\beta(p_{j\rightarrow 0})-\hat\beta(p_{j\rightarrow 1})\}|+\alpha^{-1}}{\left[\left\{c\alpha^{-1}\sum_{i\neq j}\mathbb{I}(p_i>\gamma)\right\}\vee\varepsilon^{1/(1-k)}\right]^2}\right).
	\end{align*}
	~
\end{proof}

\begin{proof}[Proof of Proposition 2]
	We divide the proof into four steps. (i) Calculate the difference between the two estimating equations (EE) associated with $\hat\beta(p_{j\rightarrow 1})$ and $\hat\beta(p_{j\rightarrow 0})$; (ii) Perform Taylor expansions to extract the leading terms in the EEs; (iii) Deduce an expansion for $\hat\beta(p_{j\rightarrow 0})-\hat\beta(p_{j\rightarrow 1})$ based on the results in steps (i)--(ii); (iv) Derive the order of the remainder terms involved in the expansion of $\hat\beta(p_{j\rightarrow 0})-\hat\beta(p_{j\rightarrow 1})$.
	
	(i) Recall that we have defined the quasi log-likelihood function
	\begin{align*}
	\sum^{m}_{i=1}\log\left[\pi(x_i)(1-\gamma)^{y_i}\gamma^{1-y_i}+\{1-\pi(x_i)\}(1-\gamma^k)^{y_i}\gamma^{k(1-y_i)}\right],\;\pi(x_i)=(1+e^{-x_i^\T \beta})^{-1}.
	\end{align*}
	For the purpose of analysis, we use the following notation and expression instead,
	\begin{align*}
	\sum^{m}_{i=1}\log\left[q_i(\beta)(1-\gamma)^{y_i}\gamma^{1-y_i}+\{1-q_i(\beta)\}(1-\gamma^k)^{y_i}\gamma^{k(1-y_i)}\right],\;q_i(\beta)=(1+e^{-x_i^\T \beta})^{-1}.
	\end{align*}
	The quasi-MLE $\hat\beta$ satisfies the estimating equation
	\begin{align*}
	\sum_{i=1}^m\frac{\left\{(1-\gamma)^{y_i}\gamma^{1-y_i}-(1-\gamma^k)^{y_i}\gamma^{k(1-y_i)}\right\}\nabla q_i(\hat\beta)}{q_i(\hat\beta)(1-\gamma)^{y_i}\gamma^{1-y_i}+\{1-q_i(\hat\beta)\}(1-\gamma^k)^{y_i}\gamma^{k(1-y_i)}}=0,
	\end{align*}
	where $\nabla q_i(\beta)=x_ie^{x_i^\T\beta}/(1+e^{x_i^\T\beta})^2$. Taking the difference between the estimating equations associated with $\hat\beta(p_{j\rightarrow 1})$ and $\hat\beta(p_{j\rightarrow 0})$, we obtain
	\begin{align*}
	&\sum_{i\neq j}\frac{\left\{(1-\gamma)^{y_i}\gamma^{1-y_i}-(1-\gamma^k)^{y_i}\gamma^{k(1-y_i)}\right\}\nabla q_i\{\hat\beta(p_{j\rightarrow 1})\}}{q_i\{\hat\beta(p_{j\rightarrow 1})\}(1-\gamma)^{y_i}\gamma^{1-y_i}+[1-q_i\{\hat\beta(p_{j\rightarrow 1})\}](1-\gamma^k)^{y_i}\gamma^{k(1-y_i)}}\\
	&-\sum_{i\neq j}\frac{\left\{(1-\gamma)^{y_i}\gamma^{1-y_i}-(1-\gamma^k)^{y_i}\gamma^{k(1-y_i)}\right\}\nabla q_i\{\hat\beta(p_{j\rightarrow 0})\}}{q_i\{\hat\beta(p_{j\rightarrow 0})\}(1-\gamma)^{y_i}\gamma^{1-y_i}+[1-q_i\{\hat\beta(p_{j\rightarrow 0})\}](1-\gamma^k)^{y_i}\gamma^{k(1-y_i)}}\\
	&+\frac{(\gamma^k-\gamma)\nabla q_j\{\hat\beta(p_{j\rightarrow 1})\}}{q_j\{\hat\beta(p_{j\rightarrow 1})\}(1-\gamma)+[1-q_j\{\hat\beta(p_{j\rightarrow 1})\}](1-\gamma^k)}-
	\frac{(\gamma-\gamma^k)\nabla q_j\{\hat\beta(p_{j\rightarrow 0})\}}{q_j\{\hat\beta(p_{j\rightarrow 0})\}\gamma+[1-q_j\{\hat\beta(p_{j\rightarrow 0})\}]\gamma^k}=0.
	\end{align*}
	For the ease of notation, let $b_{0i}=(1-\gamma)^{y_i}\gamma^{1-y_i}$ and $b_{1i}=(1-\gamma^k)^{y_i}\gamma^{k(1-y_i)}$.
	Further define
	\begin{align*}
	&Q_i\{\hat\beta(p_{j\to a})\}=q_i\{\hat\beta(p_{j\rightarrow a})\}b_{0i}+[1-q_i\{\hat\beta(p_{j\rightarrow a})\}]b_{1i},\quad \text{for }a=0,1 \text{ and } i\neq j,\\
	&Q_{j0}\{\hat\beta(p_{j\to 0})\}=q_j\{\hat\beta(p_{j\rightarrow 0})\}\gamma+[1-q_j\{\hat\beta(p_{j\rightarrow 0})\}]\gamma^k,\\
	&Q_{j1}\{\hat\beta(p_{j\to 1})\}=q_j\{\hat\beta(p_{j\rightarrow 1})\}(1-\gamma)+[1-q_j\{\hat\beta(p_{j\rightarrow 1})\}](1-\gamma^k).
	\end{align*}
	Then the above equation can be expressed as
	\begin{align}\label{eq_eed}
	\begin{split}
	&\sum_{i\neq j}\frac{(b_{0i}-b_{1i})\left[\nabla q_i\{\hat\beta(p_{j\rightarrow 1})\}Q_i\{\hat\beta(p_{j\to 0})\}-\nabla q_i\{\hat\beta(p_{j\rightarrow 0})\}Q_i\{\hat\beta(p_{j\to 1})\}\right]}{Q_i\{\hat\beta(p_{j\to 1})\}Q_i\{\hat\beta(p_{j\to 0})\}}+\\
	&(\gamma^k-\gamma)\frac{\nabla q_j\{\hat\beta(p_{j\rightarrow 1})\}Q_{j0}\{\hat\beta(p_{j\to 0})\}+\nabla q_j\{\hat\beta(p_{j\rightarrow 0})\}Q_{j1}\{\hat\beta(p_{j\to 1})\}}{Q_{j1}\{\hat\beta(p_{j\to 1})\}Q_{j0}\{\hat\beta(p_{j\to 0})\}}=0.
	\end{split}
	\end{align}
	
	(ii) We analyze the two terms in (\ref{eq_eed}). For the first term, we observe that
	\begin{align*}
	&\nabla q_i\{\hat\beta(p_{j\rightarrow 1})\}Q_i\{\hat\beta(p_{j\to 0})\}-\nabla q_i\{\hat\beta(p_{j\rightarrow 0})\}Q_i\{\hat\beta(p_{j\to 1})\}\\
	=&\nabla q_i\{\hat\beta(p_{j\rightarrow 1})\}Q_i\{\hat\beta(p_{j\to 0})\}-\nabla q_i\{\hat\beta(p_{j\rightarrow 1})\}Q_i\{\hat\beta(p_{j\to 1})\}+\nabla q_i\{\hat\beta(p_{j\rightarrow 1})\}Q_i\{\hat\beta(p_{j\to 1})\}
	\\&-\nabla q_i\{\hat\beta(p_{j\rightarrow 0})\}Q_i\{\hat\beta(p_{j\to 1})\}\\
	=&\nabla q_i\{\hat\beta(p_{j\rightarrow 1})\}\left[Q_i\{\hat\beta(p_{j\to 0})\}-Q_i\{\hat\beta(p_{j\to 1})\}\right]-Q_i\{\hat\beta(p_{j\to 1})\}\left[\nabla q_i\{\hat\beta(p_{j\rightarrow 0})\}
	-\nabla q_i\{\hat\beta(p_{j\rightarrow 1})\}\right].
	\end{align*}
	By the Taylor expansion, we have
	\begin{align*}
	\nabla q_i\{\hat\beta(p_{j\rightarrow 1})\}&=\nabla q_i(\beta^*)+\nabla^2q_i(\tilde\beta_{1})\{\hat\beta(p_{j\rightarrow 1})-\beta^*\},\\
	Q_i\{\hat\beta(p_{j\to 0})\}-Q_i\{\hat\beta(p_{j\to 1})\}&=(b_{0i}-b_{1i})\left[q_i\{\hat\beta(p_{j\rightarrow 0})\}-q_i\{\hat\beta(p_{j\rightarrow 1})\}\right]\\
	&=(b_{0i}-b_{1i})\left[\nabla q_i(\beta^*)^\T+\left\{\nabla q_i(\tilde\beta_{2})-\nabla q_i(\beta^*)\right\}^\T
	\right]\{\hat\beta(p_{j\rightarrow 0})-\hat\beta(p_{j\rightarrow 1})\}\\
	&=(b_{0i}-b_{1i})\left\{\nabla q_i(\beta^*)^\T+(\tilde\beta_{2}-\beta^*)^\T \nabla^2q_i(\tilde{\beta}_{3})^\T\right\}\{\hat\beta(p_{j\rightarrow 0})-\hat\beta(p_{j\rightarrow 1})\},
	\end{align*}
	where $\nabla^2 q_i(\beta)=x_ix_i^{\T}e^{x_i^\T\beta}(1-e^{x_i^{\T}\beta})/(1+e^{x_i^\T\beta})^3$, and $\tilde\beta_{1},\tilde\beta_{2}$ and $\tilde\beta_{3}$ are the corresponding intermediate values in the mean value theorem, which may vary with $i$ and $j$. We suppress the dependence on $i$ and $j$ for notational simplicity. Note that the mean value theorem does not generally hold for vector-valued function. Here we are applying the mean value theorem to the scalar-valued function $e^{x_i^\T\beta}/(1+e^{x_i^\T\beta})^2$. Using similar technique, we have
	\begin{align*}
	&Q_i\{\hat\beta(p_{j\to 1})\}
	=(b_{0i}-b_{1i})q_i\{\hat\beta(p_{j\rightarrow 1})\}+b_{1i}\\
	=&(b_{0i}-b_{1i})q_i(\beta^*)+b_{1i}+(b_{0i}-b_{1i})\nabla q_i(\bar{\beta}_1)^\T\{\hat\beta(p_{j\rightarrow 1})-\beta^*\}\\
	=&Q_i(\beta^*)+(b_{0i}-b_{1i})\nabla q_i(\bar{\beta}_1)^\T\{\hat\beta(p_{j\rightarrow 1})-\beta^*\}.
	\end{align*}
	Let $\phi_i(\beta)=e^{x_i^\T\beta}(1-e^{x_i^{\T}\beta})/(1+e^{x_i^\T\beta})^3$. Then $\nabla^2q_i(\beta)=\phi_i(\beta)x_ix_i^{\T}$ and $\nabla \phi_i(\beta)=x_ie^{x_i^\T\beta}(1-4e^{x_i^\T\beta}+e^{2x_i^\T\beta})(1+e^{x_i^\T\beta})^4$. Again by the mean value theorem, we have
	\begin{align*}
	&\nabla q_i\{\hat\beta(p_{j\rightarrow 0})\}-\nabla q_i\{\hat\beta(p_{j\rightarrow 1})\}\\
	=&\left[\nabla^2q_i(\beta^*)+\left\{\nabla^2q_i(\bar{\beta}_2)-\nabla^2q_i(\beta^*)\right\}\right]\{\hat\beta(p_{j\rightarrow 0})-\hat\beta(p_{j\rightarrow 1})\}\\
	=&\left[\nabla^2q_i(\beta^*)+\left\{\phi_i(\bar{\beta}_2)-\phi_i(\beta^*)\right\}x_ix_i^\T\right]\{\hat\beta(p_{j\rightarrow 0})-\hat\beta(p_{j\rightarrow 1})\}\\
	=&\left\{\nabla^2q_i(\beta^*)+\nabla \phi_i(\bar\beta_3)^\T(\bar{\beta}_2-\beta^*)x_ix_i^\T\right\}\{\hat\beta(p_{j\rightarrow 0})-\hat\beta(p_{j\rightarrow 1})\},
	\end{align*}
	where $\bar\beta_{1},\bar\beta_{2}$ and $\bar\beta_{3}$ are the corresponding intermediate values and dependent on $i$ and $j$. Using the above expansions, we deduce that
	\begin{align*}
	&\nabla q_i\{\hat\beta(p_{j\rightarrow 1})\}\left[Q_i\{\hat\beta(p_{j\to 0})\}-Q_i\{\hat\beta(p_{j\to 1})\}\right]\\
	=&(b_{0i}-b_{1i})\left[\nabla q_i(\beta^*)+\nabla^2q_i(\tilde\beta_{1})\{\hat\beta(p_{j\rightarrow 1})-\beta^*\}\right]\left\{\nabla q_i(\beta^*)^\T+(\tilde\beta_{2}-\beta^*)^\T \nabla^2q_i(\tilde{\beta}_{3})^\T\right\}\\
	&\{\hat\beta(p_{j\rightarrow 0})-\hat\beta(p_{j\rightarrow 1})\}\\
	=&(b_{0i}-b_{1i})\nabla q_i(\beta^*)\nabla q_i(\beta^*)^\T\{\hat\beta(p_{j\rightarrow 0})-\hat\beta(p_{j\rightarrow 1})\}\\
	&+(b_{0i}-b_{1i})\bigg[\nabla q_i(\beta^*)(\tilde\beta_{2}-\beta^*)^\T \nabla^2q_i(\tilde{\beta}_{3})^\T+\nabla^2q_i(\tilde\beta_{1})\{\hat\beta(p_{j\rightarrow 1})-\beta^*\}\\
	&\qquad\qquad\qquad\left\{\nabla q_i(\beta^*)^\T+(\tilde\beta_{2}-\beta^*)^\T \nabla^2q_i(\tilde{\beta}_{3})^\T\right\}\bigg]\{\hat\beta(p_{j\rightarrow 0})-\hat\beta(p_{j\rightarrow 1})\}\\
	=&(b_{0i}-b_{1i})\nabla q_i(\beta^*)\nabla q_i(\beta^*)^\T\{\hat\beta(p_{j\rightarrow 0})-\hat\beta(p_{j\rightarrow 1})\}+R_{ij}^{(1)}\{\hat\beta(p_{j\rightarrow 0})-\hat\beta(p_{j\rightarrow 1})\},
	\end{align*}
	and
	\begin{align*}
	&Q_i\{\hat\beta(p_{j\to 1})\}\left[\nabla q_i\{\hat\beta(p_{j\rightarrow 0})\}-\nabla q_i\{\hat\beta(p_{j\rightarrow 1})\}\right]\\
	=&\left[Q_i(\beta^*)+(b_{0i}-b_{1i})\nabla q_i(\bar{\beta}_1)^\T\{\hat\beta(p_{j\rightarrow 1})-\beta^*\}\right]\left\{\nabla^2q_i(\beta^*)+\nabla \phi_i(\bar{\beta}_3)^\T(\bar{\beta}_2-\beta^*)x_ix_i^\T\right\}\\
	&\{\hat\beta(p_{j\rightarrow 0})-\hat\beta(p_{j\rightarrow 1})\}\\
	=&Q_i(\beta^*)\nabla^2q_i(\beta^*)\{\hat\beta(p_{j\rightarrow 0})-\hat\beta(p_{j\rightarrow 1})\}\\
	&+\bigg[Q_i(\beta^*)\nabla \phi_i(\bar{\beta}_3)^\T(\bar{\beta}_2-\beta^*)x_ix_i^\T+(b_{0i}-b_{1i})\nabla q_i(\bar{\beta}_1)^\T\{\hat\beta(p_{j\rightarrow 1})-\beta^*\}\\
	&\qquad\left\{\nabla^2q_i(\beta^*)+\nabla \phi_i(\bar{\beta}_3)^\T(\bar{\beta}_2-\beta^*)x_ix_i^\T\right\}\bigg]\{\hat\beta(p_{j\rightarrow 0})-\hat\beta(p_{j\rightarrow 1})\}\\
	=&Q_i(\beta^*)\nabla^2q_i(\beta^*)\{\hat\beta(p_{j\rightarrow 0})-\hat\beta(p_{j\rightarrow 1})\}+R^{(2)}_{ij}\{\hat\beta(p_{j\rightarrow 0})-\hat\beta(p_{j\rightarrow 1})\},
	\end{align*}
	where
	\begin{align*}
	R_{ij}^{(1)}=&(b_{0i}-b_{1i})\bigg[\nabla q_i(\beta^*)(\tilde\beta_{2}-\beta^*)^\T \nabla^2q_i(\tilde{\beta}_{3})^\T\\
	&\qquad\qquad\quad+\nabla^2q_i(\tilde\beta_{1})\{\hat\beta(p_{j\rightarrow 1})-\beta^*\}\left\{\nabla q_i(\beta^*)^\T+(\tilde\beta_{2}-\beta^*)^\T \nabla^2q_i(\tilde{\beta}_{3})^\T\right\}\bigg],\\
	R^{(2)}_{ij}=&Q_i(\beta^*)\nabla \phi_i(\bar{\beta}_3)^\T(\bar{\beta}_2-\beta^*)x_ix_i^\T\\
	&+(b_{0i}-b_{1i})\nabla q_i(\bar{\beta}_1)^\T\{\hat\beta(p_{j\rightarrow 1})-\beta^*\}\left\{\nabla^2q_i(\beta^*)+\nabla \phi_i(\bar{\beta}_3)^\T(\bar{\beta}_2-\beta^*)x_ix_i^\T\right\}.
	\end{align*}
	To deal with the second term of (\ref{eq_eed}), we observe that
	\begin{align*}
	&\nabla q_j\{\hat\beta(p_{j\rightarrow 1})\}Q_{j0}\{\hat\beta(p_{j\to 0})\}\\
	=&\left[\nabla q_j(\beta^*)+\nabla^2q_j(\check\beta_{1})\{\hat\beta(p_{j\rightarrow 1})-\beta^*\}\right]\left[Q_{j0}(\beta^*)+(\gamma-\gamma^k)\nabla q_j(\check\beta_2)^\T\{\hat\beta(p_{j\to 0})-\beta^*\}\right]\\
	=&\nabla q_j(\beta^*)Q_{j0}(\beta^*)+(\gamma-\gamma^k)\nabla q_j(\beta^*)\nabla q_j(\check\beta_2)^\T\{\hat\beta(p_{j\to 0})-\beta^*\}\\
	&+\nabla^2q_j(\check\beta_{1})\{\hat\beta(p_{j\rightarrow 1})-\beta^*\}\left[Q_{j0}(\beta^*)+(\gamma-\gamma^k)\nabla q_j(\check\beta_2)^\T\{\hat\beta(p_{j\to 0})-\beta^*\}\right]\\
	=&\nabla q_j(\beta^*)Q_{j0}(\beta^*)+R^{(3)}_{j},
	\end{align*}
	and
	\begin{align*}
	&\nabla q_j\{\hat\beta(p_{j\rightarrow 0})\}Q_{j1}\{\hat\beta(p_{j\to 1})\}\\
	=&\left[\nabla q_j(\beta^*)+\nabla^2q_j(\check\beta_{3})\{\hat\beta(p_{j\rightarrow 0})-\beta^*\}\right]\left[Q_{j1}(\beta^*)+(\gamma^k-\gamma)\nabla q_j(\check\beta_4)^\T\{\hat\beta(p_{j\to 1})-\beta^*\}\right]\\
	=&\nabla q_j(\beta^*)Q_{j1}(\beta^*)+(\gamma^k-\gamma)\nabla q_j(\beta^*)\nabla q_j(\check\beta_4)^\T\{\hat\beta(p_{j\to 1})-\beta^*\}\\
	&+\nabla^2q_j(\check\beta_{3})\{\hat\beta(p_{j\rightarrow 0})-\beta^*\}\left[Q_{j1}(\beta^*)+(\gamma^k-\gamma)\nabla q_j(\check\beta_4)^\T\{\hat\beta(p_{j\to 1})-\beta^*\}\right]\\
	=&\nabla q_j(\beta^*)Q_{j1}(\beta^*)+R^{(4)}_{j},
	\end{align*}
	where
	\begin{align*}
	R_{j}^{(3)}=&(\gamma-\gamma^k)\nabla q_j(\beta^*)\nabla q_j(\check\beta_2)^\T\{\hat\beta(p_{j\to 0})-\beta^*\}\\
	&+\nabla^2q_j(\check\beta_{1})\{\hat\beta(p_{j\rightarrow 1})-\beta^*\}\left[Q_{j0}(\beta^*)+(\gamma-\gamma^k)\nabla q_j(\check\beta_2)^\T\{\hat\beta(p_{j\to 0})-\beta^*\}\right],\\
	R^{(4)}_{j}=&(\gamma^k-\gamma)\nabla q_j(\beta^*)\nabla q_j(\check\beta_4)^\T\{\hat\beta(p_{j\to 1})-\beta^*\}\\
	&+\nabla^2q_j(\check\beta_{3})\{\hat\beta(p_{j\rightarrow 0})-\beta^*\}\left[Q_{j1}(\beta^*)+(\gamma^k-\gamma)\nabla q_j(\check\beta_4)^\T\{\hat\beta(p_{j\to 1})-\beta^*\}\right],
	\end{align*}
	and $\check\beta_{1},\check\beta_{2},\check\beta_{3},\check\beta_{4}$ are the corresponding intermediate values in the mean value theorem which vary with $i$ and $j$.
	
	(iii) Let $v^{\otimes 2}=vv^\T$ for a vector $v$. Plugging the equations we obtain in step (ii) into (\ref{eq_eed}), we get
	\begin{align}\label{eq_taylor}
	\begin{split}
	&\sum_{i\neq j}\frac{(b_{0i}-b_{1i})\left\{(b_{0i}-b_{1i})\nabla q_i(\beta^*)^{\otimes 2}-Q_i(\beta^*)\nabla^2q_i(\beta^*)\right\}\{\hat\beta(p_{j\rightarrow 0})-\hat\beta(p_{j\rightarrow 1})\}}{Q_i\{\hat\beta(p_{j\to 1})\}Q_i\{\hat\beta(p_{j\to 0})\}}\\
	&+\sum_{i\neq j}\frac{(b_{0i}-b_{1i})(R_{ij}^{(1)}-R_{ij}^{(2)})\{\hat\beta(p_{j\rightarrow 0})-\hat\beta(p_{j\rightarrow 1})\}}{Q_i\{\hat\beta(p_{j\to 1})\}Q_i\{\hat\beta(p_{j\to 0})\}}\\
	&+(\gamma^k-\gamma)\frac{\nabla q_j(\beta^*)Q_{j0}(\beta^*)+\nabla q_j(\beta^*)Q_{j1}(\beta^*)}{Q_{j1}\{\hat\beta(p_{j\to 1})\}Q_{j0}\{\hat\beta(p_{j\to 0})\}}+(\gamma^k-\gamma)\frac{R_{j}^{(3)}+R_{j}^{(4)}}{Q_{j1}\{\hat\beta(p_{j\to 1})\}Q_{j0}\{\hat\beta(p_{j\to 0})\}}=0.
	\end{split}
	\end{align}
	Let
	\begin{align*}
	S_j&=\sum_{i\neq j}\frac{(b_{0i}-b_{1i})\left\{(b_{0i}-b_{1i})\nabla q_i(\beta^*)^{\otimes 2}-Q_i(\beta^*)\nabla^2q_i(\beta^*)\right\}}{Q_i\{\hat\beta(p_{j\to 1})\}Q_i\{\hat\beta(p_{j\to 0})\}},\\
	\widetilde S_j&=\sum_{i\neq j}\frac{(b_{0i}-b_{1i})(R_{ij}^{(1)}-R_{ij}^{(2)})}{Q_i\{\hat\beta(p_{j\to 1})\}Q_i\{\hat\beta(p_{j\to 0})\}},\\
	U_j&=(\gamma^k-\gamma)\frac{\nabla q_j(\beta^*)Q_{j0}(\beta^*)+\nabla q_j(\beta^*)Q_{j1}(\beta^*)}{Q_{j1}\{\hat\beta(p_{j\to 1})\}Q_{j0}\{\hat\beta(p_{j\to 0})\}},\\
	\widetilde U_j&=(\gamma^k-\gamma)\frac{R_{j}^{(3)}+R_{j}^{(4)}}{Q_{j1}\{\hat\beta(p_{j\to 1})\}Q_{j0}\{\hat\beta(p_{j\to 0})\}}.
	\end{align*}
	Then (\ref{eq_taylor}) can be written compactly as
	\begin{align*}
	(S_j+\widetilde S_j)\{\hat\beta(p_{j\rightarrow 0})-\hat\beta(p_{j\rightarrow 1})\}+U_j+\widetilde U_j=0.
	\end{align*}
	Further define
	\begin{align*}
	&S^{*}_j=\sum_{i\neq j}\frac{(b_{0i}-b_{1i})\left\{(b_{0i}-b_{1i})\nabla q_i(\beta^*)^{\otimes 2}-Q_i(\beta^*)\nabla^2q_i(\beta^*)\right\}}{Q_i(\beta^*)^2},\\
	&U^{*}_j=(\gamma^k-\gamma)\frac{\nabla q_j(\beta^*)Q_{j0}(\beta^*)+\nabla q_j(\beta^*)Q_{j1}(\beta^*)}{Q_{j1}(\beta^*)Q_{j0}(\beta^*)}.
	\end{align*}
	Note that
	\begin{align*}
	&Q_i\{\hat\beta(p_{j\to 1})\}=
	Q_i(\beta^*)+(b_{0i}-b_{1i})\nabla q_i(\bar{\beta}_1)^\T\{\hat\beta(p_{j\rightarrow 1})-\beta^*\},\\
	&Q_i\{\hat\beta(p_{j\to 0})\}=
	Q_i(\beta^*)+(b_{0i}-b_{1i})\nabla q_i(\bar{\beta}_0)^\T\{\hat\beta(p_{j\rightarrow 0})-\beta^*\},\\
	&Q_{j1}\{\hat\beta(p_{j\to 1})\}=Q_{j1}(\beta^*)+(\gamma^k-\gamma)\nabla q_j(\check\beta_4)^\T\{\hat\beta(p_{j\to 1})-\beta^*\},\\
	&Q_{j0}\{\hat\beta(p_{j\to 0})\}=Q_{j0}(\beta^*)+(\gamma-\gamma^k)\nabla q_j(\check\beta_2)^\T\{\hat\beta(p_{j\to 0})-\beta^*\},
	\end{align*}
	where $\bar\beta_{0}$ is the corresponding intermediate value that depends on $i$ and $j$, and $\bar\beta_1,\check\beta_4,\check\beta_2$ have been defined before. Thus we obtain
	\begin{align*}
	R_{ij}^{(5)}:=&Q_i(\beta^*)^2-Q_i\{\hat\beta(p_{j\to 1})\}Q_i\{\hat\beta(p_{j\to 0})\}\\
	=&-Q_i(\beta^*)(b_{0i}-b_{1i})\nabla q_i(\bar{\beta}_0)^\T\{\hat\beta(p_{j\rightarrow 0})-\beta^*\}-Q_i(\beta^*)(b_{0i}-b_{1i})\nabla q_i(\bar{\beta}_1)^\T\{\hat\beta(p_{j\rightarrow 1})-\beta^*\}\\
	&-(b_{0i}-b_{1i})^2\nabla q_i(\bar{\beta}_1)^\T\{\hat\beta(p_{j\rightarrow 1})-\beta^*\}\nabla q_i(\bar{\beta}_0)^\T\{\hat\beta(p_{j\rightarrow 0})-\beta^*\},
	\end{align*}
	and
	\begin{align*}
	R_{j}^{(6)}:=&Q_{j1}(\beta^*)Q_{j0}(\beta^*)-Q_{j1}\{\hat\beta(p_{j\to 1})\}Q_{j0}\{\hat\beta(p_{j\to 0})\}\\
	=&-Q_{j1}(\beta^*)(\gamma-\gamma^k)\nabla q_j(\check\beta_2)^\T\{\hat\beta(p_{j\to 0})-\beta^*\}+Q_{j0}(\beta^*)(\gamma-\gamma^k)\nabla q_j(\check\beta_4)^\T\{\hat\beta(p_{j\to 1})-\beta^*\}\\
	&+(\gamma-\gamma^k)^2\nabla q_j(\check\beta_4)^\T\{\hat\beta(p_{j\to 1})-\beta^*\}\nabla q_j(\check\beta_2)^\T\{\hat\beta(p_{j\to 0})-\beta^*\}.
	\end{align*}
	It implies that
	\begin{align*}
	&S_j-S_j^*=\sum_{i\neq j}\frac{(b_{0i}-b_{1i})\left\{(b_{0i}-b_{1i})\nabla q_i(\beta^*)^{\otimes 2}-Q_i(\beta^*)\nabla^2q_i(\beta^*)\right\}R_{ij}^{(5)}}{Q_i\{\hat\beta(p_{j\to 1})\}Q_i\{\hat\beta(p_{j\to 0})\}Q_i(\beta^*)^2},\\
	&U_j-U_j^*=(\gamma^k-\gamma)\frac{\nabla q_j(\beta^*)Q_{j0}(\beta^*)+\nabla q_j(\beta^*)Q_{j1}(\beta^*)}{Q_{j1}\{\hat\beta(p_{j\to 1})\}Q_{j0}\{\hat\beta(p_{j\to 0})\}Q_{j1}(\beta^*)Q_{j0}(\beta^*)}R_{j}^{(6)}.
	\end{align*}
	Combining the above arguments, we obtain
	\begin{align*}
	(S_j^*+\Delta_j)\{\hat\beta(p_{j\rightarrow 0})-\hat\beta(p_{j\rightarrow 1})\}+U_j^*+\Pi_j=0,
	\end{align*}
	where $\Delta_j=S_j-S_j^*+\widetilde S_j$ and $\Pi_j=U_j-U_j^*+\widetilde U_j$, and $\Delta_j, \Pi_j$ are smaller order terms.
	
	(iv) Recall that $q_i(\beta)=(1+e^{-x_i^\T \beta})^{-1}$, $b_{0i}=(1-\gamma)^{y_i}\gamma^{1-y_i}$,  $b_{1i}=(1-\gamma^k)^{y_i}\gamma^{k(1-y_i)}$ and $Q_i(\beta)=q_i(\beta)b_{0i}+\{1-q_i(\beta)\}b_{1i}$. We have uniformly over $i,j$ and $\beta\in\mathcal{B}$,
	\begin{align*}
	&|b_{0i}-b_{1i}|=\gamma^k-\gamma,\\
	&\min\{1-\gamma^k,\gamma\}\le Q_i(\beta)\le\max\{1-\gamma,\gamma^k\},\\
	&\gamma\le Q_{j0}(\beta)=q_j(\beta)\gamma+\{1-q_j(\beta)\}\gamma^k\le\gamma^k,\\
	&1-\gamma^k\le Q_{j1}(\beta)=q_j(\beta)(1-\gamma)+\{1-q_j(\beta)\}(1-\gamma^k)\le 1-\gamma,
	\end{align*}
	and
	\begin{align*}
	&\|\nabla q_i(\beta)\|=\left\|\frac{x_ie^{x_i^\T\beta}}{(1+e^{x_i^\T\beta})^2}\right\|\le c\|x_i\|,\\
	&\|\nabla^2q_i(\beta)\|=\left\|\frac{x_ix_i^{\T}e^{x_i^\T\beta}(1-e^{x_i^{\T}\beta})}{(1+e^{x_i^\T\beta})^3}\right\|\le c\|x_i\|^2,\\
	&\|\nabla \phi_i(\beta)\|=\left\|\frac{x_ie^{x_i^\T\beta}(1-4e^{x_i^\T\beta}+e^{2x_i^\T\beta})}{(1+e^{x_i^\T\beta})^4}\right\|\le c\|x_i\|.
	\end{align*}
	Then we have
	\begin{align*}
	\|S_j-S_j^*\|&\le\sum_{i\neq j}\left\|\frac{(b_{0i}-b_{1i})\left\{(b_{0i}-b_{1i})\nabla q_i(\beta^*)^{\otimes 2}-Q_i(\beta^*)\nabla^2q_i(\beta^*)\right\}R_{ij}^{(5)}}{Q_i\{\hat\beta(p_{j\to 1})\}Q_i\{\hat\beta(p_{j\to 0})\}Q_i(\beta^*)^2}\right\|
	\le c\sum_{i\neq j}\left\|x_i\right\|^2\cdot|R_{ij}^{(5)}|,
	\end{align*}
	and
	\begin{align*}
	\|\widetilde S_j\|&\le \sum_{i\neq j}\left\|\frac{(b_{0i}-b_{1i})(R_{ij}^{(1)}-R_{ij}^{(2)})}{Q_i\{\hat\beta(p_{j\to 1})\}Q_i\{\hat\beta(p_{j\to 0})\}}\right\|\le c\sum_{i\neq j}\left(\|R_{ij}^{(1)}\|+\|R_{ij}^{(2)}\|\right),
	\end{align*}
	where
	\begin{align*}
	&|R_{ij}^{(5)}|\le c\|x_i\|\left\{\|\hat\beta(p_{j\rightarrow 1})-\beta^*\|+\|\hat\beta(p_{j\rightarrow 0})-\beta^*\|\right\}+c\|x_i\|^2\|\hat\beta(p_{j\rightarrow 1})-\beta^*\|\|\hat\beta(p_{j\rightarrow 0})-\beta^*\|,\\
	&\|R_{ij}^{(1)}\|\le c\|x_i\|^3\left\{\|\tilde\beta_{2}-\beta^*\|+\|\hat\beta(p_{j\rightarrow 1})-\beta^*\|\right\}+c\|x_i\|^4\|\tilde\beta_{2}-\beta^*\|\|\hat\beta(p_{j\rightarrow 1})-\beta^*\|,\\
	&\|R_{ij}^{(2)}\|\le c\|x_i\|^3\left\{\|\bar\beta_{2}-\beta^*\|+\|\hat\beta(p_{j\rightarrow 1})-\beta^*\|\right\}+c\|x_i\|^4\|\bar\beta_{2}-\beta^*\|\|\hat\beta(p_{j\rightarrow 1})-\beta^*\|.
	\end{align*}
	Similarly, we have
	\begin{align*}
	\|U_j-U_j^*\|\le c\|x_j\|\|R_{j}^{(6)}\|\quad\text{and}\quad\|\widetilde U_j\|\le c\left(\|R_{j}^{(3)}\|+\|R_{j}^{(4)}\|\right),
	\end{align*}
	where
	\begin{align*}
	&\|R_{j}^{(6)}\|\le c\|x_j\|\left\{\|\hat\beta(p_{j\rightarrow 1})-\beta^*\|+\|\hat\beta(p_{j\rightarrow 0})-\beta^*\|\right\}+c\|x_j\|^2\|\hat\beta(p_{j\rightarrow 1})-\beta^*\|\|\hat\beta(p_{j\rightarrow 0})-\beta^*\|,\\
	&\|R_{j}^{(3)}\|\le c\|x_j\|^2\left\{\|\hat\beta(p_{j\rightarrow 1})-\beta^*\|+\|\hat\beta(p_{j\rightarrow 0})-\beta^*\|\right\}+c\|x_j\|^3\|\hat\beta(p_{j\rightarrow 1})-\beta^*\|\|\hat\beta(p_{j\rightarrow 0})-\beta^*\|,\\
	&\|R_{j}^{(4)}\|\le c\|x_j\|^2\left\{\|\hat\beta(p_{j\rightarrow 1})-\beta^*\|+\|\hat\beta(p_{j\rightarrow 0})-\beta^*\|\right\}+c\|x_j\|^3\|\hat\beta(p_{j\rightarrow 1})-\beta^*\|\|\hat\beta(p_{j\rightarrow 0})-\beta^*\|.
	\end{align*}
	As we assume that $\sup_{1\le i\le m}\mathbb{E}(\|x_i\|^{8})<\infty$, then by  Kolmogorov's strong law of large numbers, we know that $\sup_{1\le j\le m}\sum_{i\neq j}\|x_i\|^3=O_\mathbb{P}(m)$ and $\sup_{1\le j\le m}\sum_{i\neq j}\|x_i\|^4=O_\mathbb{P}(m)$.
	Combining the above inequalities with the result from Lemma \ref{lem_consistent} that $\sup_{1\leq j\leq m}|\hat\beta(p_{j\to a})- \beta^*|=o_\mathbb{P}(1)$ , we deduce that
	\begin{align}
	\begin{split}\label{ineq4}
	&S^{*}_j=\sum_{i\neq j}\frac{(b_{0i}-b_{1i})\left\{(b_{0i}-b_{1i})\nabla q_i(\beta^*)^{\otimes 2}-Q_i(\beta^*)\nabla^2q_i(\beta^*)\right\}}{Q_i(\beta^*)^2}=-\sum_{i\neq j}\nabla^2l(\beta^*;z_i),\\
	&\sup_{1\le j\le m}\|U^{*}_j\|=\sup_{1\le j\le m}\left\|(\gamma^k-\gamma)\frac{\nabla q_j(\beta^*)Q_{j0}(\beta^*)+\nabla q_j(\beta^*)Q_{j1}(\beta^*)}{Q_{j1}(\beta^*)Q_{j0}(\beta^*)}\right\|\le c\sup_{1\le j\le m}\|x_j\|=O_\mathbb{P}(1),
	\end{split}
	\end{align}
	and
	\begin{align}
	\begin{split}\label{ineq1}
	\sup_{1\le j\le m}\|\Delta_j\|\le& \sup_{1\le j\le m}\left(\|S_j-S_j^*\|+\|\widetilde S_j\|\right)
	\le  c\sup_{1\le j\le m}\sum_{i\neq j}\left(\left\|x_i\right\|^2\cdot|R_{ij}^{(5)}|+\|R_{ij}^{(1)}\|+\|R_{ij}^{(2)}\|\right)\\
	\le& c\sup_{1\le j\le m}\left\{\|\hat\beta(p_{j\rightarrow 1})-\beta^*\|+\|\hat\beta(p_{j\rightarrow 0})-\beta^*\|\right\}\sum_{i\neq j}\|x_i\|^3\\
	&+c\sup_{1\le j\le m}\|\hat\beta(p_{j\rightarrow 1})-\beta^*\|\|\hat\beta(p_{j\rightarrow 0})-\beta^*\|\sum_{i\neq j}\|x_i\|^4\\
	=&o_\mathbb{P}(m),
	\end{split}\\
	\begin{split}\label{ineq2}
	\sup_{1\le j\le m}\|\Pi_j\|\le&\sup_{1\le j\le m}\left(\|U_j-U_j^*\|+\|\widetilde U_j\|\right)\le c\sup_{1\le j\le m}\left(\|x_j\|\|R_{j}^{(6)}\|+\|R_{j}^{(3)}\|+\|R_{j}^{(4)}\|\right)\\
	\le& c\sup_{1\le j\le m}\left\{\|\hat\beta(p_{j\rightarrow 1})-\beta^*\|+\|\hat\beta(p_{j\rightarrow 0})-\beta^*\|\right\}\|x_j\|^2\\
	&+c\sup_{1\le j\le m}\|\hat\beta(p_{j\rightarrow 1})-\beta^*\|\|\hat\beta(p_{j\rightarrow 0})-\beta^*\|\|x_j\|^3\\
	=&o_\mathbb{P}(1).
	\end{split}
	\end{align}
	~
\end{proof}

\section{Other intermediate results}\label{sec-interm}

\begin{lemma}\label{tranc_moment}
	For a sequence of independent random variables $\{W_i\}_{i=1}^m$, if $\sup_{1\le i \le m}\mathbb{E}(|W_i|^{q+\epsilon})<\infty$ for some $q\ge 1$ and  any $\epsilon>0$, then we have
	\begin{align*}
	\frac{1}{m}\sup_{1\le i \le m}\mathbb{E}\left\{|W_i|^{u+1}\mathbb{I}(|W_i|\le m)\right\}\to 0,
	\end{align*}
	for any $0< u\le q.$
\end{lemma}

\begin{proof}
	Since $\sup_{1\le i \le m}\mathbb{E}(|W_i|^{q+\epsilon})<\infty$ implies $\sup_{1\le i \le m}\mathbb{E}(|W_i|^{u+\epsilon})<\infty$ for any $u<q$, we only need to show
	\begin{align*}
	\frac{1}{m}\sup_{1\le i \le m}\mathbb{E}\left\{|W_i|^{q+1}\mathbb{I}(|W_i|\le m)\right\}\to 0.
	\end{align*}
	Let $F_i(w)$ be the distribution function of $W_i$. It shows that
	\begin{align*}
	&\mathbb{E}\left\{|W_i|^{q+1}\mathbb{I}(|W_i|\le m)\right\}=\int_{|w|\le m}|w|^{q+1}\mathrm{d}F_i(w)\\
	=&(q+1)\int_{|w|\le m}\int_0^{|w|}s^q\mathrm{d}s\; \mathrm{d}F_i(w)\overset{\text{Fubini}}{=}(q+1)\int_0^m\int_{s<|w|\le m}\mathrm{d}F_i(w)\;s^q\mathrm{d}s\\
	=&(q+1)\int_0^m\left\{\mathbb{P}\left(|W_i|>s\right)-\mathbb{P}\left(|W_i|>m\right)\right\}s^q\mathrm{d}s\\
	\le&(q+1)\int_0^m\mathbb{P}\left(|W_i|>s\right)s^q\mathrm{d}s
	\le (q+1)\int_0^m\mathbb{E}\left(\frac{|W_i|^{q+\epsilon}}{s^{q+\epsilon}}\right)s^q\mathrm{d}s\\
	=&(q+1)\mathbb{E}(|W_i|^{q+\epsilon})\frac{m^{1-\epsilon}}{1-\epsilon},
	\end{align*}
	which directly implies the desired result.
\end{proof}

\begin{lemma}\label{lem_leave_one}
	Consider a sequence of independent random variables $\{W_i\}_{i=1}^m$ with $\mathbb{E}(W_i)=0$ and $\sup_{1\le i \le m}\mathbb{E}(|W_i|^{q+\epsilon})<\infty$ for some $q\ge 2$ and  any $\epsilon>0$. Then for every $t>0$, we have
	$$\mathbb{P}\left(\left|\sum_{i=1}^mW_i\right| >mt\right)=o(m^{1-q}).$$
	Furthermore, we have
	\begin{align*}
	\mathbb{P}\left(\sup_{1\le j\le m}\left|\sum_{i\ne j}W_i\right|>mt\right)=o(m^{1-q}).
	\end{align*}
	If $W_i$'s are not necessarily mean-zero, we get $\sup_{1\le j\le m}\sum_{i\ne j}W_i=O_{\mathbb{P}}(m)$.
\end{lemma}
\begin{remark}
	As seen from the proof, we can replace the condition that $\sup_{1\le i \le m}\mathbb{E}(|W_i|^{q+\epsilon})<\infty$ for some $q\ge 2$ and  any $\epsilon>0$ in Lemma \ref{lem_leave_one} by a weaker condition that $\sup_{1\le i \le m}\mathbb{E}(|W_i|^{q})<\infty$ for some $q\ge 2$  and $\{|W_i|^q\}_{i=1}^m$ are uniformly integrable.
\end{remark}

\begin{proof}
	Let $W_{m,i}=W_i\mathbb{I}(|W_i|\le m)$ and $W'_{m,i}=W_{m,i}-\mathbb{E}(W_{m,i})$. Further define $T_m=\sum_{i=1}^mW_i$, $\widehat{T}_m=\sum_{i=1}^mW_{m,i}$ and $T'_{m}=\sum_{i=1}^mW'_{m,i}$. Note that
	\begin{align*}
	&m^{q-1}\mathbb{P}\left(|T_m|>mt\right)=m^{q-1}\mathbb{P}\left(|T_m-\widehat{T}_m+\widehat{T}_m-T'_m+T'_m|>mt\right)\\
	\le&m^{q-1}\left\{\mathbb{P}\left(|T_m-\widehat{T}_m|>mt/3\right)+\mathbb{P}\left(|\widehat{T}_m-T'_m|>mt/3\right)+\mathbb{P}\left(|T'_m|>mt/3\right)\right\}.
	\end{align*}
	We show the first term in the above inequality converges to 0 by proving
	$\mathbb{P}(T_m\neq\widehat{T}_m)=o(m^{1-q}).$
	To see this, note that
	\begin{align*}
	&m^{q-1}\mathbb{P}\left(T_m\neq \widehat{T}_m\right)\le m^{q-1}\mathbb{P}\left\{\bigcup_{i=1}^m(W_i\neq W_{m,i})\right\}\\
	\le & m^{q-1}\sum_{i=1}^m\mathbb{P}\left(W_i\neq W_{m,i}\right)= m^{q-1}\sum_{i=1}^m\mathbb{P}\left(|W_i|> m\right)\\
	\le&m^{q-1}\sum_{i=1}^m\mathbb{E}\left(\frac{|W_i|^{q+\epsilon}}{m^{q+\epsilon}}\right)
	\le\frac{m\sup_{1\le i \le m}\mathbb{E}(|W_i|^{q+\epsilon})}{m^{1+\epsilon}} \to 0.
	\end{align*}
	To prove $m^{q-1}\mathbb{P}(|\widehat{T}_m-T'_m|>mt/3)\to 0$, we note that
	\begin{align*}
	&\sup_{1\le i \le m}\left|\mathbb{E}(W_{m,i})\right|=\sup_{1\le i \le m}\left|\mathbb{E}\left\{W_i\mathbb{I}(|W_i|\le m)\right\}\right|=\sup_{1\le i \le m}\left|\mathbb{E}\left\{W_i\mathbb{I}(|W_i|> m)\right\}\right|\\
	\le&\sup_{1\le i \le m}\mathbb{E}\left\{\frac{|W_i|^{1+\epsilon}}{m^\epsilon}\mathbb{I}(|W_i|> m)\right\}\le\frac{\sup_{1\le i \le m}\mathbb{E}(|W_i|^{1+\epsilon})}{m^{\epsilon}}\to 0,
	\end{align*}
	and
	\begin{align*}
	&|\widehat{T}_m-T'_m|=\left|\sum_{i=1}^m\mathbb{E}(W_{m,i})\right|\le m\sup_{1\le i \le m}\left|\mathbb{E}(W_{m,i})\right|.
	\end{align*}
	Hence $\mathbb{P}(|\widehat{T}_m-T'_m|>mt/3)=0$ as long as $m$ is large enough. Therefore, to show $\mathbb{P}\left(|T_m|>mt\right)=o(m^{1-q})$, we only need to prove $m^{q-1}\mathbb{P}\left(|T'_m|>mt\right)\to 0.$ Note that
	\begin{align*}
	&m^{q-1}\mathbb{P}\left(|T'_m|>mt\right)\le m^{q-1}\frac{\mathbb{E}(|T'_m|^{2q})}{m^{2q}t^{2q}}\\
	\le&\frac{c}{m^{q+1}t^{2q}}\left[\sum_{i=1}^m\mathbb{E}(|W'_{m,i}|^{2q})+\left\{\sum_{i=1}^m\mathbb{E}(|W'_{m,i}|^2)\right\}^q\right],
	\end{align*}
	where the last inequality follows from Rosenthal’s inequality. For the first term, we have
	\begin{align*}
	&\frac{c}{m^{q+1}t^{2q}}\sum_{i=1}^m\mathbb{E}(|W'_{m,i}|^{2q})\le \frac{2^{q-1}c}{m^2t^{2q}}\sum_{i=1}^m\mathbb{E}(|W'_{m,i}|^{q+1})\\
	\le&\frac{2^{q-1}c}{m^2t^{2q}}\sum_{i=1}^m\mathbb{E}\left[\left\{|W_{m,i}|+|\mathbb{E}(W_{m,i})|\right\}^{q+1}\right]\to 0,
	\end{align*}
	where we have used the result of Lemma \ref{tranc_moment} and the fact that $\sup_{1\le i \le m}\left|\mathbb{E}(W_{m,i})\right|\to 0$. For the second term, we deduce that
	\begin{align*}
	&\frac{c}{m^{q+1}t^{2q}}\left\{\sum_{i=1}^m\mathbb{E}(|W'_{m,i}|^2)\right\}^q\le \frac{c}{m^{q+1}t^{2q}}\left\{\sum_{i=1}^m\mathbb{E}(|W_{m,i}|^2)\right\}^q\\
	\le &\frac{c}{m^{q+1}t^{2q}}\left\{m\sup_{1\le i \le m}\mathbb{E}(|W_{i}|^2)\right\}^q\to 0.
	\end{align*}
	Thus we have proved that $\mathbb{P}\left(|T_m|>mt\right)=o(m^{1-q})$. Furthermore, we have
	\begin{align*}
	\mathbb{P}\left(\sup_{1\le j\le m}\left|\sum_{i\neq j}W_i\right|>mt\right)
	\le&\mathbb{P}\left(|T_m|>mt/2\right)+\mathbb{P}\left(\sup_{1\le j\le m}\left|W_j\right|>mt/2\right)\\
	\le&o(m^{1-q})+\sum_{i=1}^m\frac{\mathbb{E}(|W_i|^q)}{(mt/2)^q}=o(m^{1-q}),
	\end{align*}
	which indicates that $\sup_{1\le j\le m}\sum_{i\ne j}W_i=o_{\mathbb{P}}(m)$.  If $W_i$'s are not necessarily mean-zero, then
	\begin{align*}
	&\left|\sup_{1\le j\le m}\sum_{i\ne j}W_i\right|=\left|\sup_{1\le j\le m}\sum_{i\ne j}\{W_i-\mathbb{E}(W_i)+\mathbb{E}(W_i)\}\right|\\
	\le&\left| \sup_{1\le j\le m}\sum_{i\ne j}\{W_i-\mathbb{E}(W_i)\}\right|+\left|\sup_{1\le j\le m}\sum_{i\ne j}\mathbb{E}(W_i)\right|\le o_{\mathbb{P}}(m)+\sum_{i=1}^m\mathbb{E}(|W_i|)=O_{\mathbb{P}}(m).
	\end{align*}
	~
\end{proof}

\begin{proposition}\label{pro_tail_s}
	Assume that $\sup_{1\le i\le m}\mathbb{E}(\|x_i\|^{q'+\epsilon})<\infty$ for some $q'$ and any $\epsilon>0$. Then we have
	\begin{align*}
	&\sup_{1\leq j\leq m}\mathbb{P}\left\{\left\|\frac{S_j^*}{m}-\mathbb{E}\left(\frac{S_j^*}{m}\right)\right\|>t\right\}=o(m^{1-q}),\quad \text{ if }q'=2q,\\
	&\sup_{1\le j\le m}\mathbb{P}\left\{\left|\frac{\sum_{i\ne j}\|x_i\|^3}{m}-\mathbb{E}\left(\frac{\sum_{i\ne j}\|x_i\|^3}{m}\right)\right|>t\right\}=o(m^{1-q}),\quad \text{ if }q'=3q,\\
	&\sup_{1\le j\le m}\mathbb{P}\left\{\left|\frac{\sum_{i\ne j}\|x_i\|^4}{m}-\mathbb{E}\left(\frac{\sum_{i\ne j}\|x_i\|^4}{m}\right)\right|>t\right\}=o(m^{1-q}),\quad \text{ if }q'=4q,
	\end{align*}
	where $q\ge 2$.
\end{proposition}

\begin{proof}
	Let $\|A\|_F$ be Frobenius norm of a matrix $A$.
	Recall that $S_j^*=-\sum_{i\neq j}\nabla^2l(\beta^*;z_i)=\sum_{i\neq j}h(\beta^*;z_i)x_ix_i^\T$ for some function $h$, where $h(\beta;z)$ is uniformly bounded over $\beta$ and $z$ . For any $t>0$,
	\begin{align*}
	&\mathbb{P}\left\{\left\|\frac{S_j^*}{m}-\mathbb{E}\left(\frac{S_j^*}{m}\right)\right\|>t\right\}\le \mathbb{P}\left\{\left\|\frac{S_j^*}{m}-\mathbb{E}\left(\frac{S_j^*}{m}\right)\right\|_F>t\right\}\\
	\le&\sum_{u=1}^d\sum_{v=1}^d \mathbb{P}\left\{\left|\frac{S_j^*(u,v)}{m}-\mathbb{E}\left(\frac{S_j^*(u,v)}{m}\right)\right|^2>\frac{t^2}{d^2}\right\}\\
	=&\sum_{u=1}^d\sum_{v=1}^d\mathbb{P}\left\{\left|\frac{S_j^*(u,v)}{m}-\mathbb{E}\left(\frac{S_j^*(u,v)}{m}\right)\right|>\frac{t}{d}\right\}\\
	=&\sum_{u=1}^d\sum_{v=1}^d\mathbb{P}\left[\left|\frac{1}{m}\sum_{i\neq j}h(\beta^*;z_i)x_i(u)x_i(v)-\frac{1}{m}\sum_{i\neq j}\mathbb{E}\left\{h(\beta^*;z_i)x_i(u)x_i(v)\right\}\right|>\frac{t}{d}\right],
	\end{align*}
	where $S_j^*(u,v)$ is the $(u,v)$th element of the matrix $S_j^*$. For $q\ge 2$ and any $\epsilon>0$, by $\text{H}\ddot{\text{o}}\text{lder}$ inequality, we have
	\begin{align*}
	\sup_{1\le i\le m}\mathbb{E}\left\{|h(\beta^*;z_{i})x_{i}(u)x_{i}(v)|^{q+\epsilon}\right\}\le& c\sup_{1\le i\le m}\mathbb{E}\left\{|x_i(u)|^{2(q+\epsilon)}\right\}^{\frac{1}{2}}\mathbb{E}\left\{|x_i(v)|^{2(q+\epsilon)}\right\}^{\frac{1}{2}}
	\\ \le& c\sup_{1\le i\le m}\mathbb{E}\left\{\|x_i\|^{2(q+\epsilon)}\right\}.
	\end{align*}
	Then proof is completed by applying the result of Lemma \ref{lem_leave_one}.
\end{proof}

\begin{lemma}\label{lem_dudley}
	If Assumption 2 holds and $\sup_{1\le i\le m}\mathbb{E}(\|x_i\|^{2})<\infty$, then
	\begin{align*}
	\mathbb{E}\left[\sup_{\beta\in\mathcal{B}}\left|\mathbb{P}_ml(\beta)-\mathbb{E}\left\{\mathbb{P}_ml(\beta)\right\}\right|\right]=O(m^{-\frac{1}{2}}).
	\end{align*}
\end{lemma}

\begin{proof}
	Let $e_i$ be i.i.d Rademacher variables. By symmetrization argument, we have
	\begin{align*}
	\mathbb{E}\left[\sup_{\beta\in\mathcal{B}}\left|\mathbb{P}_ml(\beta)-\mathbb{E}\left\{\mathbb{P}_ml(\beta)\right\}\right|\right]\le\frac{2}{\sqrt{m}}\mathbb{E}\left[\mathbb{E}\left\{\sup_{\beta\in\mathcal{B}}\left|\frac{1}{\sqrt{m}}\sum_{i=1}^me_il(\beta;z_i)\right|\mid z_1,...,z_m\right\}\right].
	\end{align*}
	Define $\mathcal{E}_{m,\beta}=\sum_{i=1}^me_il(\beta;z_i)/\sqrt{m}$. Then
	\begin{align*}
	&\mathbb{E}\left[\exp\{\lambda(\mathcal{E}_{m,\beta_1}-\mathcal{E}_{m,\beta_2})\}\mid z_1,...,z_m\right]
	=\prod_{i=1}^m\mathbb{E}\left(\exp\left[\frac{\lambda}{\sqrt{m}}e_i\left\{l(\beta_1;z_i)-l(\beta_2;z_i)\right\}\right]\mid z_1,...,z_m\right)\\
	\le&\prod_{i=1}^m\exp\left[\frac{\lambda^2}{2m}\left\{l(\beta_1;z_i)-l(\beta_2;z_i)\right\}^2\right]
	=\exp\left[\frac{\lambda^2}{2}\frac{1}{m}\sum_{i=1}^m\left\{l(\beta_1;z_i)-l(\beta_2;z_i)\right\}^2\right]\\
	\le&\exp\left(\frac{\lambda^2}{2}\frac{1}{m}\sum_{i=1}^mc\|x_i\|^2\|\beta_1-\beta_2\|^2\right)=\exp\left(\frac{\lambda^2}{2}\|\beta_1-\beta_2\|^2\rho_m^2\right),
	\end{align*}
	where $\rho_m^2=\sum_{i=1}^mc\|x_i\|^2/m$. Define
	\begin{align*}
	N(\mathcal{B},\|\cdot\|,\epsilon)=\inf\Big\{N:&\text{ there exists a set }\{\beta_i\}_{i=1}^N,\\
	& \text{ such that for any }\beta\in\mathcal{B},\text{ there exists an }i, \text{ s.t. }\|\beta-\beta_i\|\le\epsilon  \Big\}
	\end{align*}
	and $\text{diam}(\mathcal{B})=\sup\{\|\beta_1-\beta_2\|: \beta_1,\beta_2\in\mathcal{B}\}$. By Dudley's entropy integral, we have
	\begin{align*}
	\mathbb{E}\left(\sup_{\beta\in \mathcal{B}}\left|\mathcal{E}_{m,\beta}\right|\mid z_1,...,z_m\right)\le 4\sqrt{2}\int_{0}^{\frac{\text{diam}(\mathcal{B})\rho_m}{2}}\sqrt{\log\left\{N\left(\mathcal{B},\|\cdot\|,\frac{\epsilon}{\rho_m}\right)\right\}}d\epsilon,
	\end{align*}
	Set $u=2\epsilon/\{\text{diam}(\mathcal{B})\rho_m\}$ and note that $\log N(\mathcal{B},\|\cdot\|,\epsilon)\le d\log\{1+2\mathrm{diam}(\mathcal{B})/\epsilon\}$. We deduce that
	\begin{align*}
	&\mathbb{E}\left[\sup_{\beta\in\mathcal{B}}\left|\mathbb{P}_ml(\beta)-\mathbb{E}\left\{\mathbb{P}_ml(\beta)\right\}\right|\right]\le\frac{8\sqrt{2}}{\sqrt{m}}\mathbb{E}\left[\int_{0}^{\frac{\text{diam}(\mathcal{B})\rho_m}{2}}\sqrt{\log\left\{N\left(\mathcal{B},\|\cdot\|,\frac{\epsilon}{\rho_m}\right)\right\}}d\epsilon\right]\\
	=&\frac{8\sqrt{2}}{\sqrt{m}}\mathbb{E}\left[\int_{0}^{1}\sqrt{\log\left\{N\left(\mathcal{B},\|\cdot\|,\frac{\text{diam}(\mathcal{B})u}{2}\right)\right\}}\frac{\text{diam}(\mathcal{B})\rho_m}{2}du\right]\\
	\le&\frac{4\sqrt{2}}{\sqrt{m}}\text{diam}(\mathcal{B})\mathbb{E}\left\{\rho_m\int_{0}^{1}\sqrt{d\log\left(1+\frac{4}{u}\right)}du\right\}
	=cm^{-\frac{1}{2}}\mathbb{E}(\rho_m)=O(m^{-\frac{1}{2}}).
	\end{align*}
	~
\end{proof}

\begin{lemma}\label{lem_sup}
	For two functions $f$ and $g$ defined on the same space $\mathcal{W}$, we have
	\begin{align*}
	\left|\sup_{w\in\mathcal{W}}f(w)-\sup_{w\in\mathcal{W}}g(w)\right|\le\sup_{w\in\mathcal{W}}\left|f(w)-g(w)\right|.
	\end{align*}
\end{lemma}

\begin{proof}
	Note that
	\begin{align*}
	&\sup_{w\in\mathcal{W}}f(w)-\sup_{w\in\mathcal{W}}g(w)=\sup_{w\in\mathcal{W}}\inf_{w'\in\mathcal{W}}\left\{f(w)-g(w')\right\}\\
	&\le\sup_{w\in\mathcal{W}}\left\{f(w)-g(w)\right\}\le\sup_{w\in\mathcal{W}}\left|f(w)-g(w)\right|,
	\end{align*}
	and similarly,
	\begin{align*}
	\sup_{w\in\mathcal{W}}g(w)-\sup_{w\in\mathcal{W}}f(w)\le\sup_{w\in\mathcal{W}}\left|f(w)-g(w)\right|,
	\end{align*}
	which completes the proof.
\end{proof}

\begin{proposition}\label{pro_tail_beta}
	Suppose the following conditions are satisfied:\\
	(i) Assumptions 2--4 hold;\\
	(ii) we have $\sup_{1\le i \le m}\mathbb{E}(\|x_i\|^{2})<\infty$;\\
	(iii) for some $\nu>0$, we have $\sup_{\beta\in\mathcal{B}}\left|\mathbb{E}\left\{\mathbb{P}_ml(\beta)\right\}-\mathcal{L}(\beta)\right|=O(m^{-\nu})$; \\
	(iv) the function $\mathcal{L}(\beta)$ is twice continuously differentiable;  \\
	(v) the global maximizer$\beta^*$ is not on the boundary of $\mathcal{B}$; and\\
	(vi) for some $c>0$, we have $\nabla^2 \mathcal{L}(\beta^*)\preceq -cI$, where $I$ is the identity matrix.
	\\Then for $t=O(m^{-\omega})$ with $0<\omega<(\nu/2)\wedge(1/4)$,  we have for $a=0,1$,
	\begin{align*}
	\mathbb{P}\left\{\sup_{1\le j\le m}\|\hat\beta(p_{j\to a})-\beta^*\|>t\right\}\le\exp[-c\{t^4+o(t^4)\}m].
	\end{align*}
\end{proposition}

\begin{proof}
	Under Assumption 4 and Conditions (iv)--(vi), we know that there exists $\delta$ such that for all $\|\beta-\beta^*\|\le \delta$, we have $\bigtriangledown^2 \mathcal{L}(\beta)\preceq -cI$ for some constant $c>0$. Under Condition (v) and by the Taylor expansion, we have
	\begin{align*}
	\mathcal{L}(\beta) =& \mathcal{L}(\beta^*) + (\beta-\beta^*)^\T\bigtriangledown^2 \mathcal{L}(\tilde\beta)(\beta-\beta^*)/2\\
	\le &\mathcal{L}(\beta^*) - c\|\beta-\beta^*\|^2.
	\end{align*}
	Then for small enough $t$ with $t\le\delta$, we have
	\begin{align*}
	\mathcal{L}(\beta^*)-\sup_{\beta:\|\beta-\beta^*\|>t}\mathcal{L}(\beta)=& \mathcal{L}(\beta^*) - \max\left\{\sup_{\beta:\|\beta-\beta^*\|\ge\delta}\mathcal{L}(\beta),\sup_{\beta:t<\|\beta-\beta^*\|<\delta}\mathcal{L}(\beta)\right\}\\
	\ge &\mathcal{L}(\beta^*) - \max\left\{\mathcal{L}(\beta^*)-c,\mathcal{L}(\beta^*)-ct^2\right\}\\
	\ge&\min\{c,ct^2\}= ct^2.
	\end{align*}
	For any $\beta\in\mathcal{B}$ and $z_1,z_2\in \mathbb{R}^d\times\{0,1\}$, $|l(\beta;z_1)-l(\beta;z_2)|\le L$.
	Applying the result of Lemma \ref{lem_sup}, we have
	\begin{align*}
	&\left|\sup_{\beta\in\mathcal{B}}\bigg|\frac{1}{m}\bigg\{\sum_{i\neq j}l(\beta;z_i)+l(\beta;z_j)\bigg\}-\mathbb{E}\left\{\mathbb{P}_ml(\beta)\right\}\bigg|-\sup_{\beta\in\mathcal{B}}\bigg|\frac{1}{m}\bigg\{\sum_{i\neq j}l(\beta;z_i)+l(\beta;z'_j)\bigg\}-\mathbb{E}\left\{\mathbb{P}_ml(\beta)\right\}\bigg|\right|\\
	\le&\sup_{\beta\in\mathcal{B}}\left|\bigg|\frac{1}{m}\bigg\{\sum_{i\neq j}l(\beta;z_i)+l(\beta;z_j)\bigg\}-\mathbb{E}\left\{\mathbb{P}_ml(\beta)\right\}\bigg|-\bigg|\frac{1}{m}\bigg\{\sum_{i\neq j}l(\beta;z_i)+l(\beta;z'_j)\bigg\}-\mathbb{E}\left\{\mathbb{P}_ml(\beta)\right\}\bigg|\right|\\
	\le&\sup_{\beta\in\mathcal{B}}\left|\frac{l(\beta;z_j)}{m}-\frac{l(\beta;z'_j)}{m}\right|\le\frac{L}{m},
	\end{align*}
	which means that $\sup_{\beta\in\mathcal{B}}|\mathbb{P}_ml(\beta)-\mathbb{E}\{\mathbb{P}_ml(\beta)\}|$ is $L/m$-bounded difference function. Note that
	\begin{align*}
	&\mathbb{P}(\|\hat\beta-\beta^*\|> t)\\
	\le& \mathbb{P}\left\{\sup_{\beta:\|\beta-\beta^*\|>t}\mathbb{P}_ml(\beta)>\mathbb{P}_ml(\beta^*)\right\}\\
	\le&\mathbb{P}\left(\sup_{\beta:\|\beta-\beta^*\|>t}\left[\mathbb{P}_ml(\beta)-\mathbb{E}\left\{\mathbb{P}_ml(\beta)\right\}\right]+\sup_{\beta:\|\beta-\beta^*\|>t}\mathbb{E}\left\{\mathbb{P}_ml(\beta)\right\}>\mathbb{P}_ml(\beta^*)-\mathbb{E}\left\{\mathbb{P}_ml(\beta^*)\right\}+\mathbb{E}\left\{\mathbb{P}_ml(\beta^*)\right\}\right)\\
	=&\mathbb{P}\left(\sup_{\beta:\|\beta-\beta^*\|>t}\left[\mathbb{P}_ml(\beta)-\mathbb{E}\left\{\mathbb{P}_ml(\beta)\right\}\right]-\left[\mathbb{P}_ml(\beta^*)-\mathbb{E}\left\{\mathbb{P}_ml(\beta^*)\right\}\right]>\mathbb{E}\left\{\mathbb{P}_ml(\beta^*)\right\}-\sup_{\beta:\|\beta-\beta^*\|>t}\mathbb{E}\left\{\mathbb{P}_ml(\beta)\right\}\right)\\
	\le&\mathbb{P}\left[2\sup_{\beta\in\mathcal{B}}\left|\mathbb{P}_ml(\beta)-\mathbb{E}\left\{\mathbb{P}_ml(\beta)\right\}\right|>\mathbb{E}\left\{\mathbb{P}_ml(\beta^*)\right\}-\sup_{\beta:\|\beta-\beta^*\|>t}\mathbb{E}\left\{\mathbb{P}_ml(\beta)\right\}\right],
	\end{align*}
	and
	\begin{align*}
	&\mathbb{E}\left\{\mathbb{P}_ml(\beta^*)\right\}-\sup_{\beta:\|\beta-\beta^*\|>t}\mathbb{E}\left\{\mathbb{P}_ml(\beta)\right\}\\
	=&\mathbb{E}\left\{\mathbb{P}_ml(\beta^*)\right\}-\mathcal{L}(\beta^*)+\mathcal{L}(\beta^*)-\sup_{\beta:\|\beta-\beta^*\|>t}\mathcal{L}(\beta)+\sup_{\beta:\|\beta-\beta^*\|>t}\mathcal{L}(\beta)-\sup_{\beta:\|\beta-\beta^*\|>t}\mathbb{E}\left\{\mathbb{P}_ml(\beta)\right\}\\
	\ge&-\left|\mathbb{E}\left\{\mathbb{P}_ml(\beta^*)\right\}-\mathcal{L}(\beta^*)\right|+\mathcal{L}(\beta^*)-\sup_{\beta:\|\beta-\beta^*\|>t}\mathcal{L}(\beta)-\sup_{\beta:\|\beta-\beta^*\|>t}\left|\mathbb{E}\left\{\mathbb{P}_ml(\beta)\right\}-\mathcal{L}(\beta)\right|\\
	\ge& \mathcal{L}(\beta^*)-\sup_{\beta:\|\beta-\beta^*\|>t}\mathcal{L}(\beta)-2\sup_{\beta\in\mathcal{B}}\left|\mathbb{E}\left\{\mathbb{P}_ml(\beta)\right\}-\mathcal{L}(\beta)\right|\\
	\ge&ct^2-2\sup_{\beta\in\mathcal{B}}\left|\mathbb{E}\left\{\mathbb{P}_ml(\beta)\right\}-\mathcal{L}(\beta)\right|,
	\end{align*}
	where we have applied the result of Lemma \ref{lem_sup} in the first inequality. We deduce that
	\begin{align*}
	&\mathbb{P}(\|\hat\beta-\beta^*\|> t)\\
	\le&\mathbb{P}\left[\sup_{\beta\in\mathcal{B}}\left|\mathbb{P}_ml(\beta)-\mathbb{E}\left\{\mathbb{P}_ml(\beta)\right\}\right|> ct^2-\sup_{\beta\in\mathcal{B}}\left|\mathbb{E}\left\{\mathbb{P}_ml(\beta)\right\}-\mathcal{L}(\beta)\right|\right]\\
	=&\mathbb{P}\Bigg(\sup_{\beta\in\mathcal{B}}\left|\mathbb{P}_ml(\beta)-\mathbb{E}\left\{\mathbb{P}_ml(\beta)\right\}\right|-\mathbb{E}\left[\sup_{\beta\in\mathcal{B}}\left|\mathbb{P}_ml(\beta)-\mathbb{E}\left\{\mathbb{P}_ml(\beta)\right\}\right|\right]\\
	&\qquad>ct^2-\sup_{\beta\in\mathcal{B}}\left|\mathbb{E}\left\{\mathbb{P}_ml(\beta)\right\}-\mathcal{L}(\beta)\right|-\mathbb{E}\left[\sup_{\beta\in\mathcal{B}}\left|\mathbb{P}_ml(\beta)-\mathbb{E}\left\{\mathbb{P}_ml(\beta)\right\}\right|\right]\Bigg)\\
	\le&\exp\left[-2\left\{ct^2+O(m^{-\nu})+O(m^{-\frac{1}{2}})\right\}^2m/L^2\right]=\exp[-c\{t^4+o(t^4)\}m],
	\end{align*}
	where we have used the Condition (iii) that $\sup_{\beta\in\mathcal{B}}|\mathbb{E}\{\mathbb{P}_ml(\beta)\}-\mathcal{L}(\beta)|=O(m^{-\nu})$, the result of Lemma \ref{lem_dudley}, the condition that $t=O(m^{-\omega})$ with $0<\omega<(\nu/2)\wedge(1/4)$, and the McDiarmid's inequality. Recall that
	$\sup_{1\le j\le m}\sup_{\beta\in \mathcal{B}}|\mathbb{P}^{j\to a}_ml(\beta)-\mathbb{P}_ml(\beta)|\le L/m$. Then we have
	\begin{align*}
	&\mathbb{P}\{\sup_{1\le j\le m}\|\hat\beta(p_{j\to a})-\beta^*\|> t\}\\
	\le& \mathbb{P}\left[\bigcup_{1\le j\le m}\left\{\sup_{\beta:\|\beta-\beta^*\|>t}\mathbb{P}_m^{j\to a}l(\beta)>\mathbb{P}_m^{j\to a}l(\beta^*)\right\}\right]\\
	\le&\mathbb{P}\left(\bigcup_{1\le j\le m}\left[\sup_{\beta\in\mathcal{B}}\left|\mathbb{P}_m^{j\to a}l(\beta)-\mathbb{E}\left\{\mathbb{P}_m^{j\to a}l(\beta)\right\}\right|> ct^2-\sup_{\beta\in\mathcal{B}}\left|\mathbb{E}\left\{\mathbb{P}_m^{j\to a}l(\beta)\right\}-\mathcal{L}(\beta)\right|\right]\right)\\
	=&\mathbb{P}\Bigg(\bigcup_{1\le j\le m}\Bigg[\sup_{\beta\in\mathcal{B}}\left|\mathbb{P}_m^{j\to a}l(\beta)-\mathbb{P}_ml(\beta)+\mathbb{P}_ml(\beta)-\mathbb{E}\left\{\mathbb{P}_ml(\beta)\right\}+\mathbb{E}\left\{\mathbb{P}_ml(\beta)\right\}-\mathbb{E}\left\{\mathbb{P}_m^{j\to a}l(\beta)\right\}\right|\\
	&\qquad\qquad\quad\;\;> ct^2-\sup_{\beta\in\mathcal{B}}\left|\mathbb{E}\left\{\mathbb{P}_m^{j\to a}l(\beta)\right\}-\mathbb{E}\left\{\mathbb{P}_ml(\beta)\right\}+\mathbb{E}\left\{\mathbb{P}_ml(\beta)\right\}-\mathcal{L}(\beta)\right|\Bigg]\Bigg)\\
	\le&\mathbb{P}\left(\bigcup_{1\le j\le m}\left[\sup_{\beta\in\mathcal{B}}\left|\mathbb{P}_ml(\beta)-\mathbb{E}\left\{\mathbb{P}_ml(\beta)\right\}\right|+\frac{2L}{m}> ct^2-\sup_{\beta\in\mathcal{B}}\left|\mathbb{E}\left\{\mathbb{P}_ml(\beta)\right\}-\mathcal{L}(\beta)\right|-\frac{L}{m}\right]\right)\\
	=&\mathbb{P}\Bigg(\sup_{\beta\in\mathcal{B}}\left|\mathbb{P}_ml(\beta)-\mathbb{E}\left\{\mathbb{P}_ml(\beta)\right\}\right|-\mathbb{E}\left[\sup_{\beta\in\mathcal{B}}\left|\mathbb{P}_ml(\beta)-\mathbb{E}\left\{\mathbb{P}_ml(\beta)\right\}\right|\right]\\
	&\qquad>ct^2-\sup_{\beta\in\mathcal{B}}\left|\mathbb{E}\left\{\mathbb{P}_ml(\beta)\right\}-\mathcal{L}(\beta)\right|-\frac{3L}{m}-\mathbb{E}\left[\sup_{\beta\in\mathcal{B}}\left|\mathbb{P}_ml(\beta)-\mathbb{E}\left\{\mathbb{P}_ml(\beta)\right\}\right|\right]\Bigg)\\
	\le&\exp\left[-2\left\{ct^2+O(m^{-\nu})+O(m^{-1})+O(m^{-\frac{1}{2}})\right\}^2m/L^2\right]=\exp[-c\{t^4+o(t^4)\}m].
	\end{align*}
	~
\end{proof}

\begin{lemma}\label{lem_sum_of bernoulli}
	Denote by $\text{B}(m,p)$ the binomial distribution with $m$ Bernoulli trials and success probability $p$. Consider two sequences of independent Bernoulli random variables $\{V_i\}_{i=1}^m$ and $\{U_i\}_{i=1}^m$, where $V_i\sim \text{B}(1,p_i)$ independently and $U_i\sim^{i.i.d}\text{B}(1,p)$ with $p_i\ge p$ for $i=1,...,m$. Let $W_1=\sum_{i=1}^mV_i$ and $W_2=\sum_{i=1}^mU_i$. Suppose $f$ is non-decreasing and $W_1$ and $W_2$ belong to the domain of $f$ almost surely.
	Then we have
	\begin{align*}
	\mathbb{E}\left\{f(W_1)\right\}\ge \mathbb{E}\left\{f(W_2)\right\}.
	\end{align*}
\end{lemma}
\begin{proof}
	Suppose there are $n$ $p_i$'s such that $p_i\neq p$. Without loss of generality, we assume $p_i>p$ for $i=1,...,n$ and $p_i=p$ for $i=n+1,...,m$.
	Then the conclusion is obviously true if $n=0$. We use the induction argument. Suppose the conclusion is true if $n=s$. Consider the case $n=s+1$. Let $T=\sum_{i=1}^s\text{B}(1,p_i)+\text{B}(m-s-1,p)$, where $\text{B}(1,p_i)$'s and $\text{B}(m-s-1,p)$ are independent with each other.
	We note that
	\begin{align*}
	&\mathbb{E}\left\{f(W_1)\right\}-\mathbb{E}\left\{f(W_2)\right\}=\mathbb{E}\left\{f\left(\sum_{i=1}^mV_i\right)\right\}-\mathbb{E}\left\{f\left(\sum_{i=1}^mU_i\right)\right\}\\
	=&\mathbb{E}\left[f\left\{\sum_{i=1}^s\text{B}(1,p_i)+\text{B}(1,p_{s+1})+\text{B}(m-s-1,p)\right\}\right]\\
	&-\mathbb{E}\left[f\left\{\sum_{i=1}^s\text{B}(1,p_i)+\text{B}(1,p)+\text{B}(m-s-1,p)\right\}\right]\\
	&+\mathbb{E}\left[f\left\{\sum_{i=1}^s\text{B}(1,p_i)+\text{B}(1,p)+\text{B}(m-s-1,p)\right\}\right]-\mathbb{E}\left[f\left\{\text{B}(m,p)\right\}\right]\\
	=&p_{s+1}\mathbb{E}\left\{f(T+1)\right\}+(1-p_{s+1})\mathbb{E}\left\{f(T)\right\}-p\mathbb{E}\left\{f(T+1)\right\}-(1-p)\mathbb{E}\left\{f(T)\right\}\\
	&+\mathbb{E}\left[f\left\{\sum_{i=1}^s\text{B}(1,p_i)+\text{B}(1,p)+\text{B}(m-s-1,p)\right\}\right]-\mathbb{E}\left[f\left\{\text{B}(m,p)\right\}\right]\\
	=&\left(p_{s+1}-p\right)\left[\mathbb{E}\left\{f(T+1)\right\}-\mathbb{E}\left\{f(T)\right\}\right]\\
	&+\mathbb{E}\left[f\left\{\sum_{i=1}^s\text{B}(1,p_i)+\text{B}(m-s,p)\right\}\right]-\mathbb{E}\left[f\left\{\text{B}(m,p)\right\}\right]\geq 0,
	\end{align*}
	where the first term in the last equation is greater or equal to zero as $f$ is non-decreasing and $p_{s+1}>p$, and the second term is greater or equal to zero by the  induction hypothesis.
\end{proof}

\begin{lemma}\label{lem_singular_value_contin}
	For a matrix $A$ with $\|A\|<\infty$, let $\sigma_{\min}\left(A\right)=\inf_{\|w\|=1}\|Aw\|$. For two matrices $A,B\in\mathbb{R}^{n_1\times n_2}$ and any $\epsilon>0$, there exists $\delta$ such that $|\sigma_{\min}(A)-\sigma_{\min}(B)|\le\epsilon$ as long as $\|A-B\|\le\delta$.
\end{lemma}

\begin{proof}
	Let $A=U\Sigma V^\T$ be the singular value decomposition of the matrix $A$. Then
	\begin{align*}
	\inf_{\|w\|=1}\|Aw\|=\inf_{\|w\|=1}\|U\Sigma V^\T w\|=\inf_{\|Vw\|=1}\|U\Sigma w\|=\inf_{\|w\|=1}\|\Sigma w\|=\sigma_{\min}\left(A\right).
	\end{align*}
	Define the function $A(w)=\|Aw\|$, which is continuous.
	This can be easily proved by noting that $\big|\|Aw_1\|-\|Aw_2\|\big|\le\|A\|\|w_1-w_2\|$. Thus $\inf_{\|w\|=1}\|Aw\|=\min_{\|w\|=1}\|Aw\|$. Let $\sigma_{\min}(A)=a$ and $\sigma_{\min}(B)=b$. Further set $w_a=\arg\min_{\|w\|=1}\|Aw\|$ and $w_b=\arg\min_{\|w\|=1}\|Bw\|$.
	For any $\epsilon>0$ and $\|A-B\|\le\epsilon$, we have
	\begin{align*}
	a\le\|Aw_b\|=\|(A-B+B)w_b\|\le \|A-B\|\|w_b\|+\|Bw_b\|\le\epsilon+b,
	\end{align*}
	and similarly $b\le\epsilon+a$, which leads to $|a-b|\le\epsilon$.
\end{proof}

\section{Proof of Theorem 1}\label{sec-theorem1}

\begin{proof}
	
	By Proposition 1 and the Cauchy-Schwarz inequality, we have
	\begin{align*}
	J_{m,1}=&\sum_{j=1}^m\mathbb{E}\left[\frac{|x_j^\T\{\hat\beta(p_{j\rightarrow 0})-\hat\beta(p_{j\rightarrow 1})\}|}{\left\{c\alpha^{-1}\sum_{i\neq j}\mathbb{I}(p_i>\gamma)\right\}\vee\varepsilon^{1/(1-k)}}\right]\\
	\le &\sum_{j=1}^m\left(\mathbb{E}\left[\left|x_j^\T\{\hat\beta(p_{j\rightarrow 0})-\hat\beta(p_{j\rightarrow 1})\}\right|^2\right]\right)^{1/2}\left\{\mathbb{E}\left(\left[\left\{c\alpha^{-1}\sum_{i\neq j}\mathbb{I}(p_i>\gamma)\right\}\vee\varepsilon^{1/(1-k)}\right]^{-2}\right)\right\}^{1/2}\\
	\le&\sum_{j=1}^m\left\{\mathbb{E}(\|x_j\|^4)\right\}^{\frac{1}{4}}\left[\mathbb{E}\left\{\left\|\hat\beta(p_{j\rightarrow 0})-\hat\beta(p_{j\rightarrow 1})\right\|^4\right\}\right]^{\frac{1}{4}}\left\{\mathbb{E}\left(\left[\left\{c\alpha^{-1}\sum_{i\neq j}\mathbb{I}(p_i>\gamma)\right\}\vee\varepsilon^{1/(1-k)}\right]^{-2}\right)\right\}^{\frac{1}{2}},
	\end{align*}
	and
	\begin{align*}
	J_{m,2}=&\sum_{j=1}^m\mathbb{E}\left(\frac{\alpha^{-1}\sum_{i\neq j}\mathbb{I}(p_i>\gamma)|x_i^\T\{\hat\beta(p_{j\rightarrow 0})-\hat\beta(p_{j\rightarrow 1})\}|+\alpha^{-1}}{\left[\left\{c\alpha^{-1}\sum_{i\neq j}\mathbb{I}(p_i>\gamma)\right\}\vee\varepsilon^{1/(1-k)}\right]^2}\right)\\
	\le&\alpha^{-1}\sum_{j=1}^m\left\{\mathbb{E}\left(\left[\sum_{i\neq j}\mathbb{I}(p_i>\gamma)\left|x_i^\T\{\hat\beta(p_{j\rightarrow 0})-\hat\beta(p_{j\rightarrow 1})\}\right|\right]^2\right)\right\}^{1/2}\\
	&\qquad\quad\;\;\;\left\{\mathbb{E}\left(\left[\left\{c\alpha^{-1}\sum_{i\neq j}\mathbb{I}(p_i>\gamma)\right\}\vee\varepsilon^{1/(1-k)}\right]^{-4}\right)\right\}^{1/2}\\
	&+\alpha^{-1}\sum_{j=1}^m\mathbb{E}\left(\left[\left\{c\alpha^{-1}\sum_{i\neq j}\mathbb{I}(p_i>\gamma)\right\}\vee\varepsilon^{1/(1-k)}\right]^{-2}\right)\\
	\le&\alpha^{-1}\sum_{j=1}^mm\left(\mathbb{E}\left[\left\{\frac{1}{m}\sum_{i\neq j}\mathbb{I}(p_i>\gamma)\|x_i\|\right\}^4\right]\right)^{1/4}\left[\mathbb{E}\left\{\left\|\hat\beta(p_{j\rightarrow 0})-\hat\beta(p_{j\rightarrow 1})\right\|^4\right\}\right]^{1/4}\\
	&\qquad\quad\;\;\;\left\{\mathbb{E}\left(\left[\left\{c\alpha^{-1}\sum_{i\neq j}\mathbb{I}(p_i>\gamma)\right\}\vee\varepsilon^{1/(1-k)}\right]^{-4}\right)\right\}^{1/2}\\
	&+\alpha^{-1}\sum_{j=1}^m\mathbb{E}\left(\left[\left\{c\alpha^{-1}\sum_{i\neq j}\mathbb{I}(p_i>\gamma)\right\}\vee\varepsilon^{1/(1-k)}\right]^{-2}\right)\\
	\le&\alpha^{-1}\sum_{j=1}^mm\left\{\frac{1}{m}\sum_{i=1}^m\mathbb{E}(\|x_i\|^4)\right\}^{1/4}\left[\mathbb{E}\left\{\left\|\hat\beta(p_{j\rightarrow 0})-\hat\beta(p_{j\rightarrow 1})\right\|^4\right\}\right]^{1/4}\\
	&\qquad\quad\;\;\;\left\{\mathbb{E}\left(\left[\left\{c\alpha^{-1}\sum_{i\neq j}\mathbb{I}(p_i>\gamma)\right\}\vee\varepsilon^{1/(1-k)}\right]^{-4}\right)\right\}^{1/2}\\
	&+\alpha^{-1}\sum_{j=1}^m\mathbb{E}\left(\left[\left\{c\alpha^{-1}\sum_{i\neq j}\mathbb{I}(p_i>\gamma)\right\}\vee\varepsilon^{1/(1-k)}\right]^{-2}\right).
	\end{align*}
	
	To get the last inequality, we note that for a sequence of random variables $\{W_{i}\}_{i=1}^m$ and $N\geq 1$,
	\begin{align*}
	\mathbb{E}\left\{\left(\frac{1}{m}\sum_{i=1}^m\|W_i\|\right)^N\right\}\le \mathbb{E}\left(\frac{1}{m}\sum_{i=1}^m\|W_i\|^N\right)=\frac{1}{m}\sum_{i=1}^m\mathbb{E}\left(\|W_i\|^N\right)
	\end{align*}
	by Jensen's inequality. Note that
	\begin{align*}
	&\left[\left\{c\alpha^{-1}\sum_{i\neq j}\mathbb{I}(p_i>\gamma)\right\}\vee\varepsilon^{1/(1-k)}\right]^{-N}\le
	\left[\left\{c\alpha^{-1}\sum\limits_{{i\in\mathcal{M}_0,i\neq j}}\mathbb{I}(p_i>\gamma)\right\}\vee\varepsilon^{1/(1-k)}\right]^{-N}\\
	\le&\left[\left\{c\alpha^{-1}\sum\limits_{{i,j\in\mathcal{M}_0,i\neq j}}\mathbb{I}(p_i>\gamma)\right\}\vee\varepsilon^{1/(1-k)}\right]^{-N}.
	\end{align*}
	Under Assumption 1, we have $\mathbb{P}(p_i>\gamma)\ge1-\gamma$ for $i\in\mathcal{M}_0$. Consider $m_0-1$ independent random variables $u_i$'s which follow uniform distribution on $[0,1]$. Set $\varepsilon^{1/(1-k)}\le c\alpha^{-1}$, then $\{c\alpha^{-1}\sum_{i=1}^{m_0-1}\mathbb{I}(u_i>\gamma)\ge\varepsilon^{1/(1-k)}\}=\{\sum_{i=1}^{m_0-1}\mathbb{I}(u_i>\gamma)\ge 1\}$. Let $\mu=1-\gamma$. Then by the result of Lemma \ref{lem_sum_of bernoulli}, for any positive integer $N$, we have
	\begin{align*}
	&\sup_{1\le j\le m}\mathbb{E}\left(\left[\left\{c\alpha^{-1}\sum_{i\neq j}\mathbb{I}(p_i>\gamma)\right\}\vee\varepsilon^{1/(1-k)}\right]^{-N}\right)\le \mathbb{E}\left(\left[\left\{c\alpha^{-1}\sum_{i=1}^{m_0-1}\mathbb{I}(u_i>\gamma)\right\}\vee\varepsilon^{1/(1-k)}\right]^{-N}\right)\\
	=&(1-\mu)^{m_0-1}\varepsilon^{-N/(1-k)}+c\alpha^N\sum_{i=1}^{m_0-1}\binom{m_0-1}{i}\mu^{i}(1-\mu)^{m_0-1-i}i^{-N}.
	\end{align*}
	We notice that
	\begin{align*}
	\sum_{i=1}^{m_0-1}\binom{m_0-1}{i}\mu^{i}(1-\mu)^{m_0-1-i}i^{-N}=&\sum_{j=0}^{m_0-2}\binom{m_0-1}{j+1}\mu^{j+1}(1-\mu)^{m_0-1-(j+1)}(j+1)^{-N}\\
	=&\sum_{j=0}^{m_0-2}\binom{m_0-2}{j}\frac{m_0-1}{j+1}\mu^{j+1}(1-\mu)^{m_0-2-j}(j+1)^{-N}\\
	=&(m_0-1)\mu\sum_{j=0}^{m_0-2}\binom{m_0-2}{j}\mu^{j}(1-\mu)^{m_0-2-j}(j+1)^{-(N+1)}\\
	=&(m_0-1)\mu\mathbb{E}\left\{(W+1)^{-(N+1)}\right\}=O(m_0^{-N}),
	\end{align*}
	where $W\sim\text{B}(m_0-2,\mu)$ and we have used the result from \citet{Cribari-Neto:2000}.
	As $(1-\mu)^{m_0-1}=o(\alpha^Nm_0^{-N})$ and $m_0=cm$ from Condition (viii), it follows that
	\begin{align*}
	\sup_{1\le j\le m}\mathbb{E}\left(\left[\left\{c\alpha^{-1}\sum_{i\neq j}\mathbb{I}(p_i>\gamma)\right\}\vee\varepsilon^{1/(1-k)}\right]^{-N}\right)=O(\alpha^N m_0^{-N})=O(\alpha^N m^{-N}).
	\end{align*}
	Since $\mathcal{B}$ is compact,
	$\left\|\hat\beta(p_{j\rightarrow a})-\beta^{*}\right\|\le c$ for $a=0,1$. Thus we get
	\begin{align*}
	&\mathbb{E}\left\{\left\|\hat\beta(p_{j\rightarrow 0})-\hat\beta(p_{j\rightarrow 1})\right\|^N\right\}\\
	=&\mathbb{E}\left[\left\|\hat\beta(p_{j\rightarrow 0})-\hat\beta(p_{j\rightarrow 1})\right\|^N\mathbb{I}\left\{\left\|\hat\beta(p_{j\rightarrow 0})-\hat\beta(p_{j\rightarrow 1})\right\|\le K/m\right\}\right]\\
	&+\mathbb{E}\left[\left\|\hat\beta(p_{j\rightarrow 0})-\hat\beta(p_{j\rightarrow 1})\right\|^N\mathbb{I}\left\{\left\|\hat\beta(p_{j\rightarrow 0})-\hat\beta(p_{j\rightarrow 1})\right\|> K/m\right\}\right]\\
	\le &\left(Km^{-1}\right)^{N}+\mathbb{E}\left[\left\|\hat\beta(p_{j\rightarrow 0})-\hat\beta(p_{j\rightarrow 1})\right\|^N\mathbb{I}\left\{\left\|\hat\beta(p_{j\rightarrow 0})-\hat\beta(p_{j\rightarrow 1})\right\|> K/m\right\}\right]\\
	\le&\left(Km^{-1}\right)^{N}+c\mathbb{P}\left\{\left\|\hat\beta(p_{j\rightarrow 0})-\hat\beta(p_{j\rightarrow 1})\right\|> K/m\right\}.
	\end{align*}
	Recall that $S_j^*=-\sum_{i\neq j}\nabla^2l(\beta^*;z_i)=\sum_{i\neq j}h(\beta^*;z_i)x_ix_i^\T$ with $h(\beta;z)$ being uniformly bounded over $\beta$ and $z$. Under Condition (vii) , we have  $\sigma_{\min}[-\mathbb{E}\{\sum_{i=1}^m\nabla^2l(\beta^{*};z_i)/m\}]>c$ as $m\to\infty$. Then we have
	\begin{align*}
	\inf_{1\le j\le m}\sigma_{\min}\left\{\mathbb{E}\left(\frac{S_j^{*}}{m}\right)\right\}=\inf_{1\le j\le m}\sigma_{\min}\left[\mathbb{E}\left\{\frac{-\sum_{i=1}^m\nabla^2l(\beta^{*};z_i)}{m}+\frac{h(\beta^{*};z_j)x_jx_j^\T}{m}\right\}\right]>c
	\end{align*}
	as $m\to\infty$, where we have used Condition (ii) and Lemma \ref{lem_singular_value_contin}. Let $\{\lambda_i\}^{4}_{i=1}$ be four positive numbers such that $\lambda_1=o(1)$, $\lambda_2=O(m^{-\omega})$ with $0<\omega<1/4$, $\lambda_3=o(1)$ and $\lambda_4=o(1)$. If $\|S_j^*/m-\mathbb{E}(S_j^*/m)\|\le\lambda_1$, $\|\hat\beta(p_{j\to a})-\beta^*\|\le\lambda_2$ for $a=0,1$, $|\sum_{i\ne j}\|x_i\|^3/m-\mathbb{E}(\sum_{i\ne j}\|x_i\|^3/m)|\le\lambda_3$ and $|\sum_{i\ne j}\|x_i\|^4/m-\mathbb{E}(\sum_{i\ne j}\|x_i\|^4/m)|\le\lambda_4$ for any $j=1,...,m$, then by (\ref{ineq1}) and Condition (ii), we have
	\begin{align*}
	\sup_{1\le j\le m}\left\|\frac{S_j^{*}+\Delta_j}{m}-\mathbb{E}\left(\frac{S_j^*}{m}\right)\right\|\le c\left(\lambda_1+\lambda_2+\lambda_2^2\right),
	\end{align*}
	which leads to that $\sup_{1\le j\le m}\|\{(S_j^*+\Delta_j)/m\}^{-1}\|=\sup_{1\le j\le m}[\sigma_{\min}\{(S_j^*+\Delta_j)/m\}]^{-1}<c$ as $m\to\infty$ by Lemma \ref{lem_singular_value_contin}. And by (\ref{ineq4}) and (\ref{ineq2}), for $j=1,...,m$, we have $$\left\|U_j^*+\Pi_j\right\|\le c(\|x_j\|+\lambda_2\|x_j\|^2+\lambda_2^2\|x_j\|^3).$$
	
	Applying the result from Proposition 2, we have
	\begin{align*}
	&\sup_{1\le j\le m}\mathbb{P}\left\{\left\|\hat\beta(p_{j\rightarrow 0})-\hat\beta(p_{j\rightarrow 1})\right\|> K/m\right\}\\
	=&\sup_{1\le j\le m}\mathbb{P}\left\{\left\|(S_j^*+\Delta_j)^{-1}(U_j^*+\Pi_j)\right\|>K/m\right\}\\
	\le&\sup_{1\le j\le m}\Bigg[\mathbb{P}\Bigg\{\left\|(S_j^*+\Delta_j)^{-1}(U_j^*+\Pi_j)\right\|>K/m, \left\|\frac{S_j^*}{m}-\mathbb{E}\left(\frac{S_j^*}{m}\right)\right\|\le\lambda_1,\|\hat\beta(p_{j\to a})-\beta^*\|\le\lambda_2,\\
	&\qquad\qquad\qquad\left|\frac{\sum_{i\ne j}\|x_i\|^3}{m}-\mathbb{E}\left(\frac{\sum_{i\ne j}\|x_i\|^3}{m}\right)\right|\le\lambda_3,\left|\frac{\sum_{i\ne j}\|x_i\|^4}{m}-\mathbb{E}\left(\frac{\sum_{i\ne j}\|x_i\|^4}{m}\right)\right|\le\lambda_4\Bigg\}\\
	&\qquad\quad\;+\mathbb{P}\left\{\left\|\frac{S_j^*}{m}-\mathbb{E}\left(\frac{S_j^*}{m}\right)\right\|>\lambda_1\right\}+\mathbb{P}\left\{\|\hat\beta(p_{j\to a})-\beta^*\|>\lambda_2\right\}\\
	&\qquad\quad\;+\mathbb{P}\left\{\left|\frac{\sum_{i\ne j}\|x_i\|^3}{m}-\mathbb{E}\left(\frac{\sum_{i\ne j}\|x_i\|^3}{m}\right)\right|>\lambda_3\right\}+\mathbb{P}\left\{\left|\frac{\sum_{i\ne j}\|x_i\|^4}{m}-\mathbb{E}\left(\frac{\sum_{i\ne j}\|x_i\|^4}{m}\right)\right|>\lambda_4\right\}\Bigg]\\
	\le& \sup_{1\le j\le m}\Bigg[\mathbb{P}\left(\|x_j\|+\lambda_2\|x_j\|^2+\lambda_2^2\|x_j\|^3>cK\right)+\mathbb{P}\left\{\left\|\frac{S_j^*}{m}-\mathbb{E}\left(\frac{S_j^*}{m}\right)\right\|>\lambda_1\right\}+\mathbb{P}\left\{\|\hat\beta(p_{j\to a})-\beta^*\|>\lambda_2\right\}\\
	&\qquad\quad\;+\mathbb{P}\left\{\left|\frac{\sum_{i\ne j}\|x_i\|^3}{m}-\mathbb{E}\left(\frac{\sum_{i\ne j}\|x_i\|^3}{m}\right)\right|>\lambda_3\right\}+\mathbb{P}\left\{\left|\frac{\sum_{i\ne j}\|x_i\|^4}{m}-\mathbb{E}\left(\frac{\sum_{i\ne j}\|x_i\|^4}{m}\right)\right|>\lambda_4\right\}\Bigg].
	\end{align*}
	Thus we conclude that
	\begin{align*}
	&\sup_{1\le j\le m}	\mathbb{E}\left\{\left\|\hat\beta(p_{j\rightarrow 0})-\hat\beta(p_{j\rightarrow 1})\right\|^4\right\}\\
	\le&\sup_{1\le j\le m}\Bigg[c\left(Km^{-1}\right)^4+P\left(\|x_j\|+\lambda_2\|x_j\|^2+\lambda_2^2\|x_j\|^3>cK\right)\\
	&\qquad\quad\;\;+\mathbb{P}\left\{\left\|\frac{S_j^*}{m}-\mathbb{E}\left(\frac{S_j^*}{m}\right)\right\|>\lambda_1\right\}+\mathbb{P}\left\{\|\hat\beta(p_{j\to a})-\beta^*\|>\lambda_2\right\}\\
	&\qquad\quad\;\;+\mathbb{P}\left\{\left|\frac{\sum_{i\ne j}\|x_i\|^3}{m}-\mathbb{E}\left(\frac{\sum_{i\ne j}\|x_i\|^3}{m}\right)\right|>\lambda_3\right\}+\mathbb{P}\left\{\left|\frac{\sum_{i\ne j}\|x_i\|^4}{m}-\mathbb{E}\left(\frac{\sum_{i\ne j}\|x_i\|^4}{m}\right)\right|>\lambda_4\right\}\Bigg].
	\end{align*}
	As seen from above, there are six terms to be considered. From the result of Proposition \ref{pro_tail_beta}, we know that $\mathbb{P}\{\sup_j\|\hat\beta(p_{j\to a})-\beta^*\|> \lambda_2\}\le\exp[-c\{\lambda_2^4+o(\lambda_2^4)\}m],$ which implies that the convergence rate of $\mathbb{P}\{\|\hat\beta(p_{j\rightarrow 0})-\hat\beta(p_{j\rightarrow 1})\|> K/m\}$ is not determined by the term $\mathbb{P}\{\|\hat\beta(p_{j\to a})-\beta^*\|>\lambda_2\}$. We thus focus on the other five terms. Let $K=m^{\kappa}$ with $0<\kappa<1$.
	Applying the result of Proposition \ref{pro_tail_s}, we have
	\begin{align*}
	&\sup_{1\le j\le m}\Bigg[\mathbb{P}\left\{\left\|\frac{S_j^*}{m}-\mathbb{E}\left(\frac{S_j^*}{m}\right)\right\|>\lambda_1\right\}+\mathbb{P}\left\{\left|\frac{\sum_{i\ne j}\|x_i\|^3}{m}-\mathbb{E}\left(\frac{\sum_{i\ne j}\|x_i\|^3}{m}\right)\right|>\lambda_3\right\}\\
	&\qquad\quad\;\;+\mathbb{P}\left\{\left|\frac{\sum_{i\ne j}\|x_i\|^4}{m}-\mathbb{E}\left(\frac{\sum_{i\ne j}\|x_i\|^4}{m}\right)\right|>\lambda_4\right\}\Bigg]=o(m^{1-q}).
	\end{align*}
	Note that
	\begin{align}
	&\sup_{1\le j\le m}\left\{c\left(Km^{-1}\right)^4+\mathbb{P}\left(\|x_j\|+\lambda_2\|x_j\|^2+\lambda_2^2\|x_j\|^3>cK\right)
	+o(m^{1-q})\right\} \nonumber \\
	\le&\sup_{1\le j\le m}\bigg\{c\left(Km^{-1}\right)^4+c\mathbb{E}\left(\|x_j\|^{4q}\right)K^{-4q}+c\mathbb{E}\left(\|x_j\|^{4q}\right)(\lambda_2K^{-1})^{2q} \nonumber \\
	&\qquad\quad\;\;+c\mathbb{E}\left(\|x_j\|^{4q}\right)\left(\lambda_2^2K^{-1}\right)^{\frac{4q}{3}}+o(m^{1-q})\bigg\} \nonumber\\
	\le&c\left\{\left(Km^{-1}\right)^4+K^{-4q}+\left(\lambda_2K^{-1}\right)^{2q}+\left(\lambda_2^2K^{-1}\right)^{\frac{4q}{3}}\right\}+o(m^{1-q}) \nonumber\\
	=&O\left(m^{4(\kappa-1)}+m^{-4\kappa q}+m^{-2\omega q-2\kappa q}+m^{-\frac{8}{3}\omega q-\frac{4}{3}\kappa q}\right)+o(m^{1-q}) \label{eq-order},
	\end{align}
	where $0 <\kappa <1$ and $0<\omega<1/4$. Let
	\begin{align*}
	\eta(\kappa,\omega,q)=\max\left\{v_1=4(\kappa-1),v_2=-4\kappa q,v_3=-2\omega q-2\kappa q,v_4=-\frac{8}{3}\omega q-\frac{4}{3}\kappa q\right\}.
	\end{align*}
	Finding the order of (\ref{eq-order}) is equivalent to solving the following problem
	\begin{equation}\label{eq-lp-g}
	\min\limits_{\substack{0<\kappa<1,0<\omega< 1/4}}\eta(\kappa,\omega,q),\quad q\ge 2.
	\end{equation}
	We notice that the three lines $v_2, v_3, v_4$ intersect at point $(\omega, -4\omega q)$ for any $\omega$ and $q$, and lines $v_1, v_2$ intersect at point $\left(1/(1+q),-4q/(1+q)\right)$. Observe that $\eta(\kappa,\omega_2,q)\le\eta(\kappa,\omega_1,q)$ if $\omega_2\ge\omega_1$ for any $\kappa$ and $q$. Thus we let $\omega=1/4-\epsilon$, where $\epsilon$ is an arbitrarily small positive number. If $2\le q\le 3$, then $1/(1+q)>\omega$ and hence the solution to (\ref{eq-lp-g}) is obtained when $v_1=v_4$. If $q>3$, then $1/(1+q)<\omega$ as $\epsilon$ can be arbitrarily small, and hence the solution is obtained when $v_1=v_2$. The idea is illustrated in Figure \ref{linear_prog}, where we set $\omega=1/4-0.001$.
	\begin{figure}
		\centering
		\includegraphics[scale=0.3]{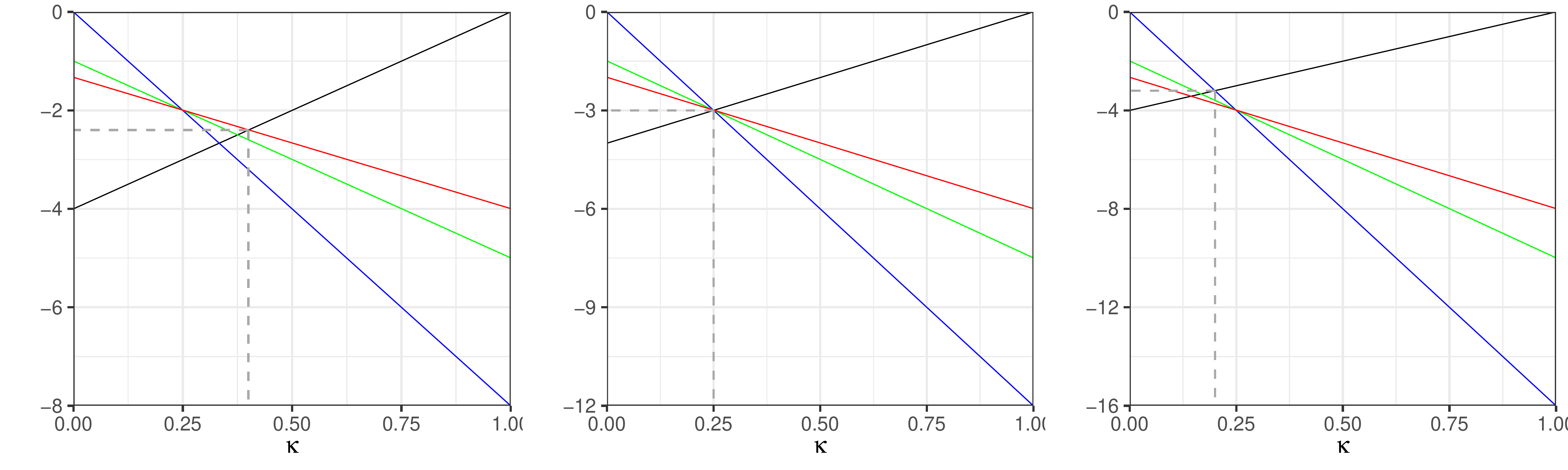}
		\caption{Illustration of the linear programming problem (\ref{eq-lp-g}) with $\omega=1/4-0.001$. The black, blue, green and red lines represent $v_1$, $v_2$, $v_3$ and $v_4$ in (\ref{eq-lp-g}) respectively. From left to right, the three panels correspond to $q=2,3,4$ respectively.}
		\label{linear_prog}
	\end{figure}
	Therefore, for $2\le q\le 3$, the solution to (\ref{eq-lp-g}) is obtained when $4(\kappa-1)=-8\omega q/3-4\kappa q/3$. In this case, $\kappa=(3-2\omega q)/(3+q)$ and $4(\kappa-1)=-4(2\omega+1)q/(3+q)$. Thus we have
	\begin{align*}
	\sup_{1\le j\le m}	\mathbb{E}\left\{\left\|\hat\beta(p_{j\rightarrow 0})-\hat\beta(p_{j\rightarrow 1})\right\|^4\right\}\le\max\left\{O\left(m^{\frac{-4(2\omega+1)q}{3+q}}\right),o\left(m^{1-q}\right)\right\}=o(m^{1-q}).
	\end{align*}
	For $q>3$, the solution is obtained when $4(\kappa-1)=-4\kappa q$. In this case, $\kappa=1/(1+q)$ and $4(\kappa-1)=-4q/(1+q)$.
	Hence we get
	\begin{align*}
	\sup_{1\le j\le m}	\mathbb{E}\left\{\left\|\hat\beta(p_{j\rightarrow 0})-\hat\beta(p_{j\rightarrow 1})\right\|^4\right\}\le\max\left\{O\left(m^{\frac{-4q}{1+q}}\right), o\left(m^{1-q}\right)\right\}=\begin{cases}
	o(m^{1-q}),&\text{ if }3<q\le 2+\sqrt{5},\\
	O\left(m^{\frac{-4q}{1+q}}\right),&\text{ if }q> 2+\sqrt{5}.\\
	\end{cases}
	\end{align*}
	
	Summarizing the above results, we have
	\begin{align*}
	J_{m,1}\le&\sum_{j=1}^m\left\{\mathbb{E}\left(\|x_j\|^4\right)\right\}^{\frac{1}{4}}\left[\mathbb{E}\left\{\left\|\hat\beta(p_{j\rightarrow 0})-\hat\beta(p_{j\rightarrow 1})\right\|^4\right\}\right]^{\frac{1}{4}}\left\{\mathbb{E}\left(\left[\left\{\alpha^{-1}\sum_{i\neq j}\mathbb{I}(p_i>\gamma)\right\}\vee\varepsilon^{1/(1-k)}\right]^{-2}\right)\right\}^{\frac{1}{2}}\\
	\le&\sum_{j=1}^m\left\{\mathbb{E}\left(\|x_j\|^4\right)\right\}^{\frac{1}{4}}	\left[\mathbb{E}\left\{\left\|\hat\beta(p_{j\rightarrow 0})-\hat\beta(p_{j\rightarrow 1})\right\|^4\right\}\right]^{\frac{1}{4}}O(\alpha m^{-1})\\
	\le&\left[\sup_{1\le j\le m}	\mathbb{E}\left\{\left\|\hat\beta(p_{j\rightarrow 0})-\hat\beta(p_{j\rightarrow 1})\right\|^4\right\}\right]^{\frac{1}{4}}O(\alpha),
	\end{align*}
	and
	\begin{align*}
	J_{m,2}		\le&\alpha^{-1}\sum_{j=1}^mm\left\{\frac{1}{m}\sum_{i=1}^m\mathbb{E}\left(\|x_i\|^4\right)\right\}^{1/4}\left[\mathbb{E}\left\{\left\|\hat\beta(p_{j\rightarrow 0})-\hat\beta(p_{j\rightarrow 1})\right\|^4\right\}\right]^{1/4}\\
	&\qquad\quad\;\;\;\left\{\mathbb{E}\left(\left[\left\{c\alpha^{-1}\sum_{i\neq j}\mathbb{I}(p_i>\gamma)\right\}\vee\varepsilon^{1/(1-k)}\right]^{-4}\right)\right\}^{1/2}\\
	&+\alpha^{-1}\sum_{j=1}^m\mathbb{E}\left(\left[\left\{c\alpha^{-1}\sum_{i\neq j}\mathbb{I}(p_i>\gamma)\right\}\vee\varepsilon^{1/(1-k)}\right]^{-2}\right)\\
	\le&\alpha^{-1}\sum_{j=1}^mm\left[\mathbb{E}\left\{\left\|\hat\beta(p_{j\rightarrow 0})-\hat\beta(p_{j\rightarrow 1})\right\|^4\right\}\right]^{\frac{1}{4}}O(\alpha^2 m^{-2})+\alpha^{-1}\sum_{j=1}^mO(\alpha^2m^{-2})\\
	\le&\left[\sup_{1\le j\le m}	\mathbb{E}\left\{\left\|\hat\beta(p_{j\rightarrow 0})-\hat\beta(p_{j\rightarrow 1})\right\|^4\right\}\right]^{\frac{1}{4}}O(\alpha)+O(\alpha m^{-1}),
	\end{align*}
	where we have used Condition (ii) to obtain that $\sup_{1\leq i\leq m}\mathbb{E}(\|x_i\|^{4})<\infty$. Finally, we obtain
	\begin{align*}
	\text{FWER}\le J_m+\alpha=\begin{cases}
	o(\alpha m^{\frac{1-q}{4}})+\alpha, &\text{if }2\le q\le 2+\sqrt{5} ,\\
	O(\alpha m^{\frac{-q}{1+q}})+\alpha, &\text{if }q>2+\sqrt{5} .
	\end{cases}
	\end{align*}
	~
\end{proof}

\section{Additional simulation results}\label{sec-simu}
Figures \ref{fig_no_signal}--\ref{fig_moderate_power} show the FWER control across different target levels. Figures \ref{fig_s1}--\ref{fig_s2_4} present the results of Setups S1--S2.

\begin{figure}
	\centering
	\includegraphics[scale=0.35]{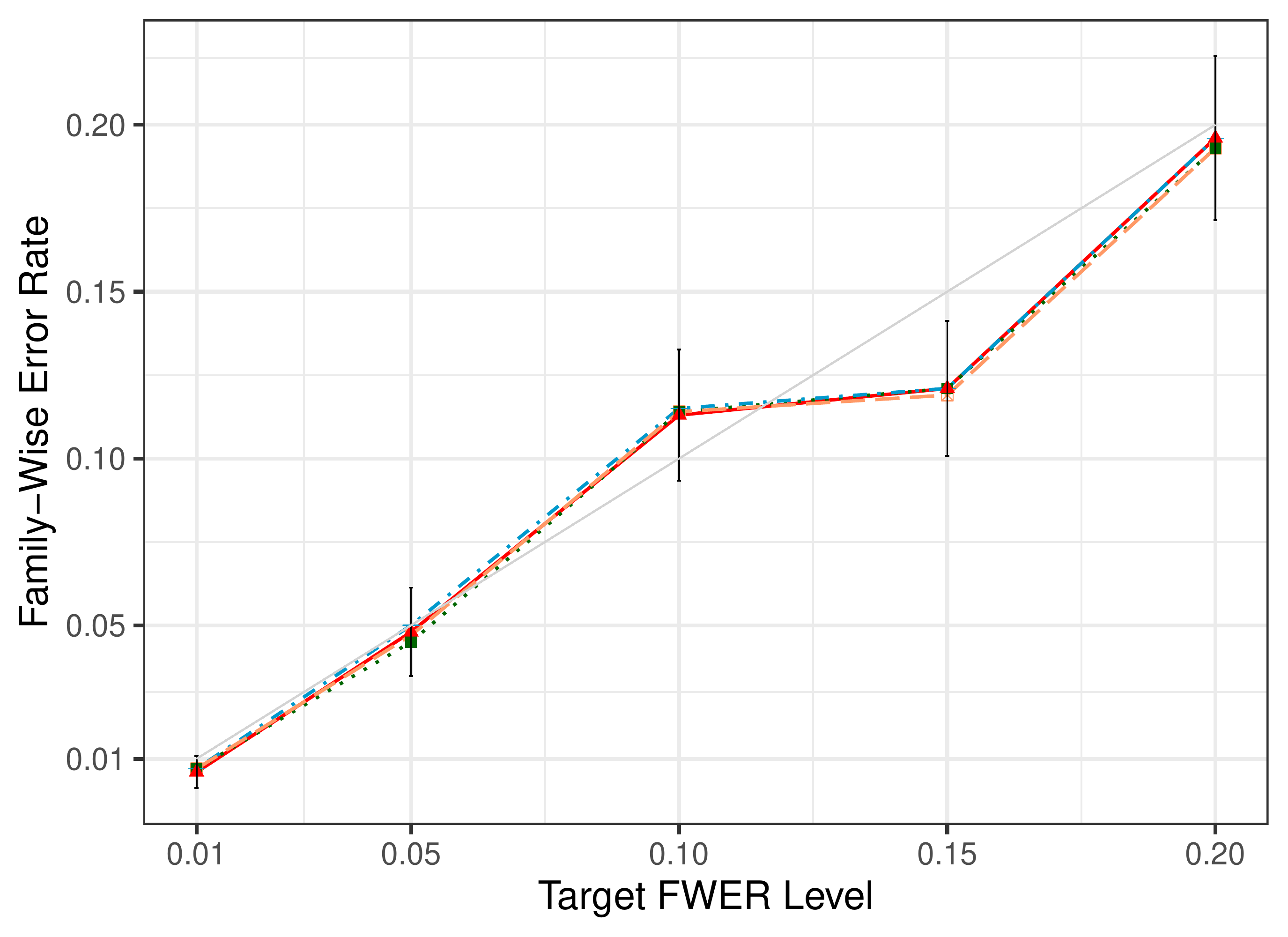}
	\caption{FWER control at various target levels (0.01 - 0.20) under the complete null (no signal was simulated). Family-wise error rates were averaged over 1,000 simulation runs. The solid red, dotted green, dot-dashed blue and long-dashed orange lines represent CAMT.fwer, IHW- Bonferroni, weighted Bonferroni and Holm’s step-down methods respectively. The gray diagonal line represents the target FWER levels and the error bars represent the 95\% CIs of the method CAMT.fwer.}
	\label{fig_no_signal}
\end{figure}
\begin{figure}
	\centering
	\includegraphics[scale=0.35]{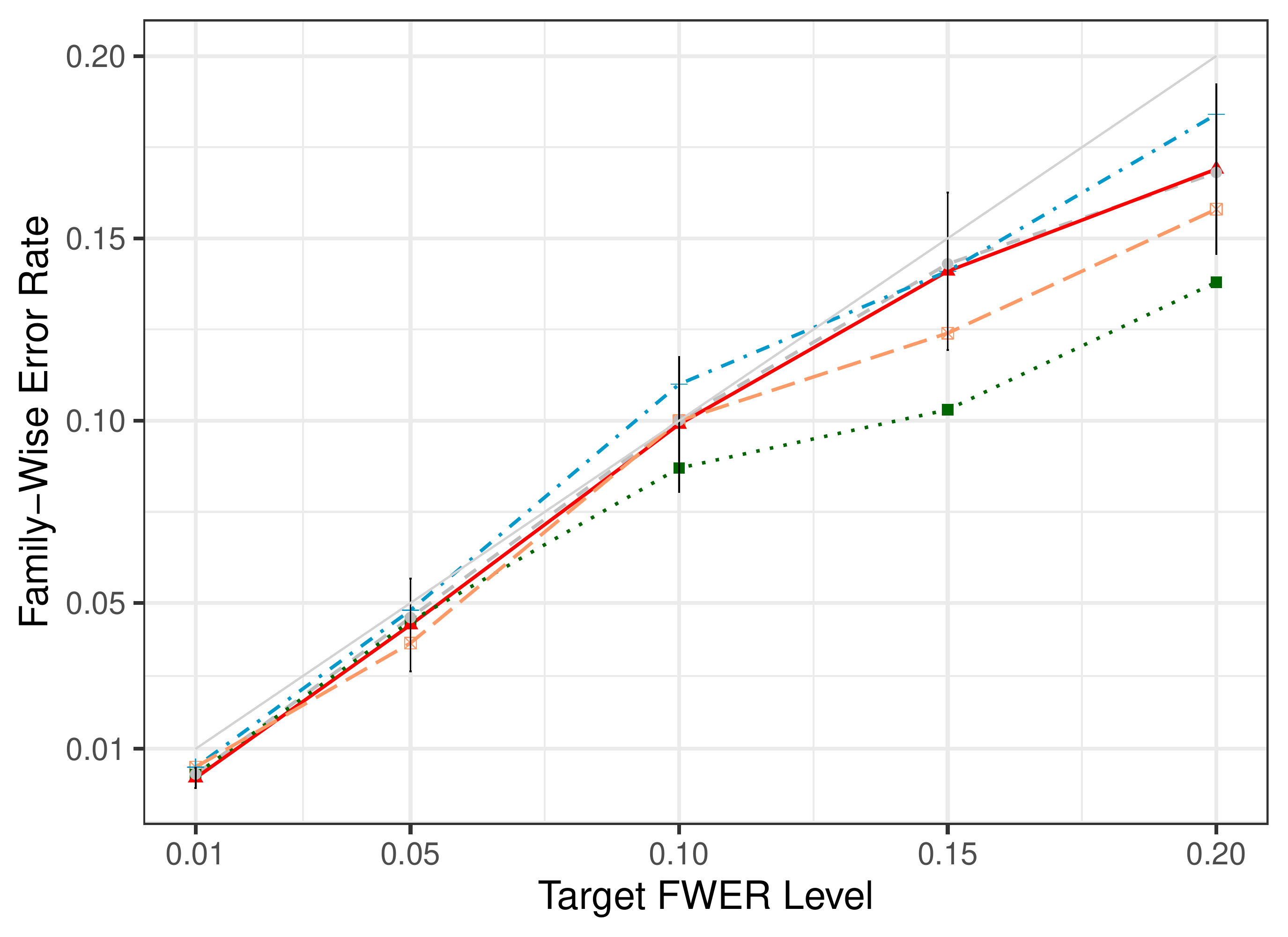}
	\caption{FWER control at various target levels (0.01 - 0.20) under Setup S0 with moderate signal density, signal strength and covariate informativeness. Family-wise error rates were averaged over 1,000 simulation runs. The solid red, dotted green, dot-dashed blue and long-dashed orange lines represent CAMT.fwer, IHW- Bonferroni, weighted Bonferroni and Holm’s step-down methods respectively. The gray diagonal line represents the target FWER levels and the error bars represent the 95\% CIs of the method CAMT.fwer.}
	\label{fig_moderate_fwer}
\end{figure}
\begin{figure}
	\centering
	\includegraphics[scale=0.35]{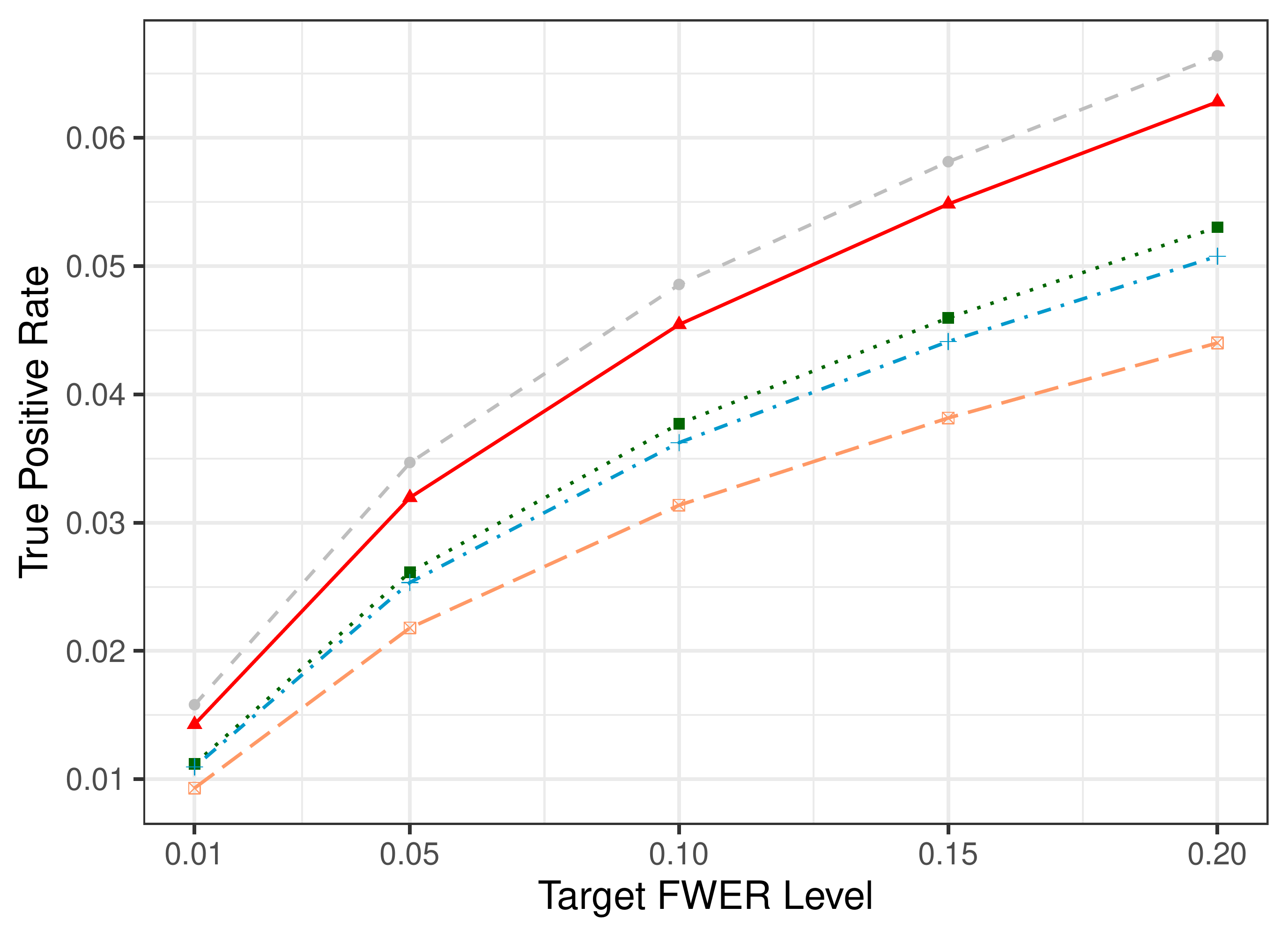}
	\caption{FWER control at various target levels (0.01 - 0.20) under Setup S0 with moderate signal density, signal strength and covariate informativeness. True positive rates were averaged over 1,000 simulation runs. The dashed gray, solid red, dotted green, dot-dashed blue and long-dashed orange lines represent the oracle, CAMT.fwer, IHW- Bonferroni, weighted Bonferroni and Holm’s step-down methods respectively. }
	\label{fig_moderate_power}
\end{figure}
\begin{figure}
	\begin{subfigure}[b]{1\textwidth}
		\centering
		\includegraphics[scale=0.45]{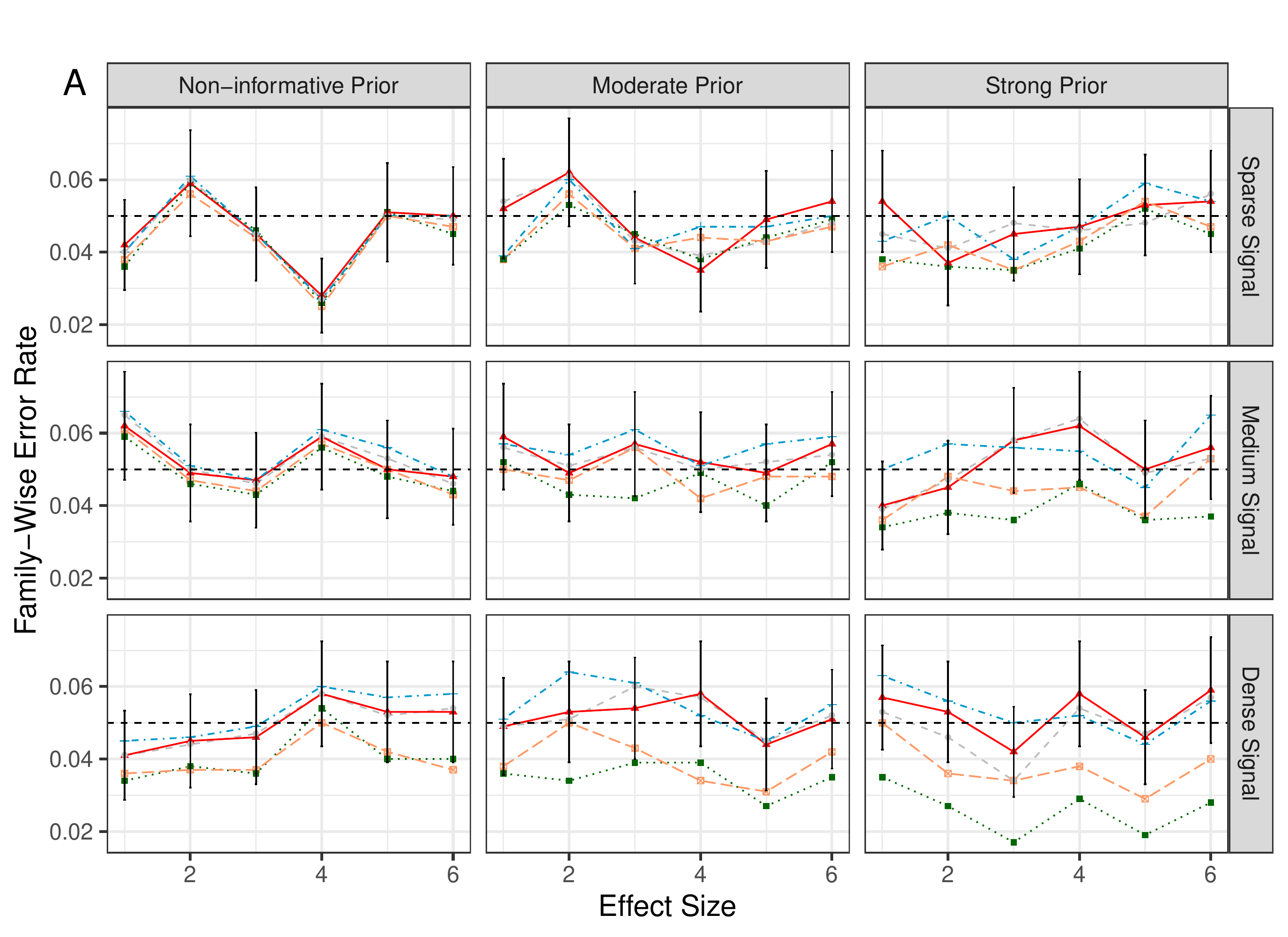}
	\end{subfigure}
	\begin{subfigure}[b]{1\textwidth}
		\centering
		\includegraphics[scale=0.45]{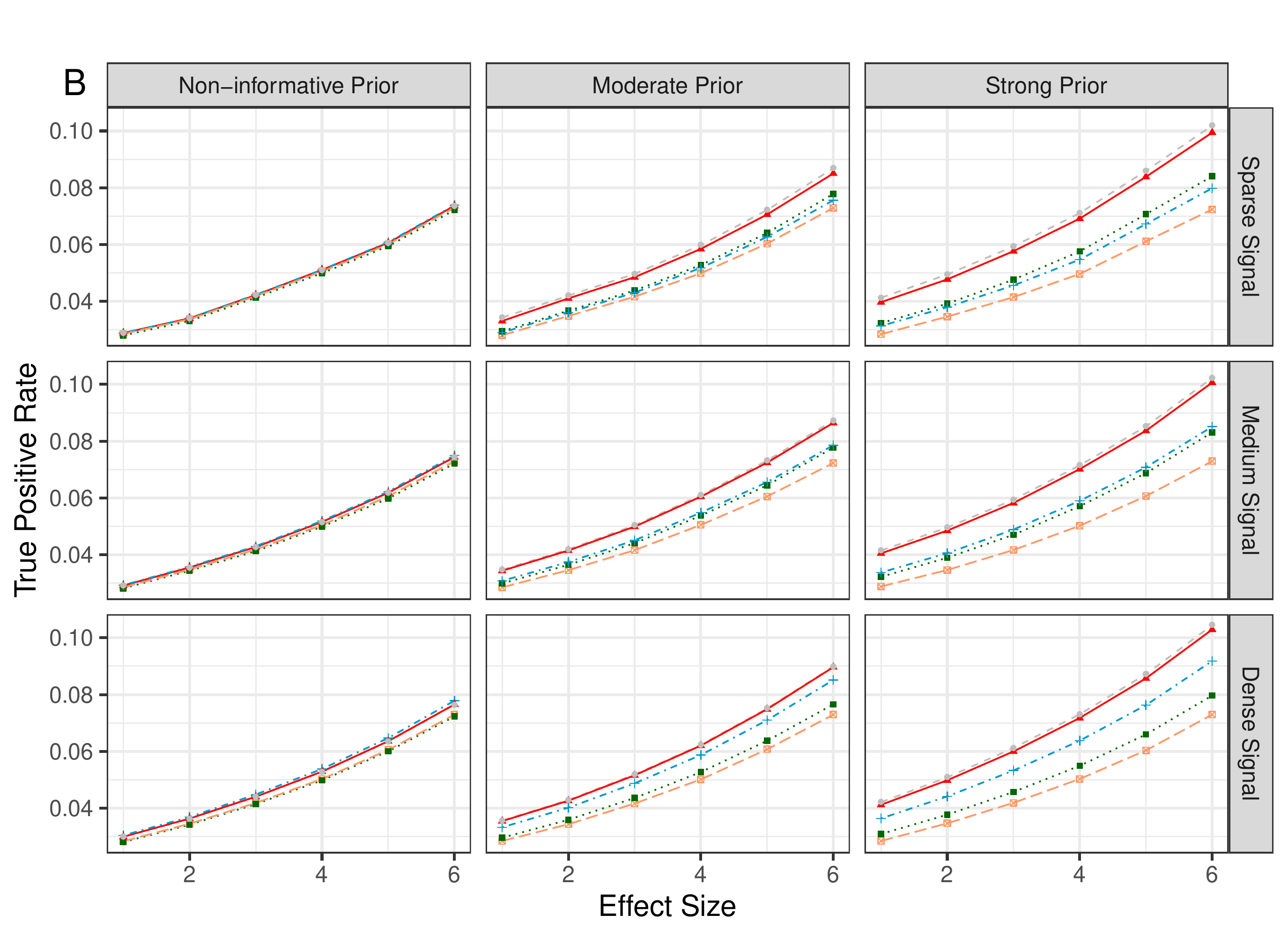}
	\end{subfigure}
	\caption{Performance comparison under S1.  Family-wise error rates (A) and true positive rates (B) were averaged over 1000 simulation runs. The dashed gray, solid red, dotted green, dot-dashed blue and long-dashed orange lines represent the oracle, CAMT.fwer, IHW- Bonferroni, weighted Bonferroni and Holm’s step-down methods respectively. The error bars (A) represent the 95\% CIs of the method CAMT.fwer and the dashed horizontal line indicates the target FWER level of 0.05.}
	\label{fig_s1}
\end{figure}
\begin{figure}
	\begin{subfigure}[b]{1\textwidth}
		\centering
		\includegraphics[scale=0.45]{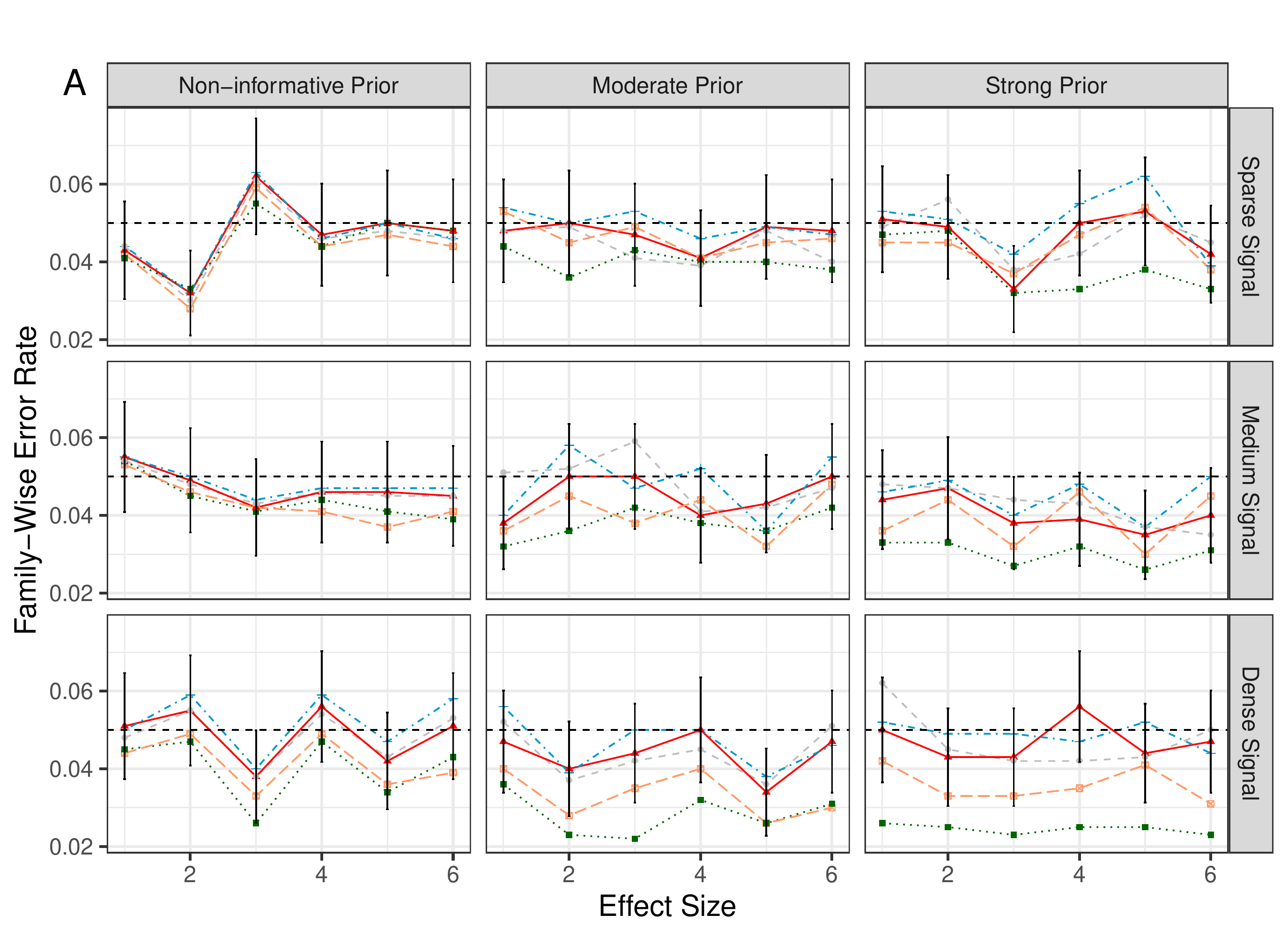}
	\end{subfigure}
	\begin{subfigure}[b]{1\textwidth}
		\centering
		\includegraphics[scale=0.45]{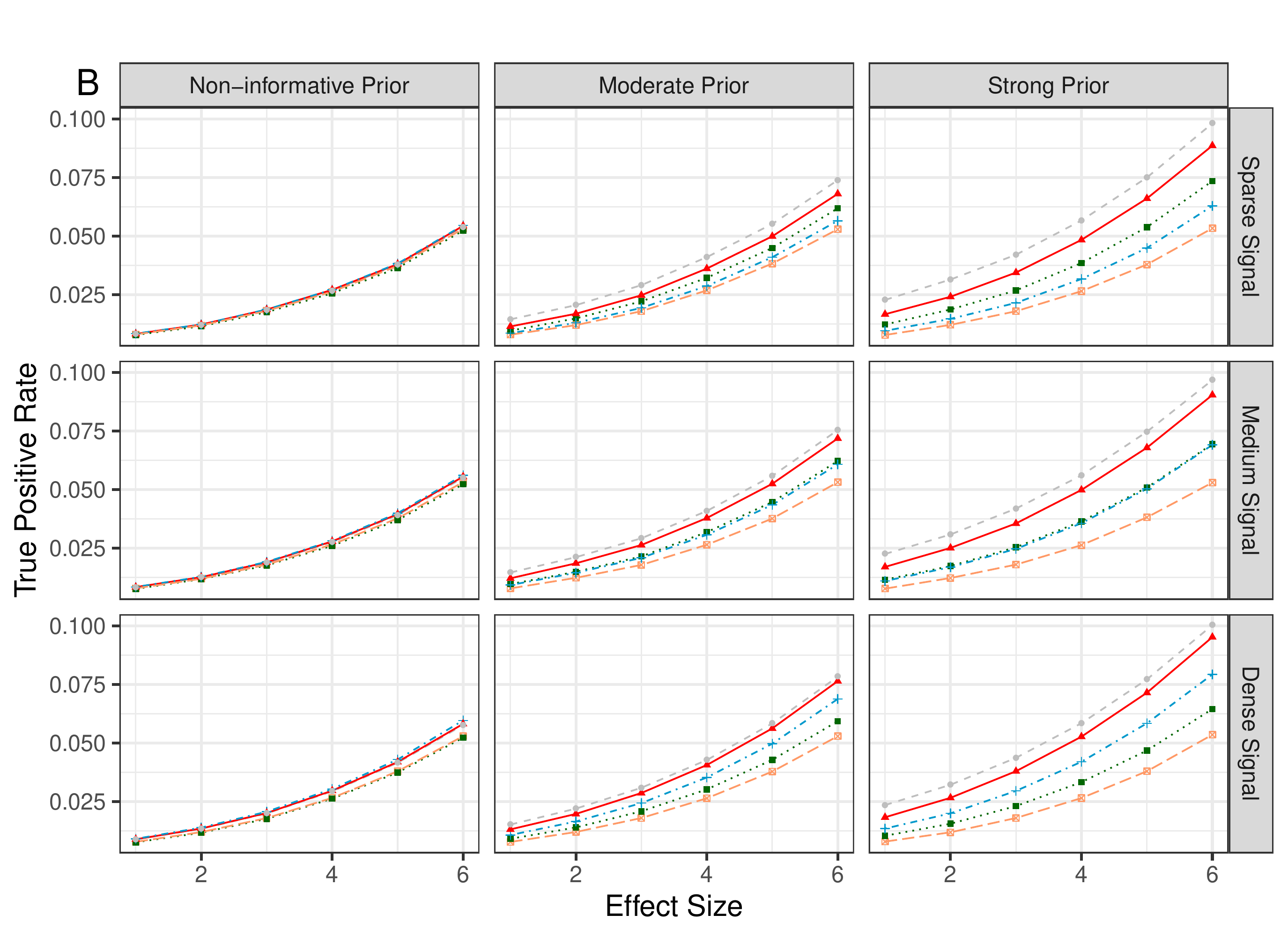}
	\end{subfigure}
	\caption{Performance comparison under S2.1.  Family-wise error rates (A) and true positive rates (B) were averaged over 1000 simulation runs. The dashed gray, solid red, dotted green, dot-dashed blue and long-dashed orange lines represent the oracle, CAMT.fwer, IHW- Bonferroni, weighted Bonferroni and Holm’s step-down methods respectively. The error bars (A) represent the 95\% CIs of the method CAMT.fwer and the dashed horizontal line indicates the target FWER level of 0.05.}
	\label{fig_s2_1}
\end{figure}
\begin{figure}
	\begin{subfigure}[b]{1\textwidth}
		\centering
		\includegraphics[scale=0.45]{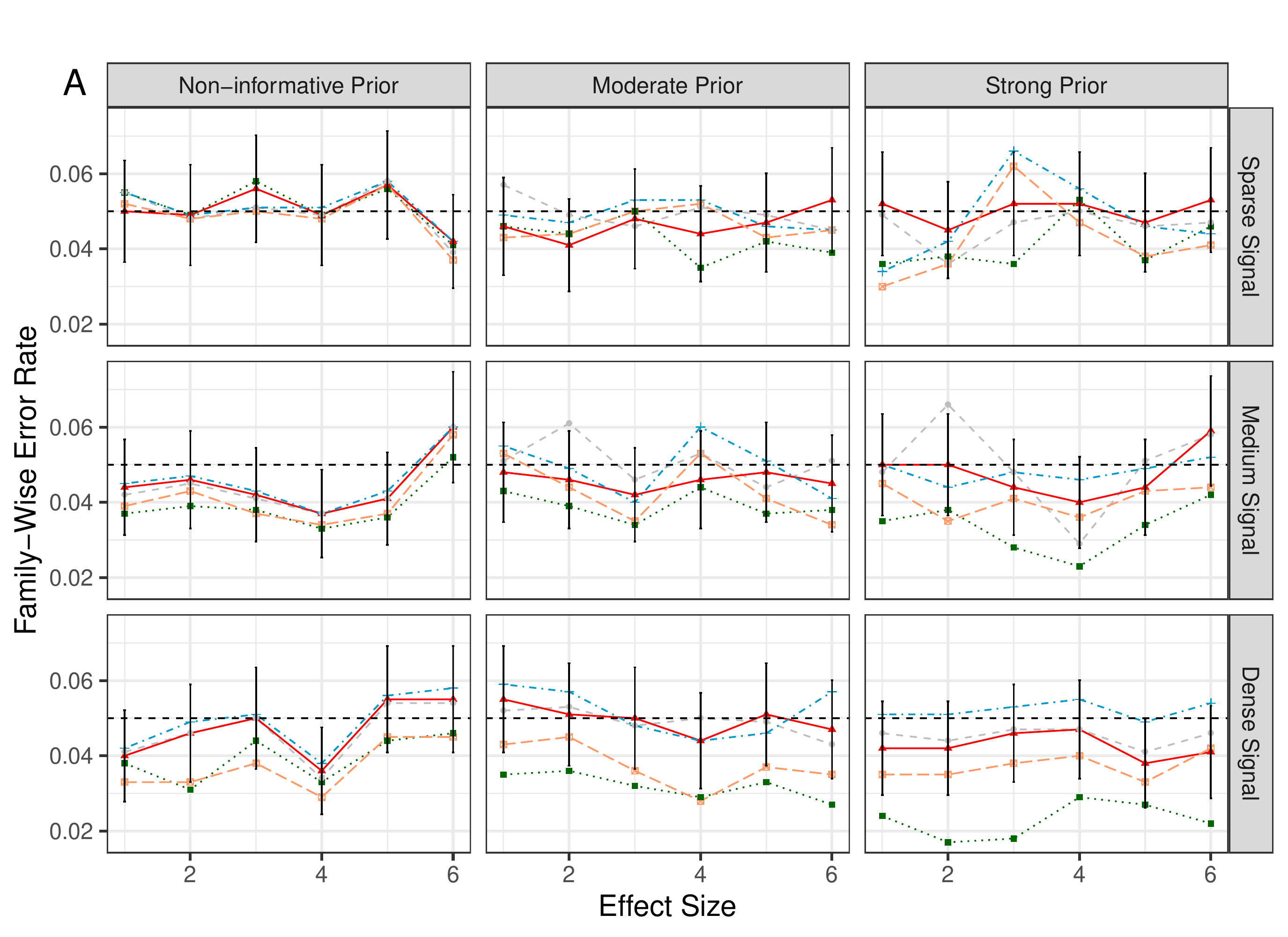}
	\end{subfigure}
	\begin{subfigure}[b]{1\textwidth}
		\centering
		\includegraphics[scale=0.45]{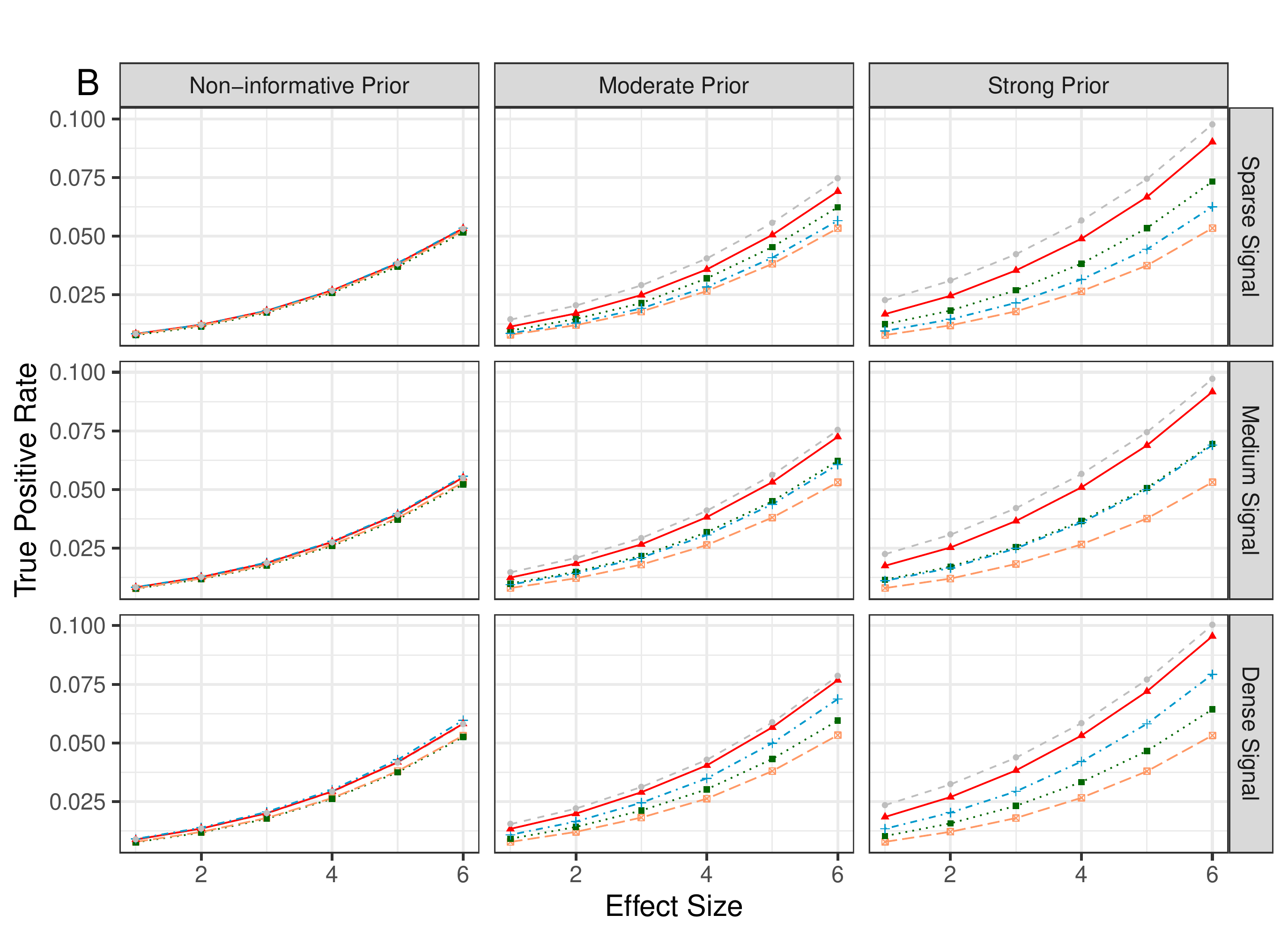}
	\end{subfigure}
	\caption{Performance comparison under S2.2.  Family-wise error rates (A) and true positive rates (B) were averaged over 1000 simulation runs. The dashed gray, solid red, dotted green, dot-dashed blue and long-dashed orange lines represent the oracle, CAMT.fwer, IHW- Bonferroni, weighted Bonferroni and Holm’s step-down methods respectively. The error bars (A) represent the 95\% CIs of the method CAMT.fwer and the dashed horizontal line indicates the target FWER level of 0.05.}
	\label{fig_s2_2}
\end{figure}
\begin{figure}
	\begin{subfigure}[b]{1\textwidth}
		\centering
		\includegraphics[scale=0.45]{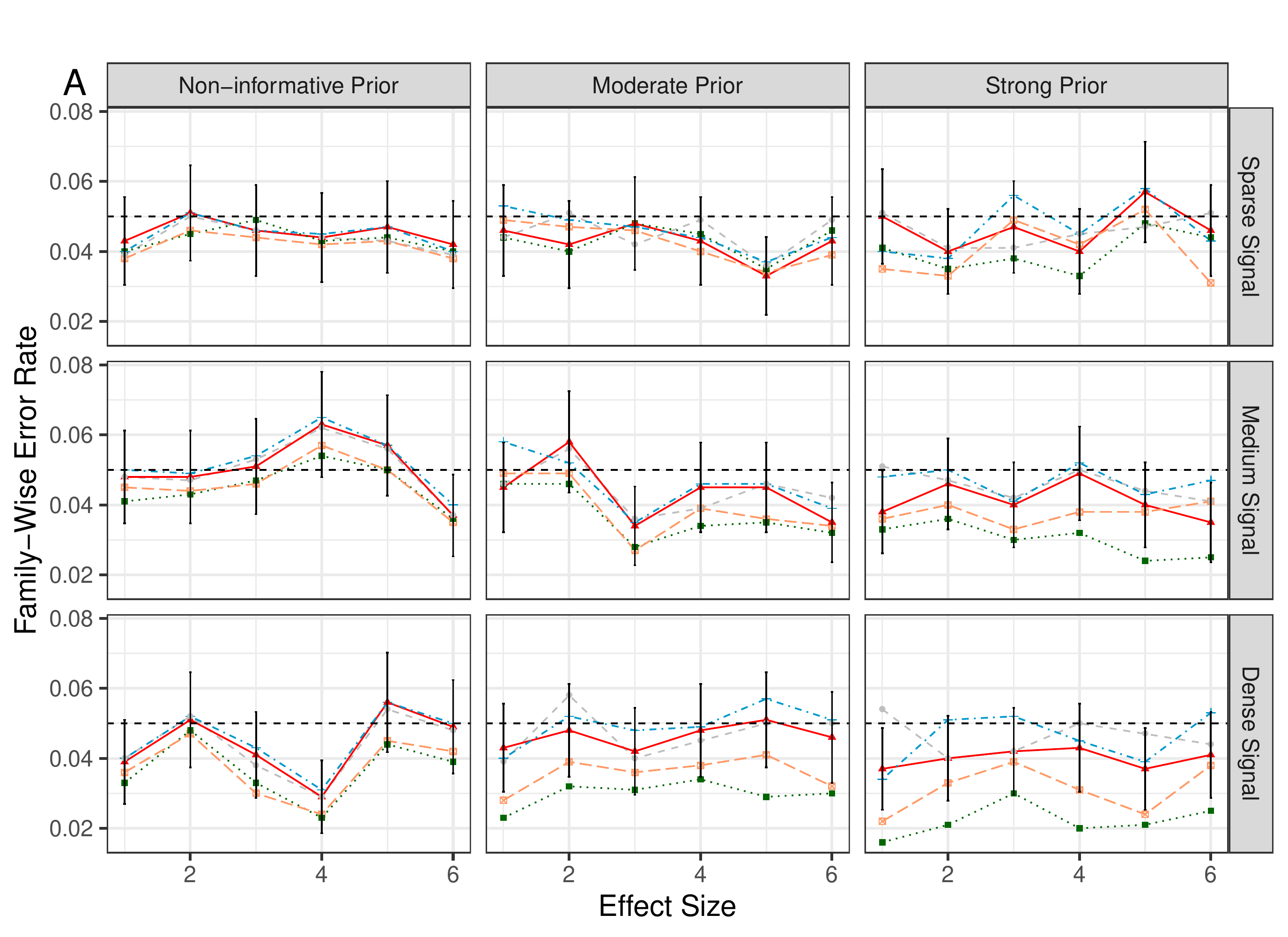}
	\end{subfigure}
	\begin{subfigure}[b]{1\textwidth}
		\centering
		\includegraphics[scale=0.45]{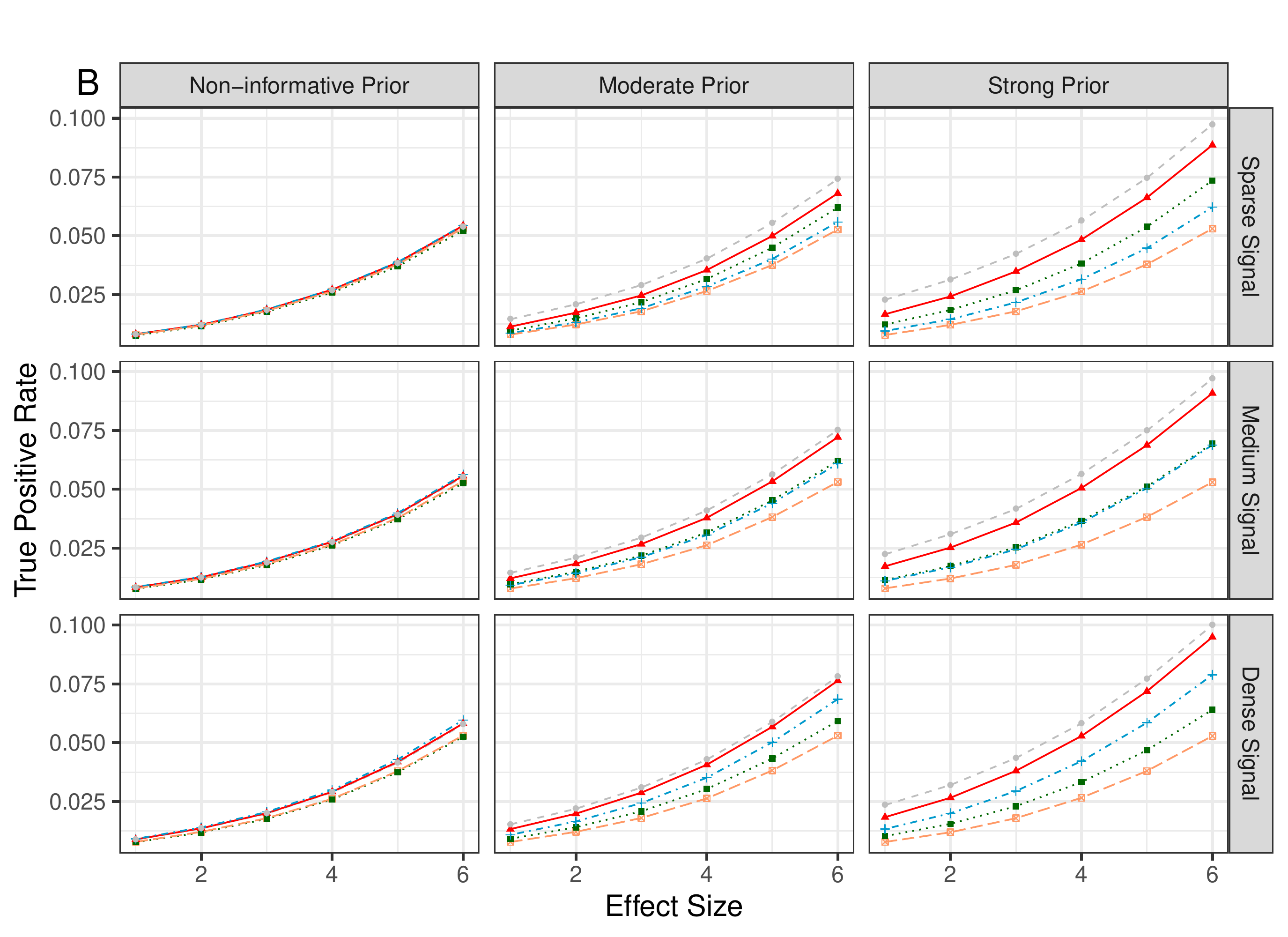}
	\end{subfigure}
	\caption{Performance comparison under S2.3.  Family-wise error rates (A) and true positive rates (B) were averaged over 1000 simulation runs. The dashed gray, solid red, dotted green, dot-dashed blue and long-dashed orange lines represent the oracle, CAMT.fwer, IHW- Bonferroni, weighted Bonferroni and Holm’s step-down methods respectively. The error bars (A) represent the 95\% CIs of the method CAMT.fwer and the dashed horizontal line indicates the target FWER level of 0.05.}
	\label{fig_s2_3}
\end{figure}
\begin{figure}
	\begin{subfigure}[b]{1\textwidth}
		\centering
		\includegraphics[scale=0.45]{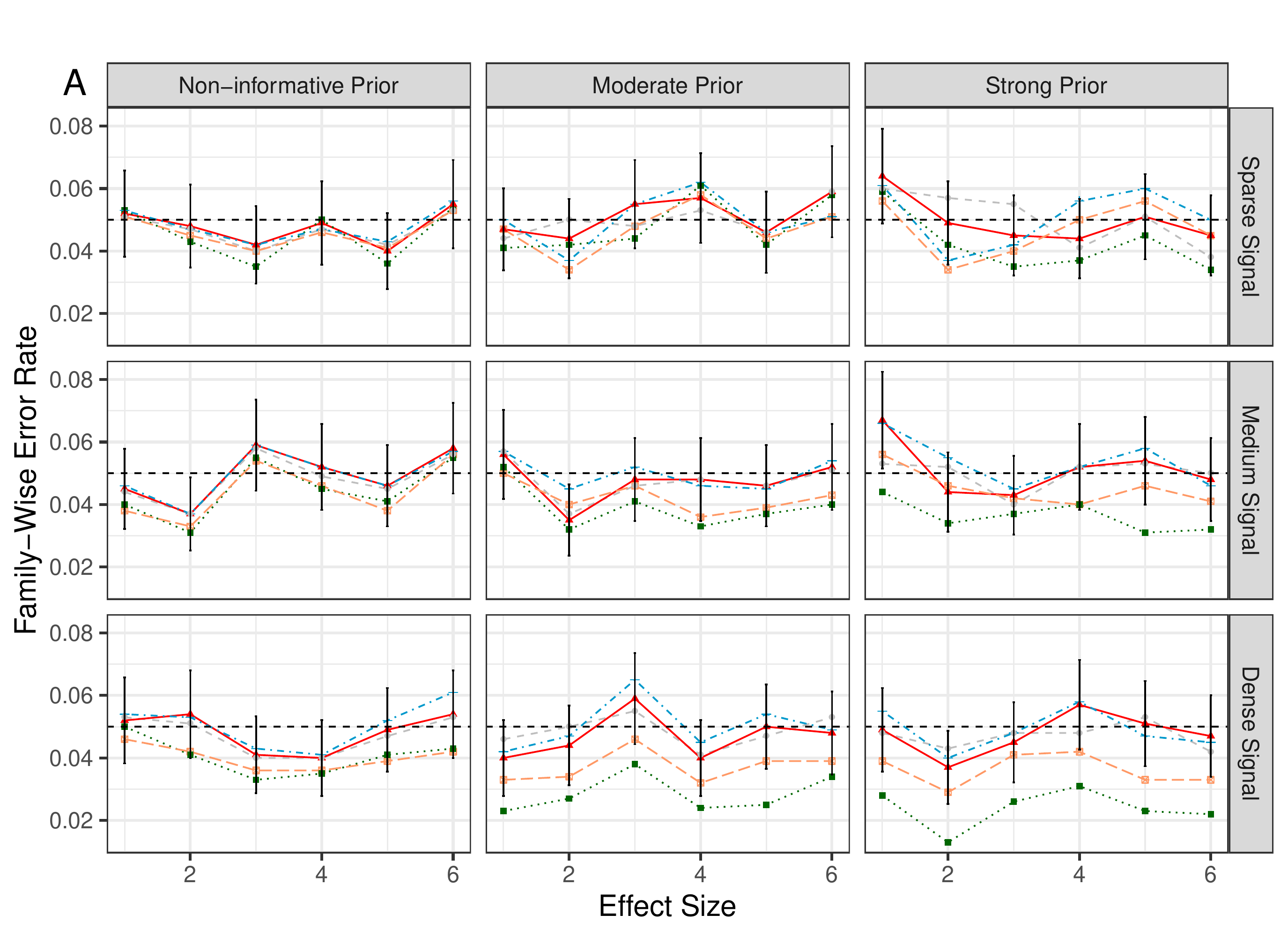}
	\end{subfigure}
	\begin{subfigure}[b]{1\textwidth}
		\centering
		\includegraphics[scale=0.45]{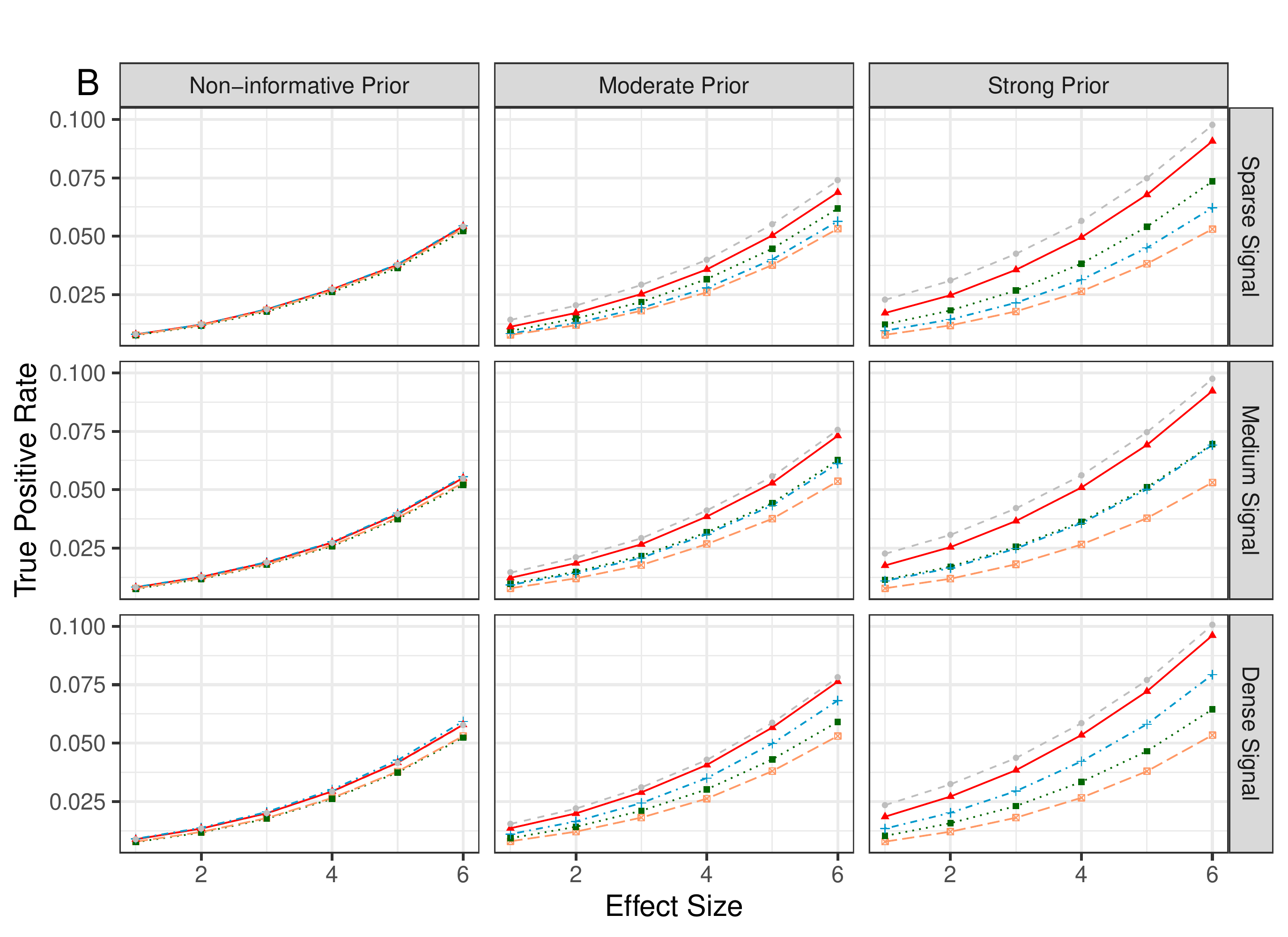}
	\end{subfigure}
	\caption{Performance comparison under S2.4. Family-wise error rates (A) and true positive rates (B) were averaged over 1000 simulation runs. The dashed gray, solid red, dotted green, dot-dashed blue and long-dashed orange lines represent the oracle, CAMT.fwer, IHW- Bonferroni, weighted Bonferroni and Holm’s step-down methods respectively. The error bars (A) represent the 95\% CIs of the method CAMT.fwer and the dashed horizontal line indicates the target FWER level of 0.05.}
	\label{fig_s2_4}
\end{figure}

\section{Addition result for Application to GWAS of UK Biobank data Section}\label{sec-gwas}
We present the numbers of rejections before clumping mentioned in Section 5 of the main paper.
\begin{table} \centering
	\caption{Significant SNPs detected at the FWER level of 0.05. Improve=$(\text{CAMT.fwer}-\text{Holm})/\text{Holm}\times 100\%$}
	\footnotesize
		\begin{tabular}{lccccc}
			\\[-1.8ex]\hline
			\hline \\[-1.8ex]
			Traits & Holm & IHW & weighted Bonferroni & CAMT.fwer& Improve \\
			\hline \\[-1.8ex]
			Balding Type I & $28,752$ & $28,752$ & $29,021$ & $30,441$ & $5.9\%$ \\
			BMI & $54,965$ & $54,965$ & $60,171$ & $64,057$ & $16.5\%$ \\
			Heel T Score & $108,112$ & $108,112$ & $113,962$ & $117,136$ & $8.3\%$ \\
			Height & $256,353$ & $256,051$ & $273,020$ & $278,034$ & $8.5\%$ \\
			Waist-hip Ratio & $35,984$ & $35,983$ & $37,720$ & $39,625$ & $10.1\%$ \\
			Eosinophil Count & $66,384$ & $66,384$ & $70,495$ & $71,623$ & $7.9\%$ \\
			Mean Corpular Hemoglobin & $92,048$ & $92,048$ & $95,790$ & $96,142$ & $4.4\%$ \\
			Red Blood Cell Count & $70,919$ & $70,919$ & $74,565$ & $78,061$ & $10.1\%$ \\
			Red Blood Cell Distribution Width & $69,583$ & $69,583$ & $73,162$ & $74,427$ & $7.0\%$ \\
			White Blood Cell Count & $55,881$ & $55,881$ & $62,483$ & $65,453$ & $17.1\%$ \\
			Auto Immune Traits & $7,571$ & $7,571$ & $7,774$ & $7,336$ & $-3.1\%$ \\
			Cardiovascular Diseases & $12,531$ & $12,531$ & $13,776$ & $14,859$ & $18.6\%$ \\
			Eczema & $13,099$ & $13,099$ & $13,513$ & $14,683$ & $12.1\%$ \\
			Hypothyroidism & $12,681$ & $12,681$ & $13,043$ & $14,651$ & $15.5\%$ \\
			Respiratory and Ear-nose-throat Diseases & $7,588$ & $7,588$ & $7,750$ & $8,709$ & $14.8\%$ \\
			Type 2 Diabetes & $2,459$ & $2,459$ & $2,524$ & $2,684$ & $9.2\%$ \\
			Age at Menarche & $25,549$ & $25,549$ & $26,519$ & $27,391$ & $7.2\%$ \\
			Age at Menopause & $6,109$ & $6,109$ & $6,211$ & $8,675$ & $42.0\%$ \\
			FEV1-FVC Ratio & $50,529$ & $50,529$ & $55,659$ & $58,503$ & $15.8\%$ \\
			Forced Vital Capacity (FVC) & $33,549$ & $33,549$ & $36,674$ & $38,985$ & $16.2\%$ \\
			Hair Color & $57,608$ & $57,608$ & $58,326$ & $60,391$ & $4.8\%$ \\
			Morning Person & $8,154$ & $8,154$ & $8,681$ & $9,559$ & $17.2\%$ \\
			Neuroticism & $5,513$ & $6,186$ & $6,073$ & $6,955$ & $26.2\%$ \\
			Smoking Status & $6,297$ & $6,623$ & $6,857$ & $8,016$ & $27.3\%$ \\
			Sunburn Occasion & $10,076$ & $10,076$ & $10,150$ & $11,000$ & $9.2\%$ \\
			Systolic Blood Pressure & $46,063$ & $46,063$ & $51,157$ & $54,749$ & $18.9\%$ \\
			Years of Education & $13,927$ & $13,927$ & $15,632$ & $16,933$ & $21.6\%$ \\
			\hline \\[-1.8ex]
	\end{tabular}
	\label{table-rej-2}
\end{table}

\end{document}